\documentclass[a4paper,twocolumn,11pt,accepted=2025-11-05]{quantumarticle}
\pdfoutput=1
\usepackage[utf8]{inputenc}
\usepackage[english]{babel}
\usepackage[T1]{fontenc}
\usepackage{amsmath, amsthm, physics}

\usepackage{cite}
\usepackage{tikz, array}
\usepackage{graphicx}
\usepackage{bbold}

\definecolor{quantumPink}{HTML}{6e1a9e} 
\usepackage{hyperref}
\hypersetup{
    colorlinks=true,
    linkcolor=quantumPink,
    citecolor=quantumPink,
    urlcolor=quantumPink
}

\usepackage[capitalise,nameinlink]{cleveref}

\usepackage{quantikz}
\usepackage{blkarray}
\usepackage{booktabs}
\usepackage{adjustbox}
\usepackage{url}

\tikzset{
  semicircle gate/.style={
    draw,
    shape=semicircle,
    shape border rotate=180,
    minimum height=1.5em,
    inner sep=2pt,
    anchor=center
  }
}

\theoremstyle{plain}
\newtheorem{proposition}{Proposition}
\newtheorem{lemma}{Lemma}
\newtheorem{theorem}{Theorem}
\newtheorem{claim}{Claim}
\newtheorem{remark}{Remark}
\newtheorem{definition}{Definition}
\usetikzlibrary{decorations.pathreplacing}
\tikzset{
mybrace/.style={decorate,decoration={brace,aspect=#1}}
}

\begin{document}

\title{The stellar decomposition of Gaussian quantum states}

\author{Arsalan Motamedi}
\affiliation{Xanadu}

\author{Yuan Yao}
\affiliation{Xanadu}

\author{Kasper Nielsen}
\affiliation{Xanadu}

\author{Ulysse Chabaud}
\affiliation{DIENS, \'Ecole Normale Sup\'erieure, PSL University, CNRS, INRIA, 45 rue d’Ulm, Paris, 75005, France}

\author{J. Eli Bourassa}
\affiliation{Xanadu}

\author{Rafael N. Alexander}
\affiliation{Xanadu}

\author{Filippo M. Miatto}
\email{filippo@xanadu.ai}
\affiliation{Xanadu}

\begin{abstract}
We introduce the \textit{stellar decomposition}, a novel method for characterizing non-Gaussian states produced by photon-counting measurements on Gaussian states. Given an $(m+n)$-mode Gaussian state $G$, we express it as an $(m+n)$-mode ``Gaussian core state'' $G_{\mathrm{core}}$ followed by an $m$-mode Gaussian transformation $T$ that only acts on the first $m$ modes. The defining property of the Gaussian core state $G_{\mathrm{core}}$ is that measuring the last $n$ of its modes in the photon-number basis leaves the first $m$ modes on a finite Fock support, i.e.~a core state. Since $T$ is measurement-independent and $G_{\mathrm{core}}$ has an exact and finite Fock representation, this decomposition exactly describes all non-Gaussian states obtainable by projecting $n$ modes of $G$ onto the Fock basis. For pure states we prove that a physical pair $(G_{\mathrm{core}}, T)$ always exists with $G_{\mathrm{core}}$ pure and $T$ unitary. For mixed states, we establish necessary and sufficient conditions for $(G_{\mathrm{core}}, T)$ to be a Gaussian mixed state and a Gaussian channel. We also develop a semidefinite program to extract the `largest' possible Gaussian channel when these conditions fail. Finally, we present a formal stellar decomposition for generic operators, which is useful in simulations where the only requirement is that the two parts contract back to the original operator. The stellar decomposition leads to practical bounds on achievable state quality in photonic circuits and for GKP state generation in particular. Our results are based on a new characterization of Gaussian completely positive maps in the Bargmann picture, which may be of independent interest.
\end{abstract}

\maketitle

\section{Introduction}

Quantum information processing promises a radical paradigm shift compared with its classical counterpart \cite{dense-coding,Shor,nielsen2010quantum}. This second quantum revolution \cite{dowling2003quantum} is based on several physical platforms, among which photonics stands out for its potential for deterministic entanglement generation, room-temperature operation, high clock speeds, and long coherence times \cite{yokoyama2013ultra,takeda2019toward, Arrazola2021, aghaee2025scaling,psiq2025}.

Quantum states of light---and by extension, quantum operations---are usually separated into two broad classes, Gaussian \cite{weedbrook2012gaussian} and non-Gaussian \cite{walschaers2021non}, which share mathematical similarities with the Clifford/non-Clifford dichotomy in qubit systems \cite{gottesman1999demonstrating}. Gaussian states and operations are generated by quadratic Hamiltonians in the bosonic canonical operators, and are often regarded as easier to handle from a theoretical point of view. However, they suffer from severe limitations for quantum information processing \cite{wenger2003maximal,garcia2004proposal,adesso2009optimal,nogo1,nogo2,nogo3,Niset2009}: in particular, Gaussian computations can be simulated efficiently by classical computers \cite{bartlett2002efficient}. This has led to the understanding of non-Gaussian states and operations as resources for many quantum information processing tasks \cite{Albarelli2018,quntao18,Takagi2018}, and to the introduction of associated resource measures \cite{walschaers2021non}.
One such measure, the stellar rank \cite{chabaud2020stellar,chabaud2022holomorphic}, provides a natural way of bookkeeping the ``degree of non-Gaussianity" of quantum states in situations where one does not have access to Hamiltonians of degree higher than quadratic, or couplings between bosonic operators and those of two-level systems such as quantum dots, nor measurements of observables that are beyond quadratic. Indeed, the stellar rank of a non-Gaussian quantum state is non-increasing under Gaussian operations and lower bounds the number of photonic non-Gaussian operations (photon-addition or photon-subtraction) necessary to engineer a non-Gaussian state from a Gaussian one \cite{chabaud2020stellar,chabaud2022holomorphic}.

From an experimental standpoint, Gaussian states and operations are also simpler to engineer in photonic platforms than non-Gaussian ones \cite{lvovsky2020production}:
currently, the standard technological toolkit in quantum photonics consists of the ability to generate deterministic Gaussian states, to implement Gaussian transformations deterministically, and to perform photon-number-resolving (PNR) measurements.
Using this toolkit, a practical method to produce non-Gaussian resource states is to start with a multimode Gaussian state and then measure a subset of its modes in the photon number basis (Fock basis) using PNR detectors. This projects the state of the leftover modes onto a non-Gaussian quantum state based on the heralded pattern, with its stellar rank depending on the number of detected photons \cite{gagatsos2019efficient}.
There is a rich history of heralding non-Gaussian states from Gaussian states \cite{lvovsky2020production,walschaers2021non}, including the generation of single-photon \cite{hong1986experimental,lvovsky2001quantum}, two-photon \cite{ourjoumtsev2006quantum,zavatta2008toward} and three-photon states \cite{cooper2013experimental}, photon-added coherent states \cite{zavatta2004quantum} and photon-added thermal states \cite{zavatta2007experimental}, displaced single-photon states \cite{lvovsky2002synthesis}, finite Fock superpositions \cite{bimbard2010quantum,yukawa2013generating}, photon-subtracted squeezed states \cite{wenger2004non,wakui2007photon,dufour2017photon,ra2020non}, cat states \cite{ourjoumtsev2006generating,neergaard2006generation,takahashi2008generation,gerrits2010generation,huang2015optical,asavanant2017generation,serikawa2018generation,chen2024generation}, multimode cat states \cite{ourjoumtsev2009preparation,morin2014remote,biagi2020entangling}, better cat states \cite{etesse2015experimental,sychev2017enlargement} and Gottesman--Kitaev--Preskill (GKP) states \cite{konno2024logical, Larsen2025} using Gaussian breeding protocols on top of non-Gaussian heralding \cite{brd-Vasconcelos2010, brd-Weigand2018, brd-Eaton2019, brd-Takase2024, aghaee2025scaling}.

Such quantum states are crucial building blocks for existing and future photonic devices \cite{bourassa2021blueprint, larsen2021deterministic, brd-Takase2024, aghaee2025scaling}. However, the intrinsic complexity of heralded non-Gaussian states generated by large-scale devices makes tracking such states challenging, even though it is an essential task for device characterization, architecture design, and benchmarking.
In this context, state decompositions \cite{williamson1936algebraic,giedke2003entanglement,botero2003modewise,houde2024matrix} are particularly useful, as they allow complex non-Gaussian states to be expressed in terms of simpler, well-understood components.

In this work, we present a new mathematical decomposition for Gaussian states, which we call the stellar decomposition due to its link with the stellar rank, that dramatically simplifies the analysis of non-Gaussian state generation protocols based on heralding photons from multimode Gaussian states. Specifically, we show that a multimode Gaussian state $G$ on $m+n$ modes can be expressed as a ``Gaussian core state'' $G_{\mathrm{core}}$ on $m+n$ modes, followed by a Gaussian channel $T$ on $m$ modes. $G_{\mathrm{core}}$ has the special property that if one counts $N$ photons across the $n$ heralding modes, then the cutoff in the Fock basis to represent the state of the remaining $m$ modes before $T$ will be at most $2N$. In general, without the stellar decomposition, the Fock cutoff required to express the output state to high numerical accuracy is much greater than $2N$. Moreover, we find the channel $T$ is independent of which photon pattern is heralded, which allows us to obtain general bounds on the levels of noise in the multimode Gaussian state $G$ that can be tolerated in non-Gaussian state preparation protocols. Our results are based on a new characterization of Gaussian CPTP maps in the Bargmann picture which may be of independent interest.

The structure of the paper is as follows: In the first section, we introduce the stellar decomposition for multimode Gaussian states, first in the pure state case, then for mixed states. We investigate properties of the decomposition, including the conditions under which the pieces of the decomposition, $G_{\mathrm{core}}$ and $T$, correspond to a physical state and channel, finding it to always be the case for pure states, but not for mixed states in general.
For convenience, we have listed our main contributions in \protect\cref{tab:summary}. 

\begin{table*}[t]
\centering
\begin{tabular}{lll}
\toprule
\textbf{Case} & \textbf{Result} & \textbf{Application} \\
\midrule

Pure States 
& \parbox[t]{6cm}{\protect\cref{prop:pure-state-case}: $|\psi\rangle = (U\otimes\mathbb{1})|\psi_\mathrm{core}\rangle$} 
& \parbox[t]{6cm}{Exact heralded state simulation\protect\phantom{\qquad} (\protect\cref{sec:heralded-simulation})} \\

Mixed States 
& \parbox[t]{6cm}{\protect\cref{prop:mixed-state-case}: $\rho = (\Phi\otimes \mathcal{I})(\rho_\mathrm{core})$}
& \parbox[t]{6cm}{Bounds on GKP generation \protect\phantom{\qquad\quad} (\protect\cref{sec:GKP-bound})} \\

Generic Operators 
& \parbox[t]{7cm}{\protect\cref{prop:formal-stellar-decomp}: $\mathrm{vec(G)} = (T\otimes\mathbb{1})\mathrm{vec}(S_\mathrm{core})$}
& \parbox[t]{5.5cm}{Exact noisy heralded state simulation (\protect\cref{sec:heralded-simulation})} \\

\bottomrule
\end{tabular}
\caption{Summary of the decomposition results. For each of our main decomposition results, we have provided their prominent application discussed in the manuscript. We highlight that the states, operators, and channels in this table are all Gaussian. The Gaussian core states (or operators) written in the second column have special Fock representations with special properties, allowing us to perform exact heralded state simulations, as discussed in \protect\cref{sec:heralded-simulation}. On top of this, our mixed state decomposition results (\protect\cref{prop:mixed-state-case} and \protect\cref{prop:m>=n}) provide us with useful bounds on the quality of any candidate GKP state that can be generated by a wide range of state preparation protocols.}
\label{tab:summary}
\end{table*}

We then explore two applications of the decomposition: simulation improvements, and bounding the quality of GKP states produced by photon heralding on multimode Gaussian states. For the former, the physicality conditions of the decomposition are irrelevant, while for the latter, if the Gaussian core state $G_{\mathrm{core}}$ and channel $T$ are physical then we can understand $T$ as adding a fundamental level of noise to any state produced by the protocol. Since we cannot always guarantee physicality conditions, we provide an additional technique to `factor out' as much of the channel $T$ as possible, while keeping the remaining state $G'$ physical, with the caveat that $G'$ does not enjoy the finite Fock cutoff benefits of Gaussian core states.

\section{Stellar decomposition}

We highlight that the stellar decomposition applies to bipartite Gaussian quantum systems. Let the subsystems denoted by $M$ and $N$ contain $m$ and $n$ modes, respectively. We often use cardinality notation to denote the size of the subsystem, i.e.~we have $|M|=m$ and $|N| = n$.

The Gaussian stellar decomposition techniques in this work make use of the Bargmann representation \cite{vourdas2006analytic}, also known as the stellar representation \cite{chabaud2020stellar}. In particular, we associate to all Gaussian states and operators a Bargmann function in the form of the exponential of a quadratic polynomial:
\begin{align}\label{eq:exp_quad_poly}
F(z) = c\exp\left(\frac12 z^{\mathrm{T}}Az + z^{\mathrm{T}}b\right)    \quad z\in\mathbb{C}^k,
\end{align}
and therefore we can focus on the finite parameters $A$, $b$, $c$, which we refer to as the ``Abc parametrization''. See \protect\cref{app:bargmann} for a detailed discussion of this formalism and the technical conventions we use. Certain features of Gaussian states are naturally represented in the Bargmann formalism. For instance, determining if a Gaussian state is pure can be naturally decided in this representation, as pure states have zero off-diagonal blocks in their A matrix, written in type-wise ordering (see \protect\cref{app:ordering-conventions} for ordering conventions). However, purity is mysteriously encoded in the covariance matrix of a Gaussian state, as one has to inspect the symplectic eigenvalues to decide if a given state is pure. Another example is projecting onto vacuum, which is equivalent to evaluating the Bargmann function at zero, and therefore it is equivalent to deleting the rows and columns of the $A$ and $b$ parameters corresponding to the variables evaluated at zero. The same operation is not as straightforward when using the covariance matrix and the vector of means. In this manuscript, we will see how a class of states which we call ``Gaussian core states'' can be naturally represented in the Bargmann formalism as well.  

A generic Gaussian completely positive (CP) map has an Abc parametrization $(A_\Phi,b_\Phi,c_\Phi)$ with the following block form
\begin{align}\label{eq:CP-map}
A_\Phi = \begin{bmatrix}
\Lambda^\ast_\Phi & \Gamma_\Phi\\
\Gamma_\Phi^\ast & \Lambda_\Phi
\end{bmatrix},
\quad
b_\Phi = \begin{bmatrix}
\beta^\ast_\Phi\\
\beta_\Phi
\end{bmatrix},
\end{align}
in the type-wise ordering (see \protect\cref{app:physicality} for details).
We obtain the following characterization, which we use as a basis for the proofs of our results:

\begin{lemma}\label{prop:channels-in-bargmann}
A Gaussian map, represented by a Gaussian Bargmann function with Abc parametrization as in \eqref{eq:CP-map}, is completely positive if and only if
\begin{align}
\Gamma_\Phi\ge0 \,\text{ and }\, c_\Phi\ge0.
\end{align}
\end{lemma} 
We refer to \protect\cref{app:physicality} for a proof, where we also identify conditions on the Bargmann representation of a Gaussian mixed state that correspond to its positivity (i.e., positivity and finite trace properties). Note that the reason different ordering conventions appear in the above proposition is that capturing complete positivity is more natural in type-wise ordering, while the trace-preserving property appears more naturally in the output-input ordering.

Throughout the manuscript we are using $\mathbb N=\{0,1,\cdots\}$ to denote non-negative integers. Moreover, given a tuple of $m$ integers $K = (k_1,\cdots,k_m)\in \mathbb N^m$, we use 
\begin{align}\label{eq:I_K}
I_K := \{(i_1,i_2,\cdots,i_m)|\,\forall j: 0\leq i_j \leq k_j \}
\end{align}
to indicate the $m$-dimensional hypercube of points in $\mathbb{N}^m$ with $K$ as the farthest corner from the origin.
We also recall that $\norm{K}_1 = \sum_{j} k_j$. For $n\in\mathbb N$, we define
the set of $m$-tuples with sum at most $n$ as
\begin{align}\label{eq:J-notation}
J^{m}_n:=\{(i_1,\cdots,i_m)|\sum_{j} i_j\leq n\}\subset \mathbb N^m.
\end{align}
For a tuple of non-negative integers, say $K = (k_1,\cdots,k_n)\in\mathbb N^n$, we use the following convention for multivariate factorial
\begin{align}
K! := k_1! k_2! \cdots k_n!,
\end{align}
and the multivariate derivative notation
\begin{align}
\partial^K_z := \partial^{k_1}_{z_1} \cdots \partial^{k_n}_{z_n},
\end{align}
where $z = (z_1,\cdots,z_n)$ is the tuple of derived variables.
\subsection{Gaussian core states}

Gaussian core states are the essence of the stellar decomposition. They are related to \textit{core states} \cite{lachman2019faithful,chabaud2020stellar}, which are defined as those with finite Fock support. We define Gaussian core states as follows.

\begin{definition}\label{def:core}
A bipartite Gaussian state $\rho$ over the $MN$ partition is a \emph{Gaussian core state} if for any $k_1,k_2\in\mathbb N^{n}$, the operator
\begin{align}
(\mathbb 1_M \otimes \bra{k_1}_N) \rho (\mathbb 1_M \otimes \ket{k_2}_N)
\end{align}
has finite support in the Fock basis.
\end{definition}
The above definition readily implies that measuring the subsystem $N$ in the Fock basis generates core states on $M$. Indeed, in \protect\cref{app:variation} we prove that enforcing $k_1=k_2$ in the definition above does not change the set of Gaussian core states, and hence they can be exactly understood as the states on $MN$ whose conditional states on $M$ have finite Fock support. The two-mode squeezed vacuum is a familiar example of Gaussian core state. The following claim establishes the condition for a state to be a Gaussian core state.

\begin{proposition}
A state $\rho$ is Gaussian core, if and only if its Ab part (denoted by matrix $A_\rho$ and vector $b_\rho$), written in the mode-wise order, appears in the following block form:
\begin{align}\label{eq:core-state-block-form}
A_\rho = \begin{bmatrix}
0 & \ast\\
\ast & \ast
\end{bmatrix}, \quad b_\rho = \begin{bmatrix}
0\\
\ast
\end{bmatrix}.
\end{align}
Moreover, for a Gaussian core state $\rho$, the stellar rank of the post-selected state
\begin{align}
(\mathbb 1\otimes \bra{k})\rho(\mathbb 1\otimes \ket{k})
\end{align}
is upper bounded by $\norm{k}_1$ i.e., the total number of measured photons.
\label{prop:core-states-in-bargmann}
\end{proposition}
We emphasize that the above proposition applies to both pure and mixed states. The proof is provided in \protect\cref{app:proof-of-ket-decomposition}.
Note that postselecting the first $m$ modes of the Gaussian core state does not project the last $n$ onto a finite Fock support: this property is directional. In order to herald the bottom $n$ modes onto a finite Fock support, the bottom-right block of $A$ has to be zero too.
In what follows, we consider different scenarios for decomposing a state into a Gaussian core state followed by local operations on $M$.

\subsection{Pure states}

We begin by the stellar decomposition of a Gaussian pure state $\ket\psi$ over $MN$, illustrated in \protect\cref{fig:ket_decomp}.

\begin{theorem}\label{prop:pure-state-case}
Let $|\psi\rangle$ be a pure $(m+n)$-mode Gaussian state bipartite over $MN$. There exists an $m$-mode Gaussian unitary $U$ such that
\begin{align}
\begin{split}
    |\psi\rangle &= (U\otimes\mathbb{1})|\psi_\mathrm{core}\rangle\\
    &=(U\otimes\mathbb{1})\sum_{k\in \mathbb{N}^{n}}|\mathrm{core}_{k}\rangle|k\rangle\\
    &=(U\otimes\mathbb{1})\sum_{k\in \mathbb{N}^{n}}\left(\sum_{j\in J^m_{\norm{k}_1}}c _{j,k}|j\rangle\right)|k\rangle,
\end{split}
\end{align}
where $|\psi_\mathrm{core}\rangle$ is an $(m+n)$-mode Gaussian state such that projecting its last $n$ modes onto a Fock state $\ket k$ with $k\in\mathbb{N}^n$ leaves a state on the first $m$ modes with Fock support up to $\|k\|_1 = \sum_{i=1}^n k_i$. The Abc parametrization of $\ket{\psi_{\mathrm{core}}}$ and $U$ can be computed in time $O((n+m)^3)$.
\end{theorem}
We highlight that replacing $(\ket{\psi_{\mathrm{core}}}, U)$ by $((W\otimes \mathbb 1)|\psi_\mathrm{core}\rangle, UW^\dagger)$ for any passive Gaussian unitary $W$ acting on $M$ is a valid stellar decomposition. Interestingly, the unitary $U$ can be found in a simple and intuitive way, as shown visually in \protect\cref{fig:proof_of_U}. We refer to \protect\cref{app:pf-of-pure-state} for a detailed proof of the theorem.

\begin{figure}[t]
\centering
\begin{tikzpicture}
    \draw[thick] (-.3,0) -- (1,1) -- (1,-1) -- cycle;
    \node at (0.5,0) {\large $\ket{\psi}$};
    \draw[thick] (1,0.5) -- (2,0.5);
    \draw[thick] (1,-0.5) -- (2,-0.5);
    \node[above] at (2,0.5) {$m$ modes};
    \node[below] at (2,-0.5) {$n$ modes};
    \node at (3,0) {$=$};
    \draw[thick] (3.5,0) -- (5,1) -- (5,-1) -- cycle;
    \node at (4.4,0) {\large $\ket{\psi_{\text{core}}}$};
    \draw[thick] (5.5,1) rectangle (6.5,0);
    \node at (6,.5) {\large $U$};
    \draw[thick] (5,0.5) -- (5.5,0.5);
    \draw[thick] (5,-0.5) -- (7.0,-0.5);
    \draw[thick] (6.5,0.5) -- (7.0,0.5);
\end{tikzpicture}
\caption{Stellar decomposition for pure Gaussian states (see \protect\cref{prop:pure-state-case}). The triangle shape indicates a Hilbert space vector (a ket). Every pure Gaussian state $|\psi\rangle$ can be decomposed as a Gaussian core state $|\psi_\mathrm{core}\rangle$ followed by an $m$-mode Gaussian unitary $U$ acting only on $M$. Any $(W|\psi_\mathrm{core}\rangle, UW^\dagger)$ pair for $W$ the unitary of an $m$-mode interferometer is a unitarily equivalent stellar decomposition.}
    \label{fig:ket_decomp}
\end{figure}

\begin{figure}[htbp]
\centering
\begin{tikzpicture}
    \draw[thick] (3,0) -- (4.5,1) -- (4.5,-1) -- cycle;
    \node at (4,0) {\large $\ket{\psi}$};
    \draw[thick] (4.5,0.5) -- (5.5,0.5);
    \draw[thick] (4.5,-0.5) -- (5.0,-0.5);
    \node at (5.5, -0.5) {\large $\bra{0^n}$};
    \node at (6.5, 0.0) {\large $=$};
    \draw[thick] (7.0,0) -- (8.5,1) -- (8.5,-1) -- cycle;
    \node at (7.9,0) {\large $\ket{\psi_\mathrm{core}}$};
    \draw[thick] (8.5,0.5) -- (9,0.5);
    \draw[thick] (8.5,-0.5) -- (9,-0.5);
    \node at (9.5, -0.5) {\large $\bra{0^n}$};
    \draw[thick] (9,1) -- (10,1) -- (10,0) -- (9,0) -- cycle;
    \node at (9.5, 0.5) {\large $U$};
    \draw[thick] (10,0.5) -- (10.5,0.5);
    \node at (6.5, -2.0) {\large $=$};
    \node at (8, -2) {\large $|0^m\rangle$};
    \draw[thick] (8.5,-2) -- (9,-2);
    \draw[thick] (9,-1.5) -- (10,-1.5) -- (10,-2.5) -- (9,-2.5) -- cycle;
    \node at (9.5, -2) {\large $U$};
    \draw[thick] (10,-2) -- (10.5,-2);
\end{tikzpicture}
\caption{A visual demonstration that the unitary of the pure stellar decomposition is one that maps the $m$-mode vacuum $\ket{0^m}$ to the heralded state $(\mathbb 1 \otimes\langle0^n|)|\psi\rangle$ corresponding to measuring vacuum on the last $n$ modes of $|\psi\rangle$.} 
\label{fig:proof_of_U}
\end{figure}
\begin{figure}[t]
\centering
\begin{tikzpicture}
    \draw[thick] (3.5,0) -- (5,1) -- (5,-1) -- cycle;
    \node at (4.4,0) {\large $\ket{\psi_{\text{core}}}$};
    \draw[thick] (5,0.5) -- (6,0.5);
    \draw[thick] (5,-0.5) -- (7.5,-0.5);
    \node at (8, -.5) {\large $\bra{k}$};
    \node at (8,.5) {\large Stellar rank $\leq \norm{k}_1$};
\end{tikzpicture}
\caption{Measuring a Fock pattern $k\in\mathbb{N}^n$ on subsystem $N$ leaves the state on subsystem $M$ with Fock support on the set of Fock states $J_{\norm{k}_1}^{m}$ (see \eqref{eq:J-notation} for the definition of such sets). This means the stellar rank is upper bounded by $\norm{k}_1$ i.e.~the total number of measured photons.}
\label{fig:stellar-rank-bound}
\end{figure}

\subsection{Mixed states}

Analogously to \protect\cref{prop:pure-state-case}, we ask if a mixed state bipartite over $MN$ can be decomposed into a Gaussian core state followed by a Gaussian channel $\Phi$ on $M$. Formally, we want to have
\begin{align}
\rho = (\Phi\otimes\mathcal I) (\rho_{\mathrm{core}}),
\end{align}
where $\Phi$ is a Gaussian channel acting on $M$, and $\rho_{\mathrm{core}}$ is a mixed Gaussian core state as illustrated in \protect\cref{fig:ket_decomp2}. The following theorem summarizes when this is possible. We use the following convention for the Ab parametrization of $\rho$ in mode-wise order:
\begin{align}\label{eq:Arho}
\begin{split}
A_\rho &= \begin{bmatrix}
A_\rho^{(m)} & R_\rho^{\mathrm{T}}\\
R_\rho & A_\rho^{(n)}
\end{bmatrix}\\
b_\rho &= \begin{bmatrix}
{b_\rho^{(m)}} \\{b_\rho^{(n)}}
\end{bmatrix}.
\end{split}
\end{align}

\begin{theorem}\label{prop:mixed-state-case}
A mixed $(m+n)$-mode Gaussian state $\rho$ bipartite over $MN$ admits a stellar decomposition
\begin{align}
    \rho = (\Phi\otimes \mathcal I)(\rho_\mathrm{core}),
\end{align}
{where $\Phi$ is a Gaussian channel on $M$}, if and only if \begin{align}\label{eq:iff-condition}
\mathrm{rank}\left( r_\rho r_\rho^\dag + \sigma_\rho \sigma_\rho^\dag \right) \leq m,
\end{align}
where $r_\rho$ and $\sigma_\rho$ are the block entries of the off-diagonal block of $A_\rho$ (see Eq.~\eqref{eq:Arho}):
\begin{align}
R_\rho = \begin{bmatrix}
r_\rho^\ast & \sigma_\rho^\ast\\
\sigma_\rho & r_\rho
\end{bmatrix}.
\end{align}
The time-complexity for computing the Abc parametrization of $\rho_{\mathrm{core}}$ and $\Phi$ is $O((n+m)^3)$.
\end{theorem}
\begin{figure}[t]
    \centering
\scalebox{1.2}{
\begin{tikzpicture}

\draw[thick, black] 
(-1, -1.4) arc[start angle=270, end angle=90, radius=0.9] -- (-1, 0.4) -- (-1, -1.4) -- cycle;
\node at (-1.35, -0.5) {\large $\rho$};

\draw[thick, black] (-1, 0) -- (0, 0) node[ above] {\small $m$ modes};
\draw[thick, black] (-1, -1) -- (0, -1) node[below] {\small $n$ modes};
\node at (.5, -.5) {$=$};

\draw[thick, black] 
(2, -1.4) arc[start angle=270, end angle=90, radius=0.9] -- (2, 0.4) -- (2, -1.4) -- cycle;
\node at (1.6, -0.5) {\large $\rho_{\text{core}}$};

\draw[thick, black] (2.6, -.4) rectangle (3.4, 0.4);
\node at (3, 0) {\large $\Phi$};

\draw[thick, black] (2, 0) -- (2.6, 0);
\draw[thick, black] (2, -1) -- (4, -1);

\draw[thick, black] (3.4, 0) -- (4, 0);

\end{tikzpicture}}
    \caption{Stellar decomposition for mixed Gaussian states (see \protect\cref{prop:mixed-state-case}). Not every mixed Gaussian state $\rho$ admits a stellar decomposition in terms of \emph{physical} parts (i.e.~ a Gaussian mixed state and a Gaussian channel). The rounded shape indicates density matrices.}
    \label{fig:ket_decomp2}
\end{figure}
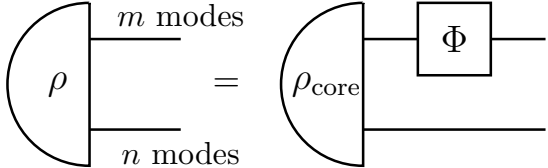

We refer to \protect\cref{app:pf-of-mixed-state} for a proof. If $m < n$, not all mixed Gaussian states over $MN$ admit a physical stellar decomposition. On the other hand, if $m\geq n$ not only can the decomposition always be carried out, but the Gaussian core state can always be chosen to be pure. This is because the condition \eqref{eq:iff-condition} is always satisfied if $m\geq n$ as $r_\rho r_\rho^\dag + \sigma_\rho\sigma_\rho^\dag$ is an $n\times n$ matrix and hence $m$ is greater than or equal to its rank.

\begin{proposition}\label{prop:m>=n}
In the special case where the $m+n$ modes of the initial mixed state $\rho$ are partitioned such that $m\geq n$, a physical stellar decomposition always exists and the Gaussian core state is pure:
\begin{align}
    \rho = (\Phi\otimes \mathcal I)(|\psi_\mathrm{core}\rangle\langle\psi_\mathrm{core}|).
\end{align}
We can compute the Abc parameters of $\ket{\psi_{\mathrm{core}}}$ and $\Phi$ in time $O(m^3)$.
\end{proposition}

We refer to \protect\cref{app:pf-of-m>=n} for a proof. Intuitively, this result shows that the channel on the top mode is capable of extracting \textit{all of the noise} present in the initial Gaussian state. In other words, we can factor out the entire source of impurity of $\rho$ on subsystem $M$. This idea of extracting noise out of a subset of modes will later help us put bounds on the quality of heralded states.

We remark that in cases where $m<n$, one can simply add $n-m$ ancillary vacuum states to $M$ and obtain the stellar decomposition from \protect\cref{prop:m>=n}. This decomposition gives us a core state $\rho_{\mathrm{core}}$ over $2n$ modes where the $n-m$ ancillary modes can be entangled with the rest of the modes, and a channel $\Phi$ that maps $n$ modes to $m$ modes. From a state preparation perspective, we would like to apply interferometers on several instances of our heralded Gaussian states \cite{aghaee2025scaling}, and are interested in cases where we can commute the channel through the subsequent beam-splitter network. Following this practical motivation, we investigated the application of stellar decomposition to cases where the channel $\Phi$ maps only $m$ modes to $m$ modes (characterized by our results above).

\subsection{General operators}

Recall that corresponding to any linear operator is a Bargmann function and that the operator is called Gaussian if its Bargmann function is in the Gaussian form given by Eq.~\eqref{eq:exp_quad_poly} (see \protect\cref{app:bargmann} for more details). We can vectorize any Gaussian operator, say $G$, to obtain a Gaussian vector $\mathrm{vec}(G)$. Formally, letting $F_{G}:\mathbb C^{n}\to\mathbb C$ represent the stellar function of $G$, we have (see \cite{chabaud2021continuous})
\begin{align}
\mathrm{vec}(G):=F_G(a_1^\dagger,\cdots, a_n^\dagger) \ket{0^n}.
\end{align}
As an example of vectorization, we have that the Fock damping operator on one mode after vectorization becomes proportional to a two-mode squeezed vacuum state. Note that by vectorization, we obtain a Gaussian vector with the same Abc triple as the initial object. In the rest of this section, we provide a result on the stellar decomposition of a generic Gaussian vector. We start by defining Gaussian core vectors. Note that such Gaussian vectors do not need to satisfy any physicality constraint as we do not give physical meanings to them, and they might correspond to functions outside of the Siegel-Bargmann Hilbert space as they can be non-normalizable.

\begin{definition}
A Gaussian vector $\mathrm{vec}(G)$ defined in the joint space $MN$ is called \emph{Gaussian core} if it has a finite Fock support when we project the $N$ subspace onto number states.
\end{definition}

In general, the stellar decomposition can be carried out on any Gaussian operator $G$ in a formal way, if no physicality requirement is needed. Formally:
\begin{proposition}\label{prop:formal-stellar-decomp}
Any Gaussian vector $\mathrm{vec}(G)$ defined on spaces $MN$ can be decomposed as
\begin{align}
\mathrm{vec}(G) = (T\otimes \mathbb 1) \mathrm{vec}(S_{\mathrm{core}}),
\end{align}
where $\mathrm{vec}(S_{\mathrm{core}})$ is a Gaussian core vector, and $T$ is a Gaussian operator that acts as identity on space $N$. This decomposition can be computed in time $O(n^2)$. In fact, each entry of the Ab representations of $T$ and $S_{\mathrm{core}}$ can be computed in constant time.
\end{proposition}
The proof is provided in \protect\cref{app:pf-of-formal-decomposition}.
Also, we refer to \protect\cref{fig:ket_decomp3} for an illustration of the decomposition. We highlight that although not all mixed states admit a physical stellar decomposition, one can always vectorize a Gaussian mixed state and perform the formal decomposition above.
Note that the formal decomposition can be performed more efficiently than \protect\cref{prop:pure-state-case} and \protect\cref{prop:mixed-state-case}, as we do not impose physicality conditions on $\ket{S_{\mathrm{core}}}$ and $T$. Nevertheless, as we discuss below, this decomposition provides us with a powerful tool for simulating physical systems.

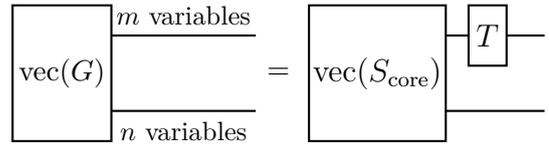
\begin{figure}
    \centering
    \begin{tikzpicture}
    \draw[thick, black] (-2.2, -1.4) rectangle (-.9, 0.4);
    \node at (-1.55, -0.5) {$\mathrm{vec}(G)$};

    \draw[thick, black] (-.9, 0) -- (1, 0) node[midway, above] {\small $m$ variables};
    \draw[thick, black] (-.9, -1) -- (1, -1) node[midway, below] {\small $n$ variables};
    \node at (1.3, -.5) {$=$};
    \draw[thick, black] (1.7, -1.4) rectangle (3.5, 0.4);
    \node at (2.6, -0.5) {\large $\mathrm{vec}(S_{\mathrm{core}})$};

    \draw[thick, black] (3.8, -.4) rectangle (4.3, 0.4);
    \node at (4.05, 0) {\large $T$};

    \draw[thick, black] (3.5, 0) -- (3.8, 0);
    \draw[thick, black] (3.5, -1) -- (4.8, -1);

    \draw[thick, black] (4.3, 0) -- (4.8, 0);

\end{tikzpicture}
    \caption{Stellar decomposition for Gaussian operators (see \protect\cref{prop:formal-stellar-decomp}). If no physicality requirement is necessary, any multimode Gaussian object $G$ can be decomposed into two Gaussian parts such that $S$  has the core property of needing a finite Fock cutoff on the subsystem $M$, when we project the subsystem $N$ onto the Fock basis. 
    }
    \label{fig:ket_decomp3}
\end{figure}

\section{Applications}\label{sec:applications}

In this section we consider applications of the stellar decomposition such as simulating heralded non-Gaussian states and assessing the quality of such states. In particular, in \protect\cref{sec:heralded-simulation} we show how one can exactly simulate heralded non-Gaussian states, and in \protect\cref{sec:GKP-bound} we show how one can upper bound the best quality (measured in terms of effective squeezing \cite{duivenvoorden2017single}) of a candidate GKP state that can be obtained based on Gaussian Boson sampling circuits.

\subsection{Simulating heralded non-Gaussian states}\label{sec:heralded-simulation}

As we shall demonstrate, the stellar decomposition leads to useful techniques for evaluating the Bargmann function, and for simplifying derivatives of Bargmann functions.

For instance, consider a $(1+n)$-mode pure Gaussian state with Bargmann function $F_{|\psi\rangle}(z,w)$ with $z\in\mathbb{C}$ and $w\in\mathbb{C}^n$, and project the last $n$ modes onto $\ket{N_{\mathrm{ph}}}^{\otimes n}$. Let us denote the Fock representation of a state $\ket{\psi}$ over $1+n$ modes as $\psi[\ell]$ for $\ell\in\mathbb N^{1+n}$, i.e.
\begin{align}
\ket{\psi} = \sum_{\ell\in\mathbb N^{1+n}} \psi[\ell] \, \ket\ell.
\end{align}
Recall that using state-of-the art algorithms for computation of Fock amplitudes \cite{yao2024riemannian, bjorklund2019faster}, we can compute the Fock coefficient $\psi[\ell]$ for a Gaussian state $\ket\psi$ over $n$ modes in time $O(\prod_i \ell_i)$ or $O(n^2\|\ell\|_1^{n/2})$ depending on the algorithm. Therefore, to compute the post-selected state $\langle{N_{\mathrm{ph}}}^{\otimes n}|\psi\rangle$ up to some cutoff $C$, we require time $O(C{N_{\mathrm{ph}}^n})$ or $O(n^2\|C+nN_\mathrm{ph}\|_1^{n/2})$. We now show that through the stellar decomposition, one can choose $C = nN_{\mathrm{ph}}$ and obtain an exact representation of the state. To see this, note that by \protect\cref{prop:pure-state-case} we can write
\begin{align}
\ket{\psi} = (U\otimes \mathbb 1)\ket{\psi_{\mathrm{core}}},
\end{align}
where $U$ is a single-mode Gaussian unitary. Therefore, we have
\begin{align}
|\psi_\mathrm{out}\rangle &=(\mathbb 1 \otimes \bra{N_{\mathrm{ph}}}^{\otimes n}) \ket\psi \\
&= U\left((\mathbb 1 \otimes \bra{N_{\mathrm{ph}}}^{\otimes n}) \ket{\psi_{\mathrm{core}}}\right)\!.
\end{align}
Using the core properties of $\ket{\psi_{\mathrm{core}}}$, we know that
\begin{align}
\psi_{\mathrm{core}}[k, \underbrace{N_{\mathrm{ph}}, \cdots, N_{\mathrm{ph}}}_{n\text{ times}}] = 0\quad \mathrm{if\ }k\geq nN_\mathrm{ph},
\end{align}
due to the stellar rank upper bound in \protect\cref{prop:pure-state-case}. Therefore, we can compute the state $(\mathbb 1 \otimes \bra{N_{\mathrm{ph}}})\ket{\psi_{\mathrm{core}}}$ exactly if we choose a cutoff $C = n N_{\mathrm{ph}}$ on the first mode. We can view the stellar decomposition as transforming the problem of computing $n$ multivariate derivatives of degree $N_{\mathrm{ph}}$ on an $(1+n)$-variate complex Gaussian, into the problem of computing a \emph{single} $(nN_{\mathrm{ph}})$-th order derivative of another complex Gaussian with only two variables. Mathematically, we have
\begin{align}
F_{|\psi_\mathrm{out}\rangle}(z_1)&=\frac{\partial_{z_2}^{N_{\mathrm{ph}}}\cdots\partial_{z_{n+1}}^{N_{\mathrm{ph}}}}{\sqrt{N_{\mathrm{ph}}!^{n}}} F_{\psi}(z_1,\dots,z_{n+1})|_{z_{i>1}=0}\\
&=\sum_{k=0}^{n N_{\mathrm{ph}}}c_k\frac{\partial^k_{z_2}}{\sqrt{k!}} \,  F_U(z_1,z_2)|_{z_2=0},
\end{align}
where
\begin{align}
c_k := \psi_{\mathrm{core}}[k,N_{\mathrm{ph}}, \cdots, N_{\mathrm{ph}}].
\end{align}
Note that although $(\partial_{z_2}^k F_U(z,0))_{k=0}^{n N_{\mathrm{ph}}}$ can be computed in polynomial time, the computation of $c_k$ still takes $O({N_{\mathrm{ph}}^n})$ time. 

We note that leveraging the stellar decomposition, we can precisely compute the effective squeezing parameters for any state obtained by measuring all but one mode of a Gaussian state in the photon number basis (for background information, see \protect\cref{sec:GKP-bound}). To this end, note that $U^\dagger D_{\vec v} U = D_{\vec v'}$, meaning we only need to evaluate $\bra{\psi_{\mathrm{out}}}D_{\vec v'}\ket{\psi_{\mathrm{out}}}$. As $\ket{\psi_{\mathrm{out}}}$ has a finite Fock support, and the Fock matrix elements of displacement operators have a closed-form expression \cite{cahill1969ordered, NIST:DLMF}, we can exactly compute the effective squeezing.

The kind of non-Gaussian objects that we are interested in are those obtained by projecting one or more modes onto the Fock basis. As seen above, the Bargmann function of such objects is the original Bargmann function with derivatives applied to it and evaluated at zero. 
More generally, if we want to project an $(m+n)$-mode state $|\psi\rangle$ onto the $n$-mode state $|\phi\rangle=\sum_{k\in I_K}\phi[k]\, |k\rangle$ where $k\in\mathbb{N}^n$ is a multi-index, we obtain the $m$-mode state $(\mathbb 1\otimes\langle\phi|)|\psi\rangle$ whose Bargmann function is
\begin{align}
\begin{split}
 F_{(\mathbb 1\otimes \langle\phi|)|\psi\rangle}(z)&=\sum_{k\in I_K}\phi[k]^* (\mathbb 1 \otimes \langle k|)|\psi\rangle\\
 &= \sum_{k\in I_K}\phi[k]^*\frac{\partial^k_v}{\sqrt{k!}}F_{|\psi\rangle}(z,v)|_{v=0}\\
 &=\sum_{j\in J_{\|K\|_1}^m} d_j\frac{\partial^j_w}{\sqrt{j!}}F_U(z,w)|_{w=0}
\end{split}
\end{align}
where $d_j = \sum_{k\in I_K} \phi[k] \, \psi_{\mathrm{core}}[j,k]$. In the last step we have used the stellar decomposition of $|\psi\rangle$: $F_{|\psi\rangle}(z,v) =\int_{\mathbb{C}^m}d\mu(w)F_U(z,w)F_{|\psi_\mathrm{core}\rangle}(w^*,v)$.
So we can regard the polynomial of derivatives $P_{\langle\phi|}(\partial_v)=\sum_{k\in I_K}\phi[k]^*\frac{\partial^k_v}{\sqrt{k!}}$ as a differential operator acting on the Bargmann function, and we can delay its application as needed in order to gain an advantage. For example, we can apply Gaussian operations to the leftover variables $z$, because it is much simpler to do so while the Bargmann function is still in Gaussian form (note that we used a similar technique above for exact computation of effective squeezing). We provide a numerical implementation of this technique in the MrMustard library \cite{mrmustard}.

\subsection{General bounds on the quality of GKP states}
\label{sec:GKP-bound}

A promising approach to creating GKP states is based on Gaussian boson sampling (GBS) devices. Such architectures have been proposed to prepare GKP states in all-photonic settings \cite{Tzitrin2020, gbs-Fukui2022, gbs-Takase2023, aghaee2025scaling} and demonstrated experimentally in Ref. \cite{Larsen2025}. In this setting, we prepare a Gaussian state, then post-select on all but one of the modes via photon number measurements, to make a non-Gaussian state. We aim to make the non-Gaussian state resemble a GKP state, with the goal to maximize the effective squeezing of the state as a figure of merit.

An example of a GBS architecture for making GKP states is the ``staircase'' family of circuits~\cite{gbs-Takase2023,aghaee2025scaling, Larsen2025}. In \protect\cref{fig:staircase-examples} we have demonstrated a two-mode example of the staircase architectures. Note that in realistic settings, we need to consider a loss channel before the number measurement. This imperfection can pose fundamental bounds on the best quality of state we can achieve by any post-selection procedure. For completeness, we have also provided an example of a four-mode staircase architecture in \protect\cref{fig:wide-staircase}, to show how one can generalize this setting to an arbitrary number of modes.

\begin{figure}[t]
\centering
\begin{quantikz}
\lstick{$\ket{0}$} & \gate{S(r)} & \qw\arrow[d] &\qw \quad\text{candidate GKP}\\
\lstick{$\ket{0}$} & \gate{S(-r)} & \qw & \gate{\mathcal L_\eta} & \meter{}
\end{quantikz}
\caption{The staircase architecture for GKP state generation. The figure shows a $2$-mode architecture. The second mode is subject to a loss channel $\mathcal L_\eta$ and is post-selected by a number-basis measurement.}
\label{fig:staircase-examples}
\end{figure}
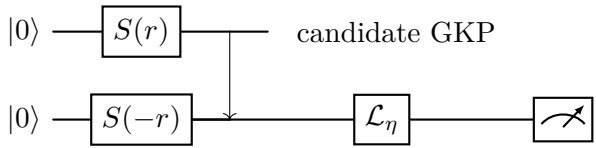

\begin{figure}[t]
\begin{quantikz}
\lstick{$\ket{\xi_1}$} & \arrow[d, "\theta_1"'] & \quad\text{candidate GKP}\\
\lstick{$\ket{\xi_2}$}  & & \arrow[d, "\theta_2"'] & & & \meter{}\\
\lstick{$\ket{\xi_{3}}$} & & &  &\arrow[d, "\theta_3"'] & \meter{} \\
\lstick{$\ket{\xi_4}$} & & & & & \meter{} 
\end{quantikz}
\caption{An example of a staircase architecture on $1+3$ modes. The state $\ket{\xi_i} = S(\xi_i) \ket{0}$ represents the squeezed vacuum state. The circuit can be made arbitrarily wider following the staircase pattern. The loss channels before the photon-number measurements are omitted, but taken into account in the theoretical analysis.}
\label{fig:wide-staircase}
\end{figure}
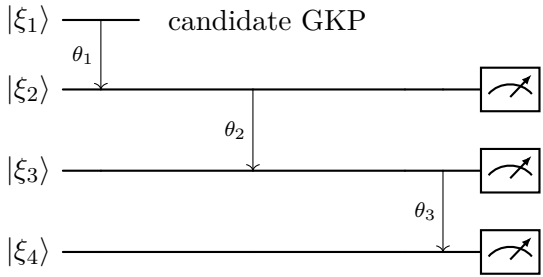

Motivated by such GKP generation protocols, we consider having a state factory that allows for the preparation of Gaussian states and has the ability to post-select using non-Gaussian measurements (note that we do not restrict to number measurements in our analysis). We then can ask: given a mixed, i.e., noisy, Gaussian state $\rho$, is there any scheme that measures all but one of the modes that can prepare a \textit{good enough} GKP state with \textit{any non-zero probability}? To fix our definition of a \textit{good enough} GKP state candidate, we employ the stabilizer expectation values of the candidate GKP state \cite{ duivenvoorden2017single, brd-Weigand2018, hastrup2021measurement}, which can be related to the effective squeezing that typically parametrizes thresholds for fault-tolerant quantum computation using GKP states \cite{menicucci2014fault, bourassa2021blueprint, larsen2021deterministic, aghaee2025scaling}. 

To fix our notation, we let
\[D_{\vec v} = \exp(i(v_2 Q - v_1 P))\]
denote the displacement operator that shifts by $\vec v = (v_1,v_2)$ in the phase space. We also recall that a Gaussian channel $\Phi$ is parametrized by $(X,Y)$ if it maps a state with covariance matrix $\Sigma$ via
\begin{align}
\Phi: \Sigma \mapsto X \Sigma X^{\mathrm{T}} + Y.
\end{align}
Later, in \protect\cref{app:Abc-of-Gaussians}, we provide formulas to transform between the phase-space and Bargmann representations of channels.

From hereon, we focus on the task of preparing a sensor state, which is a particular instance of a GKP state. We highlight that these choices are arbitrary and the results of this section are readily applicable to any choice of GKP state, by choosing the proper stabilizers. 

We now recall the formal definition of $Q$ and $P$ quadrature effective squeezings (denoted by $\sigma_q$ and $\sigma_p$ respectively) 
\begin{align}\label{eq:eff-sq}
\begin{split}
\sigma_p^2 :&= \frac{2}{\pi}\log\left(\frac{1}{|\tr[\rho D_{\sqrt{2\pi\hbar}\vec e_2}]|} \right),\\
\sigma_q^2 :&= \frac{2}{\pi}\log\left(\frac{1}{|\tr[\rho D_{\sqrt{2\pi\hbar}\vec e_1}]|} \right),
\end{split}
\end{align}
which intuitively capture the average broadening of each peak of a candidate GKP state $\rho$. To provide an explanation about this definition, we note that the GKP encoding naturally defines a periodic domain, and for a random variable $Z$ defined on a periodic domain $[0,2\pi)$ one can define variance via $\sigma^2 = -2\log |\mathbb E[e^{iZ}]|$ \cite{mardia1975statistics}. We refer to discussions around equation (6) in the Supplementary Information of \cite{aghaee2025scaling} for a more elaborate explanation. One might be interested to work with the symmetric effective squeezing \cite{aghaee2025scaling}, defined as the mean of the above
\begin{align}\label{eq:symmeric-eff-sqz}
\sigma_{\mathrm{sym}}^2 := \frac{\sigma_q^2 + \sigma_p^2}{2}.
\end{align}

Note that effective squeezing has a different definition from how we define the input squeezing levels of squeezed states used in GBS devices to prepare non-Gaussian states, so in general the two are not related. For example, preparing a two-mode squeezed vacuum state and counting one photon in the heralding mode will yield, in the heralded mode, a single photon state which has a positive effective squeezing value regardless of the squeezing level of the input state, meaning the output state effective squeezing can in general exceed the input squeezing level. The input squeezing level is, however, directly related to the probability of obtaining a given PNR outcome, rather than the effective squeezing of the output state associated with that outcome. Instead, it is the number of photons detected in a GBS device that is related to the effective squeezing of the GKP state produced. It was shown numerically that the quality of approximate GKP states can be improved with higher stellar rank (see Fig. S19 of Ref. \cite{Larsen2025}), and it was conjectured (with numerical evidence) \cite{PhysRevA.100.052301} and proven \cite{Aralov2025} that an arbitrary state of a given stellar rank $N$ can be produced using a GBS device detecting $N$ photons, which suggests GKP states with arbitrarily high quality can be produced by GBS devices detecting high enough numbers of photons.

In general, computing the best achievable effective squeezing, given a Gaussian state $\rho$, might be an extremely difficult problem. This is due to the fact that the POVM element should be a highly non-Gaussian object, and computing the projection of Gaussian states into non-Gaussian states is $\#\mathsf{P}$-hard \cite{aaronson2011computational, kruse2019detailed, hamilton2017gaussian, grier2022complexity}, let alone finding the optimal measurement. 
Instead, we aim to derive efficiently computable upper bounds on the quality of any candidate GKP state (measured with either the quadrature or the symmetric effective squeezings) that can be prepared via post-selection on a given Gaussian state $\rho$.

Interestingly, if $\rho = (\Phi\otimes \mathcal I) (\sigma)$ for some one-mode channel $\Phi$, from \cite[Eq.~(5.55)]{serafini2017quantum}, we have that for any displacement $D_{\vec v}$
\begin{align}
\tr(\rho D_{\vec{v}}) = \exp(-\frac1{2\hbar^2}\vec v^{\mathrm{T}} Y \vec v) \tr(\sigma D_{\vec u}),
\end{align}
for some other displacement $D_{\vec u}$. As a result, if $\Phi$ can be factored out from $\rho$, we can bound the effective squeezing of $\rho$ by a stabilizer $\vec v$ via
\begin{align}\label{eq:rachel-bound}
\abs{\tr(\rho D_{\vec v})} \leq \exp(-\frac1{2\hbar^2} \vec v^{\mathrm{T}} Y \vec v).
\end{align}
This implies that a physical stellar decomposition for a state $\rho$ as in \protect\cref{prop:mixed-state-case} directly leads to bounds on the quality of the best GKP state which can be obtained from that Gaussian  state by heralding.

\begin{figure}
    \centering
\begin{tikzpicture}

\draw[thick, black] 
(-1, -2.4) arc[start angle=270, end angle=90, radius=1.3];
\draw[thick, black] (-1, 0.2) -- (-1, -2.4);  
\node at (-1.5, -1.1) {\Large $\rho$};

\draw[thick, black] (-1, 0) -- (0, 0);
\draw[thick, black] (-1, -1.2) -- (0, -1.2);
\node at (-.5, -1.5) {$\vdots$};
\draw[thick, black] (-1, -2) -- (0,-2);
\node at (1, -1.5) {\small $n$ modes};

\node at (1.5, -1) {$=$};

\draw[thick, black] 
(3.5, -2.4) arc[start angle=270, end angle=90, radius=1.3];
\draw[thick, black] (3.5, 0.2) -- (3.5, -2.4);  
\node at (3, -1) {\Large $\rho'$};

\draw[mybrace=0.5] (.2,-1) -- (.2,-2.2);

\draw[thick, black] (4.5, -.5) rectangle (5.5, 0.5);
\node at (5, 0) {\Large $\Phi$};

\draw[thick, black] (3.5, 0) -- (4.5, 0);
\draw[thick, black] (3.5, -1.2) -- (5.7, -1.2);
\draw[thick, black] (3.5, -2) -- (5.7, -2);
\draw[thick, black] (5.5, 0) -- (5.7, 0);
\node at (4.4, -1.5) {$\vdots$};

\node[rotate = 270] at (2,-3) {\Large $\Rightarrow$};
\node at (2, -4) {\large $|\mathrm{tr}(\rho D_{\vec v})| \leq \exp(-\frac{1}{2\hbar^2}\vec v^{\mathrm{T}} Y \vec v)$};

\end{tikzpicture}

\caption{Factoring out a channel can help bound the quality of a candidate GKP state. The expression $|\tr(\rho D(\vec v))|$ represents the effective squeezing in direction $\vec v$ (see \eqref{eq:eff-sq}). \protect\cref{prop:sdp} provides an efficient computation of the best bound.}
\end{figure}
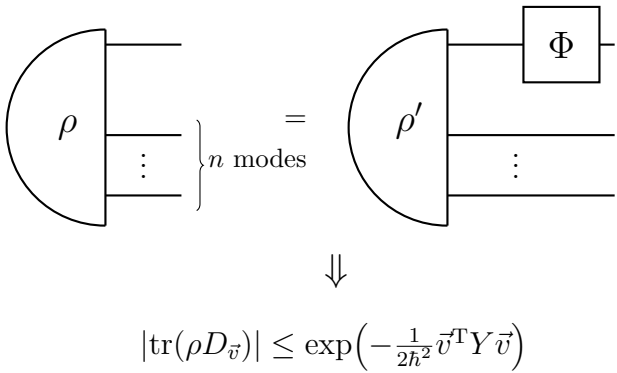

Further, when no physical stellar decomposition is available, we are looking to maximize $\vec v^{\mathrm{T}} Y \vec v$ under the condition that there exists a channel with $(X,Y)$ representation that can be factored out from $\rho$. We show that this optimization can be done efficiently by reformulating it as a semi-definite program (SDP):
\begin{proposition}\label{prop:sdp}
For any Gaussian density matrix $\rho$ with covariance matrix $\Sigma$, the following SDP
\begin{align}
\begin{split}
&\mathrm{maximize} \quad \frac{1}{\pi \hbar^2}\vec v^{\mathrm{T}} Z \vec v\\
&\mathrm{s.t.} \quad Z \geq i\frac{\hbar}{2} \Omega_1,\\
& \quad \Sigma + 
\begin{bmatrix}
0 & 0\\
0 & i\frac{\hbar}{2}\Omega_n
\end{bmatrix}
\geq  \begin{bmatrix}
Z & 0\\
0 & 0
\end{bmatrix},\\
&\quad Z\in\mathrm{Herm}(\mathbb C^{2\times 2}),
\end{split}
\end{align}
computes an upper bound on the largest achievable effective squeezing, with respect to a given displacement $D_{\vec v}$ that one might achieve by performing any (possibly non-Gaussian) post-selection that may succeed with any non-zero probability.
\end{proposition}

The proof is given in \protect\cref{app:pf-of-sdp}. In \protect\cref{app:dual-sdp}, we discuss how other figures of merit can also be directly bounded by a slight alteration of the SDP.  Furthermore, recall that the value of an SDP is achievable if and only if it is equal to the value of its dual formulation \cite{boyd2004convex}. We provide the dual problem in \protect\cref{app:dual-sdp} as well. We highlight that this technique can be extended to obtain an SDP computing bounds on the quality of multi-mode GKP encodings, as presented in \protect\cref{app:sdp-extension}.

Although the computation of the best achievable quality could be very hard, we have provided an efficient algorithm for non-trivially bounding this number. It remains open how tight this bound is, or if we can further close this gap with better algorithms. We study the bounds that can be imposed by our method to the staircase architecture, introduced in \protect\cref{fig:staircase-examples}.

\subsection{Staircase architectures}

In this section, we study the staircase GBS architecture, and provide bounds on the symmetric effective squeezing of the GKP states produced by the device, as a simple example demonstrating the power of the stellar decomposition of Gaussian states. Recall the two-mode circuit shown in \protect\cref{fig:staircase-examples}. Since it only has two modes, i.e., $m=n=1$, we can use either the stellar decomposition (\protect\cref{prop:m>=n}), or the optimal solution provided by the SDP (\protect\cref{prop:sdp}) to obtain bounds on the potential best GKP state which can be obtained using this architecture. We have provided both bounds as a function of the loss parameter of the circuit in \protect\cref{fig:example-circuit}. The circuit parameters are chosen to achieve high symmetric effective squeezing based on the study in \cite{aghaee2025scaling}. 

Interestingly, the decompositions based on \protect\cref{prop:m>=n} (stellar decomposition) and \protect\cref{prop:sdp} provide the same result. Note that \protect\cref{prop:sdp} is guaranteed to provide the optimal bound. We attribute this coincidence to the fact that the stellar decomposition factors out all of the non-unitary evolution and puts it on the first mode. Note that \protect\cref{prop:sdp} is applicable for arbitrary architectures, while the stellar decomposition's application is limited to cases identified by \protect\cref{prop:dm-physicality}.

\begin{figure}
\centering
\includegraphics[width = \columnwidth]{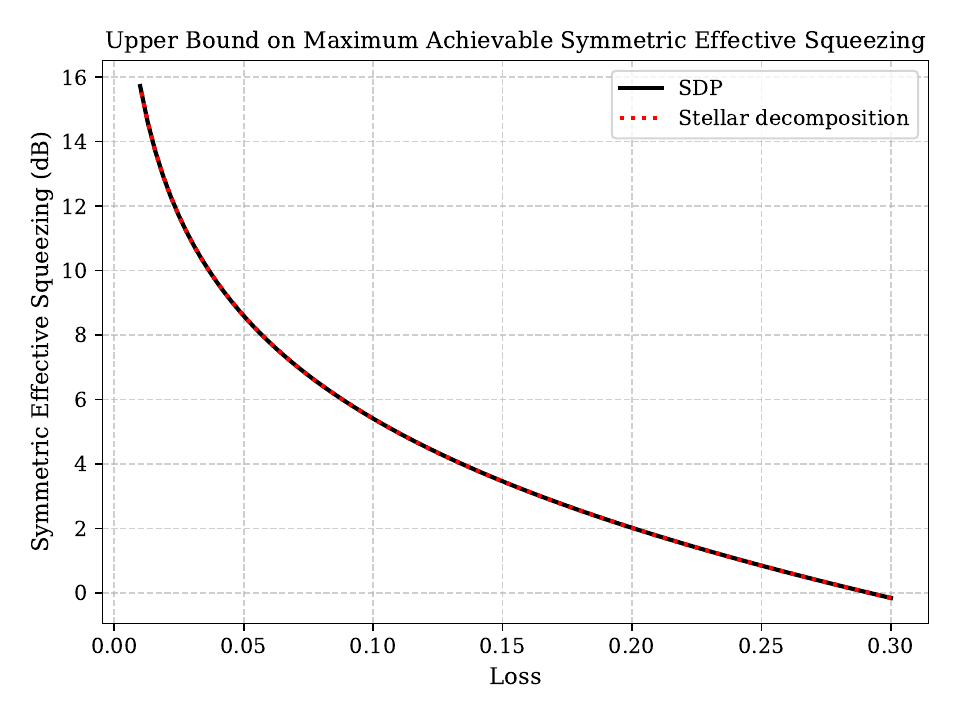}
\caption{The upper bound on staircase architectures. Here, we interfere 15 dB squeezed vacuum states at a beam-splitter with angle $\theta = 0.4$ \cite{aghaee2025scaling} for the circuit shown in \protect\cref{fig:staircase-examples}. Using the bounds developed in our work, we can upper bound the maximum possible symmetric effective squeezing of this proposed architecture. We can see that the result we get from the stellar decomposition, which factors out a channel on a pure state, coincides with the bound given by the SDP. As shown in \protect\cref{prop:staircase-values}, the same bound also applies to larger staircases.}
\label{fig:example-circuit}
\end{figure}

Next, we consider multimode staircase architectures, with the four-mode instance depicted in \protect\cref{fig:wide-staircase}. Such GBS architectures have been studied as probabilistic sources of high-quality GKP states \cite{gbs-Takase2023, fukui2022efficient, aghaee2025scaling}. We determine how the bound given by the SDP from \protect\cref{prop:sdp} changes as we increase the number of modes. Interestingly, the bounds found from the two-mode case are identical to any bound found for the multimode cases. To formulate this statement concretely, let the output of the SDP on an architecture with squeezings $\xi_1,\xi_2,\cdots,\xi_M$, beam-splitter angles $\theta_1,\cdots,\theta_{M-1}$, and loss $\eta$ be denoted by 
\begin{align}
B((\xi_1,\xi_2,\cdots,\xi_M), (\theta_1,\cdots,\theta_{M-1}), \eta),
\end{align}
noticing that the length of the input determines the size of the staircase architecture. Moreover, let the feasible set of the SDP be denoted by
\begin{align}
\mathcal Z((\xi_1,\xi_2,\cdots,\xi_M), (\theta_1,\cdots,\theta_{M-1}), \eta).
\end{align}
We use the shorthand notation $\xi_{1:M}=(\xi_1,\xi_2,\cdots,\xi_M)$ and $\theta_{1:M-1}=(\theta_1,\cdots,\theta_{M-1})$ for brevity.
We then have:
\begin{proposition}\label{prop:staircase-values}
For any $M\geq 2$,
\begin{align}
\mathcal Z(\xi_{1:M}, \theta_{1:M-1}, \eta) = \mathcal Z(\xi_{1:2}, \theta_1, \eta),
\end{align}
and therefore,
\begin{align}
B(\xi_{1:M}, \theta_{1:M-1}, \eta) = B(\xi_{1:2},\theta_1, \eta).
\end{align}
\end{proposition}
The proof is provided in \protect\cref{app:pf-of-staircase}. This result allows us to put universal bounds on staircase designs with arbitrarily many modes. For example, with maximum input squeezing of $15$ dB, and $\theta_1 = 0.4$ (parameters reported in \cite{aghaee2025scaling}) no staircase architecture can get more than $10$ dB symmetric effective squeezing (the approximate threshold obtained in \cite{aghaee2025scaling} for fault tolerance) by post-selection on any PNR pattern if there is more than $3\%$ loss in the path to the PNR detectors.

\section{Conclusions}
In this work we have introduced the stellar decomposition of Gaussian quantum states, which leads to an exact and compact description of states heralded by photon-counting measurements, providing a versatile tool for both theoretical analysis and experimental simulation of non-Gaussian photonic circuits.

For pure bipartite Gaussian states, we proved that a local Gaussian unitary can always be factored out such that the residual core state has a finite Fock support on the unmeasured mode determined by the measurement outcome. For mixed states, we derived necessary and sufficient conditions for the existence of a physical decomposition, and in cases where such a decomposition is impossible, we formulated a semidefinite program to factor out as much of a Gaussian channel as possible. 

A particularly compelling application of our approach is in the characterization and simulation of GKP state preparation protocols. By relating the stellar decomposition to the effective squeezing of candidate states, we derived rigorous bounds on the quality of GKP states achievable via heralding, even in realistic scenarios with loss. These results not only simplify the simulation of complex photonic circuits, but also set practical limitations on state quality that are critical for designing fault-tolerant quantum information processors.

Looking ahead, our work opens several exciting avenues. Extensions to broader classes of non-Gaussian operations, optimized measurement strategies, and deeper investigations into the interplay between Gaussian noise extraction and non-Gaussian resource quality are natural next steps. We anticipate that the stellar decomposition will prove to be a valuable tool in bridging the gap between idealized theoretical models and realistic experimental implementations in quantum optics.

\section{Acknowledgements}
We acknowledge fruitful discussions with Rachel Chadwick and Jacob Hastrup. We thank Mahnaz Jafarzadeh and Gana\"el Roeland for providing comments on the manuscript. 

\bibliography{ref}

\onecolumngrid
\newpage
\appendix

\section{Bargmann representation}\label{app:bargmann}

The Bargmann function (also known as the stellar function) arises when we use the Bargmann basis of rescaled coherent states, e.g.
\begin{align}
    |\psi\rangle \rightarrow F_{|\psi\rangle}(z) = e^{\frac{|z|^2}{2}}\langle z^*|\psi\rangle,\quad z\in\mathbb{C}.
\end{align}
In this representation, the creation operator acts as multiplication by $z$, and the annihilation operator acts as a derivative with respect to $z$.
The Bargmann function can be extended to other objects such as mixed states, unitaries, channels, Kraus operators, etc. We can do so by using the Bargmann basis in all the Hilbert spaces spanned by a given object. For instance, the Bargmann function of a unitary operator $U$ is defined as:
\begin{align}\label{eq:Barg_U}
    F_{U}(z,w) = e^{\frac{|z|^2+|w|^2}{2}}\langle z^*|U|w\rangle.
\end{align}
The presence or absence of conjugates in the variables is purely convention, and it is typically adjusted such that the Bargmann $F$ function does not depend on conjugated variables.
The exponential rescaling makes the Bargmann function holomorphic, which leads to a series of benefits in terms of integrability, derivatives and relationship to other representations, as we shall see below.

In general the Bargmann function belongs to the Segal-Bargmann Hilbert space (here $z\in\mathbb{C}^n$):
\begin{align}
\mathcal{H}L^2(\mathbb{C}^n, \mu_t)=\biggl\{{F\in\mathcal{H}(\mathbb{C}^n)}\bigg|\int_{\mathbb{C}^n}|F(z)|^2\mu_t(z)d^nz < \infty\biggr\},
\end{align}
where $\mathcal{H}(\mathbb{C}^n)$ is the space of holomorphic functions from $\mathbb{C}^n$ to $\mathbb{C}$, the measure is $\mu_t(z) = \frac{1}{(\pi t)^n}e^{-\norm{z}^2/t}$ and $d^nz$ is shorthand for $d^n\Re(z)d^n\Im(z)$. We will choose $t=1$, but the choice $t=1/\pi$ is also common.
The inner product between two Bargmann functions belonging to the same Segal--Bargmann Hilbert space is given by
\begin{align}
    \langle F,G\rangle = \int_{\mathbb{C}^{n}}d^nz\mu_t(z)F(z)^*G(z).
\end{align}
We are also interested in inner products between functions defined on different Hilbert spaces, often over only a subset of the variables. As an example, consider an $m$-mode unitary $U$ acting on the first $m$ modes of an $(m+n)$-mode ket $|\psi\rangle$. The inner product then involves the integral of the product of the two corresponding Bargmann functions and the result is the Bargmann function of the state $(U\otimes I)|\psi\rangle$:
\begin{align}\label{eq:Upsi}
    F_{(U\otimes I)|\psi\rangle}(z,v) = \int_{\mathbb{C}^{m}}d\mu(w)F_U(z,w^*)F_{|\psi\rangle}(w, v),
\end{align}
with $z\in\mathbb{C}^m$ and $v\in\mathbb{C}^{n}$. We will see a few explicit examples below, and derive close-form expressions when working with Gaussian objects in particular.

\subsection{Conventions for ordering variables}\label{app:ordering-conventions}

Recall that the Bargmann function is defined using a basis of rescaled coherent states and that when the object it represents is defined on a tensor product of Hilbert spaces and/or dual Hilbert spaces, the Bargmann function has multiple variables. Therefore, one needs to specify a canonical ordering of these variables. Here we introduce the main conventions used throughout the manuscript. We order the variables in ``\emph{mode-wise}'' or ``\emph{type-wise}'' order, depending on whether we group the variables by mode or by type, i.e.~all types of variable for a given mode before moving on to the next mode, or all the variables of a given type in ascending mode order before moving on to the next type. We now go through examples for kets, density matrices, unitaries and channels, to see the orderings in action.

The Bargmann function of an $n$-mode pure state $\ket\psi$ is 
\begin{align}
    F_{\ket\psi}(z_1,\dots,z_n) = e^{\frac{|z_1|^2+\dots+|z_n|^2}{2}}\langle z_1^\ast,\dots,z_n^\ast | \psi\rangle, \quad \text{(type-wise or mode-wise order)}.
\end{align}
In this case there is only one type of variable and therefore the two orderings coincide.

For an $n$-mode density matrix $\rho$, we have $n$ variables on the ket side and another $n$ on the bra side, giving rise to two orderings. In the type-wise ordering we have
\begin{align}\label{eq:dm-bbkk}
F_\rho(w_1,\dots,w_n,z_1,\dots,z_n) = e^{\frac{\|z\|^2+\|w\|^2}{2}}\bra{z_1^*,\dots,z_n^*} \rho \ket{w_1,\dots,w_n}, \quad \text{(type-wise order)}
\end{align}
whereas in the mode-wise order we have
\begin{align}\label{eq:dm-bkbk}
F_\rho(w_1,z_1,\dots,w_n,z_n) = e^{\frac{\|z\|^2+\|w\|^2}{2}}\bra{z_1^*,\dots,z_n^*} \rho \ket{w_1,\dots,w_n}, \quad \text{(mode-wise order)}.
\end{align}

For a unitary $U$, in the type-wise order we have
\begin{align}
F_U(z_1,\dots,z_n,z'_1,\dots,z'_n) = e^{\frac{\|z\|^2+\|z'\|^2}{2}}\bra{z_1^*,\dots,z_n^*} U \ket{z'_1,\dots,z'_n}, \quad \text{(type-wise order)},
\end{align}
whereas in the mode-wise order we have
\begin{align}\label{eq:u-bkbk}
F_U(z_1,z'_1,\dots,z_n,z'_n) = e^{\frac{\|z\|^2+\|z'\|^2}{2}}\bra{z_1^*,\dots,z_n^*} U \ket{z'_1,\dots,z'_n}, \quad \text{(mode-wise order)}.
\end{align}
Note that we use $z$ and $z'$ rather than $z$ and $w$. This is a convention to indicate that a unitary acts on kets and returns other kets, so we put a prime on ``input'' variables and no prime of ``output'' variables.

Lastly, for an $n$-to-$n$ mode channel $\Phi$ we have four variables per mode: a bra and a ket variable at the input and a bra and a ket variable at the output. So we define
\begin{align}
\begin{split}
&F_\Phi(w_1,\dots,w_n,w_1',\dots,w_n',z_1,\dots,z_n,z_1',\dots,z_n') =\\
&\quad\quad e^{\frac{\|w\|^2 + \|z\|^2+\|w'\|^2 + \|z'\|^2}{2}} \bra{z^*}\Phi\left( \ket{z'} \bra{w'^\ast}\right) \ket{w},\quad \text{(type-wise order)},
\end{split}
\end{align}
whereas in the mode-wise order we have
\begin{align}
F_\Phi(w_1,w_1',z_1,z_1',\dots,w_n,w_n',z_n,z_n') = e^{\frac{\|w\|^2 + \|z\|^2+\|w'\|^2 + \|z'\|^2}{2}} \bra{z^\ast}\Phi\left( \ket{z'} \bra{w'^\ast}\right) \ket{w},\quad \text{(mode-wise order)}.
\end{align}
At times for channels we may want to introduce yet a different ordering, where we write all the output variables before the input ones: $F_\Phi(w_1,\dots,w_n,z_1,\dots,z_n,w_1',\dots,w_n',z_1',\dots,z_n')$, which we will refer to as the output-input order. Note that for a unitary, the output-input order coincides with the type-wise order.

So far it may look like specifying these orderings is just a pedantry, but in the next section we introduce Gaussian objects, whose parametrization does depend on the chosen ordering, and some operations and proofs in this paper may require one to switch between different orderings.

\subsection{Gaussian stellar functions and the Abc parametrization}\label{app:Abc-of-Gaussians}
A generic multivariable Gaussian complex function is the exponential of a quadratic polynomial:
\begin{align}\label{app:exp_abc}
    F(z) &= c \exp\left(\frac{1}{2}z^{\mathrm{T}} Az+z^{\mathrm{T}}b\right)\quad z\in\mathbb{C}^n,
\end{align}
where $A$ is a complex symmetric matrix (not necessarily hermitian) with eigenvalues in the unit disk, $b$ is a complex vector and $c$ is a complex scalar. We refer to $(A,b,c)$ as the ``Abc'' parametrization.

As an example to show how the orderings change the Abc parametrization, consider a two-mode mixed Gaussian state, which is defined by an $A$ matrix in type-wise order as
\begin{align}\label{eq:example-type-wise}
A_\rho = 
\begin{blockarray}{c cccc}
  & \substack{w_1 (\text{bra 1})} 
  & \substack{w_2 (\text{bra 2})} 
  & \substack{z_1 (\text{ket 1})} 
  & \substack{z_2 (\text{ket 2})} \\ 
\begin{block}{c(cccc)}
\substack{w_1 (\text{bra 1})}   & a_{11} & a_{12} & a_{13} & a_{14} \\
\substack{w_2 (\text{bra 2})}   & a_{21} & a_{22} & a_{23} & a_{24} \\
\substack{z_1 (\text{ket 1})}   & a_{31} & a_{32} & a_{33} & a_{34} \\
\substack{z_2 (\text{ket 2})}   & a_{41} & a_{42} & a_{43} & a_{44} \\
\end{block}
\end{blockarray},
\end{align}
while the same state in mode-wise order would have the two central rows and columns swapped:
\begin{align}\label{eq:example-mode-wise}
A_\rho = 
\begin{blockarray}{c cccc}
  & \substack{w_1 (\text{bra 1})} 
  & \substack{z_1 (\text{ket 1})} 
  & \substack{w_2 (\text{bra 2})} 
  & \substack{z_2 (\text{ket 2})} \\ 
\begin{block}{c(cccc)}
\substack{w_1 (\text{bra 1})}   & a_{11} & a_{13} & a_{12} & a_{14} \\
\substack{z_1 (\text{ket 1})}   & a_{31} & a_{33} & a_{32} & a_{34} \\
\substack{w_2 (\text{bra 2})}   & a_{21} & a_{23} & a_{22} & a_{24} \\
\substack{z_2 (\text{ket 2})}   & a_{41} & a_{43} & a_{42} & a_{44} \\
\end{block}
\end{blockarray}.
\end{align}

Associated with any Gaussian pure state, mixed state, unitary, channel, Kraus operator, and any object that can be obtained by compositions of such elements is a Gaussian stellar function. Such a unified parametrization grants complete flexibility when working with Gaussian stellar functions because it requires only one formula for the inner product that can then be used for unitary evolution, Gaussian measurements, tracing operations, and so on.

Note that for several kinds of objects, $c$ is not independent of $A$ and $b$. For example, for pure states we can require the norm to be 1, for mixed states we can require the trace to be 1, for unitaries we can require the product with the inverse to be the identity, and $c$ will follow from $A$ and $b$ via special cases of Eq.~\eqref{eq:gauss_intc}. In this paper we will omit $c$ if it can be derived from $A$ and $b$ and by knowing the nature of the object they parametrize.

\subsection{Relation to other representations}\label{app:representations}
In this subsection we lay out the connections between the Bargmann representation and the Fock, quadrature and phase-space representations.

\subsubsection{Fock representation}
The Bargmann and Fock representations are intimately connected, as the Taylor series of the Bargmann function (rescaled by a square-root factorial) consists precisely in the Fock amplitudes. For example,
\begin{align}
    |\psi\rangle=\sum_{n=0}^\infty\psi_n|n\rangle \rightarrow F_{|\psi\rangle}(z) = \sum_{n=0}^\infty \psi_n\frac{z^n}{\sqrt{n!}},
\end{align}
In particular, the Bargmann function of a number state is a monomial: $F_{|n\rangle}(z) = e^{|z|^2/2}\langle z^*|n\rangle = \frac{z^n}{\sqrt{n!}}$.
This implies that the Fock amplitudes can be computed in two ways: using derivatives
\begin{align}
    \langle n|\psi\rangle = \frac{\partial^n_z}{\sqrt{n!}}F_{|\psi\rangle}(z)|_{z=0},
\end{align}
or using integrals
\begin{align}
    \langle n|\psi\rangle = \int_\mathbb{C}\frac{d^2z}{\pi}e^{-|z|^2}\frac{z^n}{\sqrt{n!}}F_{|\psi\rangle}(z).
\end{align}

In the case of Gaussian objects their Bargmann function is the exponential of a quadratic polynomial, and so the rescaled Taylor series $G_k$ (with $k\in\mathbb{N}^N$) of an $N$-variables Gaussian object parametrized by $A$, $b$, $c$ as in Eq.~\eqref{app:exp_abc} can be found recursively \cite{miatto2019recursive, miatto2020fast, yao2024riemannian}:
\begin{align}\label{eq:recrel}
G_{k+1_i} = \frac{1}{\sqrt{k_i+1}}\left[b_iG_k + \sum_{j=1}^N\sqrt{k_j}A_{ij}G_{k-1_j}\right],
\end{align}
or using the more numerically stable (albeit slower) version
\begin{align}\label{eq:recrel_stable}
G_{k} = \frac{1}{|\{i|k_i>0\}|}\sum_{i|k_i>0}\frac{1}{\sqrt{k_i}}\left[b_iG_{k-1_i} + \sum_{j=1}^N\sqrt{k_j-
\delta_{ij}}A_{ij}G_{k-1_j-1_i}\right],
\end{align}
with $G_0 = c$ and where $1_i$ is the vector of zeros with a 1 at index $i$. Note that this recursive formula is independent of the nature of the object, as long as its Bargmann function has the functional form in Eq.~\eqref{eq:exp_quad_poly}.

It is possible to apply the recurrence relation on a subset of lattice points, and produce a subset of Fock amplitudes. This can be more efficient than calculating all the amplitudes up to a certain cutoff on all the modes \cite{de2023quadratic, yao2021fast}.

\subsubsection{Quadrature representation}
There are two natural ways to go from the Bargmann representation to the quadrature representation.
The first consists in the observation that the integral of the Bargmann function along an axis yields the wave function (using $z = x+ip$):
\begin{align}
\begin{split}
\sqrt{2}(\pi^3\hbar)^{1/4}e^{\frac1{4\hbar}x^2}\psi\left(\frac{x}{\sqrt{2\hbar}}\right) &= \int_\mathbb{R}dp\,e^{-\frac{1}{2}p^2}F_{|\psi\rangle}(x+ip),\\
\sqrt{2}(\pi^3\hbar)^{1/4}e^{\frac{1}{4\hbar}p^2}\tilde\psi\left(-\frac{p}{\sqrt{2\hbar}}\right) &= \int_\mathbb{R}dx\,e^{-\frac{1}{2}x^2}F_{|\psi\rangle}(x+ip),
\end{split}
\end{align}
and this holds for any intermediate direction, in which case one obtains the wavefunction on the quadrature orthogonal to the integration axis.

Alternatively, one can use an integral transform with an appropriate kernel written in the Bargmann representation. This allows us to use yet again the Abc formulation. For example an $n$-mode ket with Bargmann function $F_{|\psi\rangle}(z)$ where $z\in\mathbb{C}^n$ has a wavefunction along the quadrature $\lambda\in\mathbb{R}^n$ at angle $\phi$ with the position axis given by:
\begin{align}
    \psi_\phi(\lambda) = \int_{\mathbb{C}^n}\frac{d^{2n}z}{\pi^n}e^{-\|z\|^2}K_\phi(\lambda,z)F_{|\psi\rangle}(z),
\end{align}
where the kernel is the exponential of a quadratic polynomial and therefore admits an Abc parametrization: \begin{align}
    K_\phi(\lambda,z) = \frac{1}{(\pi\hbar)^{n/4}}e^{\frac{1}{2}(\begin{smallmatrix}
    \lambda& z
\end{smallmatrix})A_\phi\left(\begin{smallmatrix}
    \lambda\\ z
\end{smallmatrix}\right)}.\end{align}
The Abc parametrization of the kernel is
\begin{align}
\begin{split}
A_\phi&=\begin{bmatrix}
-\frac{\mathbb{1}}{\hbar} & e^{-i\phi}\sqrt{\frac{2}{\hbar}}\mathbb{1}\\
e^{-i\phi}\sqrt{\frac{2}{\hbar}}\mathbb{1} & -e^{-2i\phi}\mathbb{1}
\end{bmatrix},\\
b_\phi &= 0,\\
c_\phi &= \frac{1}{(\pi\hbar)^{n/4}}.
\end{split}
\end{align}
An advantage of this formulation is that we can decide which subset of variables to transform to other representations. For example, one could start with a two-mode ket $F_{|\psi\rangle}(z_1,z_2)$, apply the integral transform with $\phi=0$ to the first mode obtaining a new function $F'_{|\psi\rangle}(x_1,z_2)$, and then transform the second mode to the Fock representation $F''(x_1|k)=\frac1{\sqrt{k!}}\partial_{z_2}^kF'_{|\psi\rangle}(x_1,z_2)|_{z_2=0}$ yielding position wave functions on the first mode conditional on projecting the second mode onto number states $|k\rangle$.

\subsubsection{Phase-space representations}
Analogously to the case of the wavefunction, one can transform the Bargmann function to the phase-space functions (Wigner, Husimi, Glauber, and their characteristic functions) by integrating against kernels that are themselves exponentials of a quadratic polynomial---therefore admitting an Abc parametrization. This means that the phase-space functions of Gaussian objects can be computed with a Gaussian integral and expressed as exponentials of a quadratic polynomial. However, note that contrary to the Bargmann function, the phase-space functions, being real and non-constant, are not holomorphic.

The $s$-parametrized kernel $\Delta_s$ that underpins the transformation from the Bargmann representation to the phase-space functions is known as the Stratonovich-Weyl kernel \cite{brif1998general}. The map from $\rho$ to the $s$-parametrized phase-space functions is then:
\begin{align}
    W_s(z^*,z) = \mathrm{Tr}[\rho\,\Delta_s(z^*,z)] = \int_{\mathbb{C}^{2n}}\frac{d^{2n}y}{\pi^{n}}\frac{d^{2n}w}{\pi^{n}}e^{-\|y\|^2-\|w\|^2}F_\rho(w,y)F_{\Delta_s}(z,z^\ast, w,y)^*,
\end{align}
such that for $s=0$ we obtain the Wigner function, for $s=1$ we obtain the Glauber $P$ function and for $s=-1$ we obtain the Husimi $Q$ function. The Abc triple of the Stratonovich-Weyl kernel in output-input order is:
\begin{align}
\begin{split}
A_{\Delta_{s}} &= \frac{2}{s-1}\begin{bmatrix}
    X & -\mathbb 1\\
    -\mathbb 1 & \frac{s+1}{2}X
\end{bmatrix},\\
b_{\Delta_{s}}&=0,\\
c_{\Delta_{s}}&=\frac{2}{\pi^n\abs{s-1}^n},
\end{split}
\end{align}
where 
\begin{align*}
X =\begin{bmatrix}
    0&\mathbb{1}\\\mathbb{1}&0
\end{bmatrix}.
\end{align*}
We highlight that the Glauber $P$ function of states can be even more singular than a delta function, and therefore the appearance of $|s-1|$ in the denominator of $A_{\Delta_s}$ and $c_{\Delta_s}$ is not surprising. Note that in order to obtain the usual Wigner function defined over the position and momentum coordinates we need a further transformation of the coordinates from $(z^*,z)$ to $(q,p)$, and a rescaling by $\sqrt{2\hbar}$ per mode. In other words, we have
\begin{align}
W^{PS}_s(x,p) = \frac{1}{(2\hbar)^n}W_s(z^\ast,z)
\end{align}
with $z = \frac1{\sqrt{2\hbar}}(x+ip)$, where the superscript `PS' is shorthand for `phase space.' For instance, the function $W^{PS}_0$ is the standard phase-space Wigner function.

Closely related is the complex Fourier transform version which maps from the Bargmann representation to the $s$-parametrized characteristic functions:
\begin{align}
    \chi_s(z^*,z) = \mathrm{Tr}[\rho \,T_s(z^*,z)] = \int_{\mathbb C^{2n}} \frac{\mathrm d^{2n} y}{\pi^n}\frac{\mathrm d^{2n} w}{\pi^n} e^{-\norm{y}^2 - \norm{w}^2} F_\rho(w,y) F_{T_s}(z,z^\ast,w,y)^\ast,
\end{align}
whose Abc parametrization is given by
\begin{align}
\begin{split}
    A_{T_{s}} &= \begin{bmatrix}
        \frac{s-1}{2}X&\Omega^T\\
        \Omega&X
    \end{bmatrix},\\
    b_{T_{s}} &=0,\\
    c_{T_{s}} &=1,
\end{split}
\end{align}
where $\Omega = \left(\begin{smallmatrix}
    0&\mathbb{1}\\-\mathbb{1}&0
\end{smallmatrix}\right)$.

Note that these integrals can also be mathematically carried out with $s$ taking non-integer values.

\subsection{Stellar function of Gaussian objects}
Let us now see a few examples of the Abc parametrization of common Gaussian objects.

\subsubsection{Mixed states}
The Abc parametrization of an $n$-mode Gaussian state $\rho$ is related to the covariance matrix $\Sigma$ and mean vector $\mu$ in the $q/p$ basis via the Husimi covariance $Q$ and means $\beta$. These are defined as $Q=W\left(\frac{1}{\hbar}\Sigma+\frac{\mathbb{1}}{2}\right)W^\dagger$ and $\beta=\frac{1}{\sqrt{\hbar}}W\mu$ where 
\begin{align}\label{eq:W}
W=\frac{1}{\sqrt{2}}\begin{bmatrix}
    \mathbb{1} & \mathbb{1}i\\\mathbb{1}&-\mathbb{1}i
\end{bmatrix}
\end{align}
is the unitary that rotates the $q/p$ basis into the amplitude basis. Then we can define
\begin{align}\label{eq:rho-from-PS}
    A_\rho &= X(\mathbb{1}-Q^{-1})=X\left[\mathbb{1}-W\left(\frac{1}{\hbar}\Sigma+\frac{\mathbb{1}}{2}\right)^{-1}W^\dagger\right]\\
    b_\rho &=XQ^{-1}\beta=\frac{1}{\sqrt{\hbar}}XW\left(\frac{1}{\hbar}\Sigma+\frac{\mathbb{1}}{2}\right)^{-1}\mu\\
    c_\rho &=\frac{e^{-\frac12\beta^\dagger Q^{-1}\beta}}{\sqrt{\det(Q)}}=\frac{e^{-\frac{1}{2\hbar}\mu^{\mathrm{T}} \left(\frac{1}{\hbar}\Sigma+\frac{\mathbb{1}}{2}\right)^{-1}\mu}}{\sqrt{\det\left(\frac{1}{\hbar}\Sigma+\frac{\mathbb{1}}{2}\right)}}
\end{align}
where $X = \left[\begin{smallmatrix}
    0& \mathbb{1}\\\mathbb{1}&0
\end{smallmatrix}\right]$.
Note that the covariance matrix of an $n$-mode state has $2n$ rows and columns whether the state is pure or mixed. In contrast, the size of the Abc parametrization depends on the nature of the object it represents, so $A_\rho$ has $2n$ rows and columns, but $A_{|\psi\rangle}$ has only $n$ of them.

\subsubsection{Unitaries}
Let us now turn to the Abc parameterization of Gaussian unitaries. Recall that a Gaussian unitary, say $U$, represented by the symplectic transformation $S$, corresponds to the phase-space transformation
\begin{align}
U: \vec v \mapsto S \vec v + \vec d.
\end{align}
Recall that for an $n$-mode unitary, the corresponding symplectic transforms form the symplectic group $\mathrm{Sp(2n, \mathbb R)}$,  defined as
\begin{align}
\mathrm{Sp}(2n,\mathbb R):=\{M\in\mathbb R^{2n\times 2n}| M^{\mathrm{T}} \Omega M = \Omega\},
\end{align}
with 
\begin{align}
\Omega = \begin{bmatrix}
0 & \mathbb 1_n\\
-\mathbb 1_n & 0
\end{bmatrix}.
\end{align}
The following proposition shows how we can efficiently transform the Bargmann representation of a unitary to its symplectic form. Moreover, as the map can be readily inverted, we get also a simple map from the symplectic representation to the Bargmann parametrization. 
\begin{proposition}\label{prop:AofU}
Consider a Gaussian unitary $U$ with the symplectic matrix
\begin{align}
S = W \begin{bmatrix}
S_1 & S_2\\
S_2^\ast & S_1^\ast
\end{bmatrix}
W^\dag,
\end{align}
and the displacement part
\begin{align}
\begin{bmatrix}
\vec\gamma\\
\vec \gamma^\ast
\end{bmatrix}
=
W
\begin{bmatrix}
\vec x\\
\vec p
\end{bmatrix},
\end{align}
with $W$ as defined in \eqref{eq:W}.
Then, the Abc parametrization of this unitary is\footnote{We note that the phase of $c_U$ cannot be inferred from the phase-space description of the unitary alone, as $U$ and $e^{i\phi} U$ have the same phase space effect. However, this phase is calculated in \cite{barnett2002methods, quesada2025s}, given the quadratic Hamiltonian generating the dynamics.}
\begin{align}\label{eq:Abc-of-U}
\begin{split}
    A_U &= \begin{bmatrix}S_2 S_1^\ast {}^{-1} & S_1^\dag {}^{-1}\\
    S_1^\ast {}^{-1} & -S_1^{\ast} {}^{-1} S_2^\ast
    \end{bmatrix}\\
    b_U &=\begin{bmatrix}
        -S_2 S_1^\ast {}^{-1} \gamma^\ast + \gamma\\
        -S_1^\ast {}^{-1} \gamma^\ast 
    \end{bmatrix}\\
    |c_U| &= \frac{\exp(-\frac{1}{2\hbar}d^{\mathrm{T}} (\mathbb 1 + SS^{\mathrm{T}})^{-1}d)}{\left[\mathrm{det}\left(\frac12 (SS^{\mathrm{T}}+\mathbb 1)\right)\right]^{1/4}}.
\end{split}
\end{align}
\end{proposition}
\begin{proof}
We show this through the identities
\begin{align}\label{eq:a-transformation-under-U}
\begin{split}
a_i  U = \sum_{j} (S_1)_{ij} U a_j + (S_2)_{ij} U  a^\dagger_j +\gamma_i U,\\
a^\dagger_i  U = \sum_{j} (S_2^\ast)_{ij} U a_j + (S_1^\ast)_{ij} U  a^\dagger_j +\gamma^\ast_i U,
\end{split}
\end{align}
which are essentially the formula for evolving annihilation operators under Gaussian transformations.
Defining the unitary's Bargmann function as \eqref{eq:Barg_U}, we can rewrite \eqref{eq:a-transformation-under-U} as
\begin{align}\label{eq:computation-of-Au}
\begin{split}
\partial_{z_i} F_U(z,w) &= \sum_j (S_1)_{ij} w_j F_U(z,w) + \sum_j (S_2)_{ij} \partial_{w_j} F_U(z,w) + \gamma_i F_U(z,w),\\
z_i F_U(z,w) &= \sum_j (S_2)^\ast_{ij} w_j F_U(z,w) + \sum_j (S_1)^\ast_{ij} \partial_{w_j} F_U(z,w) + \gamma_i^\ast F_U(z,w).
\end{split}
\end{align}
Expanding the above equation using the Gaussian expression for $F_U$ gives Ab provided in \eqref{eq:Abc-of-U}. To elaborate, note that
\begin{align}\label{eq:partials-of-Fu}
\begin{split}
\partial_{z_i} F_U(z,w) &= c \cdot \left( (b_1)_i + (A_1 z)_i + (A_2^{\mathrm{T}} w)_i \right) F_U(z,w)\\
\partial_{w_i} F_U(z,w) &= c \cdot \left( (b_2)_i + (A_2 z)_i + (A_3 w)_i \right) F_U(z,w),
\end{split}
\end{align}
where we have used blocks of Ab parametrization of $U$ as
\begin{align}
A_U = 
\begin{bmatrix}
A_1 & A_2^{\mathrm{T}}\\
A_2 & A_3
\end{bmatrix},
\quad b_U = \begin{bmatrix}
b_1\\
b_2
\end{bmatrix}.
\end{align}
Plugging this into \eqref{eq:computation-of-Au}, gives
\begin{align}\label{eq:generating-S2Au}
\begin{split}
b_1 + A_1 z + A_2^{\mathrm{T}} w &= S_1 w + S_2 (b_2 + A_2 z + A_3 w) + \gamma\\
z &= S_2^\ast w + S_1^\ast (b_2 + A_2 z + A_3 w) + \gamma^\ast.
\end{split}
\end{align}
As the above set of equations must hold for all $z,w\in\mathbb C^n$, we get
\begin{align}\label{eq:equations-for-S2Au}
\begin{split}
\mathbb 1 &= S_1^\ast A_2 \Rightarrow A_2 = S_1^\ast {}^{-1},\\
0 &= S_2^\ast + S_1^\ast A_3 \Rightarrow A_3 = -S_1^\ast {}^{-1} S_2^\ast
\end{split}
\end{align}
Furthermore, we have
\begin{align}
A_1 = S_2 A_2 \overset{(i)}{=} S_2 S_1^\ast {}^{-1},
\end{align}
where $(i)$ follows from \eqref{eq:equations-for-S2Au}. Finally, note that \eqref{eq:generating-S2Au} also gives
\begin{align}\label{eq:b2-in-S2Au}
0 = S_1^\ast b_2 + \gamma^\ast \Rightarrow b_2 = -S_1^\ast {}^{-1} \gamma^\ast,
\end{align}
and also
\begin{align}
b_1 = S_2 b_2 + \gamma \overset{(ii)}{=} -S_2 S_1^\ast {}^{-1} \gamma^\ast +\gamma,
\end{align}
where $(ii)$ follows from \eqref{eq:b2-in-S2Au}.

Finally, the expression for $c$ is obtained from unitarity (i.e., computing c of $U U^\dagger$ in the Bargmann representation, using contraction formulas of \protect\cref{app:inner-product-in-bargmann}). 

Note that the inverse map of \eqref{eq:Abc-of-U} is readily constructible. In other words, given Abc triple of a Gaussian unitary, it is straightforward to use \eqref{eq:Abc-of-U} and compute the symplectic representation.
\end{proof}

\subsubsection{Channels}
As obtained in \cite[Appendix B]{yao2024riemannian}, for Gaussian channels we have
\begin{align}
\begin{split}
    A_\Phi &= X {L}  \begin{bmatrix}
        \mathbb{1}_{2m} -  {\xi}^{-1} &   {\xi}^{-1} {X_\Phi}  \\
        X_\Phi^{\mathrm{T}} {\xi}^{-1} & \mathbb{1}_{2m} -  {X_\Phi}^{\mathrm{T}} {\xi}^{-1} {X_\Phi}
    \end{bmatrix} {L}^\dagger, \\
    &=X{L}\left(\mathbb{1}_{4m}-\begin{bmatrix}
    {\xi}^{-1} &   -{\xi}^{-1} {X_\Phi}  \\
        -{X_\Phi}^{\mathrm{T}} {\xi}^{-1} & {X_\Phi}^{\mathrm{T}} {\xi}^{-1} {X_\Phi}
    \end{bmatrix}\right){L}^\dagger,\\
    b_\Phi &= \frac{1}{\sqrt{\hbar}}  {L}^* \begin{bmatrix} {\xi}^{-1} {d}  \\ - {X_\Phi}^{\mathrm{T}}{\xi}^{-1} {d}   \end{bmatrix},\\
    c_\Phi &= \frac{\exp\left[ -\tfrac{1}{2\hbar} {d}^{\mathrm{T}} {\xi}^{-1} {d} \right]}{\sqrt{\det({\xi})}},
\end{split}
\end{align}
where $X = \left[\begin{smallmatrix}
    0_{m} & \mathbb{1}_m \\
    \mathbb{1}_m & 0_m 
\end{smallmatrix} \right] $ and
\begin{align}\label{eq:Rdef}
    {L} = \frac{1}{\sqrt{2}}\left[\begin{array}{cccc} \mathbb{1}_m & i \mathbb{1}_m & 0_m & 0_m \\ 0_m & 0_m & \mathbb{1}_m & -i \mathbb{1}_m \\ \mathbb{1}_m & -i \mathbb{1}_m & 0_m & 0_m \\ 0_m & 0_m & \mathbb{1}_m & i \mathbb{1}_m \\\end{array}\right], \quad \xi = \frac{1}{2}\left(\mathbb{1}_{2m} + {X_\Phi}{X_\Phi}^{\mathrm T} + \frac{2 {Y_\Phi}}{\hbar} \right).
\end{align}
and $X_\Phi,Y_\Phi,d_\Phi$ represent the phase-space transformation, meaning that, assuming a Gaussian state $\rho$ with covariance matrix $\Sigma_\rho$ and mean $\mu_\rho$, we have that the covariance matrix and vector of means for $\Phi[\rho]$ is given by
\begin{align}
\Sigma_{\Phi[\rho]} = X_\Phi \Sigma_\rho X_\Phi^{\mathrm{T}} + Y_\Phi, \quad \mu_{\Phi[\rho]} = X_\Phi \mu_\rho + d_\Phi.
\end{align}
One can invert this map, to convert the Abc parametrization to the $X_\Phi,Y_\Phi$ representation. This is a new result, that we summarize in the following proposition.
\begin{proposition}\label{prop:AtoXY-for-channels}
Consider a Gaussian channel with A matrix given as
\begin{align}
A_\Phi = \begin{bmatrix}
A^{\mathrm{out}}_\Phi & \Gamma_\Phi^{\mathrm{T}}\\
\Gamma_\Phi & A^{\mathrm{in}}_\Phi
\end{bmatrix}.
\end{align}
It is the case that
\begin{align}\label{eq:XY}
X_\Phi &= W^\dagger (\mathbb 1 - XA^{\mathrm{out}}_\Phi)^{-1} X \Gamma_\Phi^{\mathrm{T}} X W,\\
\frac{1}{\hbar}Y_\Phi &= W^\dag ( \mathbb 1 - XA_\Phi^{\mathrm{out}})^{-1} W - \frac12 X_\Phi X_\Phi^{\mathrm{T}},
\end{align}
with $W$ defined in \eqref{eq:W}.
\end{proposition}
\begin{proof}
Note that we are guaranteed the invertibility of $\mathbb 1 - XA^{\mathrm{out}}_\Phi$ from the physicality constraints (\protect\cref{rem:physicality}). To obtain \eqref{eq:XY}, we provide two separate proofs. 

\paragraph{First proof} We apply a perturbative approach. Note that the covariance matrix $\Sigma$ corresponding to a state can be written as (see \eqref{eq:rho-from-PS})
\begin{align}
\frac{1}{\hbar}\Sigma =  W^\dag (\mathbb 1 - X A)^{-1} W - \frac{1}{2} \mathbb 1.
\end{align}
We use perturbations to find the $X_\Phi$ matrix. Let $A$ correspond to the $A$ matrix of a state $\rho$ that is close to vacuum i.e., $\norm{A}\ll 1$. We have
\begin{align}
\frac{1}{\hbar}\Sigma = \frac12 \mathbb 1 + W^\dagger X A W + O(\norm{A}^2).
\end{align}
Moreover, for small $A$ we have that if $\rho' = \Phi[\rho]$, then
\begin{align}
A' = A^{\mathrm{out}}_\Phi + \Gamma^{\mathrm{T}}_\Phi A \Gamma_\Phi+O(\norm{A}^2).
\end{align}
This concludes that
\begin{align}
\begin{split}
\frac{1}{\hbar}\Sigma' &= W^\dagger ( \mathbb 1 - X A_{\Phi}^{\mathrm{out}} -X \Gamma^{\mathrm{T}}_\Phi A \Gamma_\Phi)^{-1} W - \frac12 \mathbb 1 + O(\norm{A}^2),\\
&= W^\dag\left( (\mathbb 1 -X  A_\Phi^{\mathrm{out}})^{-1} + (\mathbb 1 - X  A_\Phi^{\mathrm{out}})^{-1} X\Gamma^{\mathrm{T}}_\Phi A \Gamma_\Phi (\mathbb 1 - X  A_\Phi^{\mathrm{out}})^{-1}\right) W + O(\norm{A}^2),\\
&= \frac{1}{\hbar}Y_\Phi + \left[ W^\dag (\mathbb 1 - X  A_\Phi^{\mathrm{out}})^{-1}  X  \Gamma^{\mathrm{T}}_\Phi  X W\right] \frac{\Sigma}{\hbar} \left[ W^\dag \Gamma_\Phi (\mathbb 1 -  X A_\Phi^{\mathrm{out}})^{-1} W\right] + O(\norm{A}^2),
\end{split}
\end{align}
where
$
\frac{1}{\hbar}Y_\Phi = W^\dag (\mathbb 1 - X A^{\mathrm{out}}_\Phi)^{-1} W - \frac12 X_\Phi X_\Phi^{\mathrm{T}}
$
is the zeroth-order term. Now, note that $\overline{W^\dag} = W^\dag X$, and $(\mathbb 1 - X A_\Phi^{\mathrm{out}})^{-1} {}^{\mathrm{T}} = (\mathbb 1-  A^{\mathrm{out}}_\Phi X)^{-1} =  X(\mathbb 1 - X  A^{\mathrm{out}}_\Phi)^{-1} X$, which leads to
\begin{align}\label{eq:X_phit}
W^\dag \Gamma_\Phi (\mathbb 1 - X A^{\mathrm{out}}_\Phi)^{-1} W = X_\Phi^{\mathrm{T}},
\end{align}
and therefore
\begin{align}
\Sigma' =  Y_\Phi +  X_\Phi \Sigma X_\Phi^{\mathrm{T}} + O(\norm{A}^2),
\end{align}
which is the formula for channel transformation in the limit where the higher-order term vanishes.

\paragraph{Second proof} Let $D(\alpha) = \exp(\alpha a^\dagger - \alpha^\ast a)$ and note that
\begin{align}
\Phi^\dagger(D(\alpha)) \propto D(\alpha'),
\end{align}
where $\alpha = (x+ip)/\sqrt{2}$ and, $\alpha' = (x'+ip')/\sqrt2$ are related via \cite[Eq. (5.55)]{serafini2017quantum} 
\begin{align}
\begin{bmatrix}
x'\\
p'
\end{bmatrix} = -\Omega X^{\mathrm{T}}_\Phi \Omega \begin{bmatrix}
x\\
p
\end{bmatrix}.
\end{align}
Using the fact that $D(\alpha)$ is parametrized by Ab parameters
\begin{align}
A_{D(\alpha)}= X, \quad b_{D(\alpha)} = \begin{bmatrix}
\alpha\\
-\alpha^\ast
\end{bmatrix} = Z W \begin{bmatrix}
x\\
p
\end{bmatrix},
\end{align}
with $Z = \begin{bmatrix}
\mathbb 1 & 0\\
0 & -\mathbb 1
\end{bmatrix}$. Note that this parametrization can be readily obtained from \protect\cref{prop:AofU}. Therefore, we get
\begin{align}\label{eq:b-transformation-via-serafini}
b_{D(\alpha')} = -W X_{\Phi}^{\mathrm{T}} W^\dagger b_{D(\alpha)},
\end{align}
where in doing so, we have used the identity $ZW\Omega = i W$. Alternatively, we can use inner product formulas in Bargmann representation (see \protect\cref{app:inner-product-in-bargmann}) to obtain that the b vector corresponding to $\Phi^\dagger(D(\alpha))$ is
\begin{align}\label{eq:b-transformation}
\begin{split}
b_{D(\alpha')} &= -\begin{bmatrix}
0 & \Gamma_{\Phi}
\end{bmatrix}
\begin{bmatrix}
X & -\mathbb 1\\
-\mathbb 1 & A_\Phi^{\mathrm{out}}
\end{bmatrix}^{-1}
\begin{bmatrix}
b_{D(\alpha)}\\
0
\end{bmatrix}\\
&= -\Gamma_\Phi (\mathbb 1 - X A_{\Phi}^{\mathrm{out}})^{-1} b_{D(\alpha)}.
\end{split}
\end{align}
Combining \eqref{eq:b-transformation-via-serafini} with \eqref{eq:b-transformation}, we get
\begin{align}
X_{\Phi}^{\mathrm{T}} = W^\dagger \Gamma (\mathbb 1 - X A_{\Phi}^{\mathrm {out}})^{-1} W,
\end{align}
which is identical to \eqref{eq:X_phit}, and hence, we have obtained the same result. For the $Y_{\Phi}$ matrix, we can simply send a vacuum state inside the channel and compute its A matrix via Bargmann inner products and from there, compute the covariance matrix. The resulting covariance matrix must be equal to $\frac\hbar2X_\Phi X_\Phi^{\mathrm{T}} + Y_\Phi$, and from there, we get $Y_\Phi$. As this is equivalent to calculations performed for computation of $Y_\Phi$ in our first proof, we do not repeat it.
\end{proof}

\subsection{Inner product formulas}\label{app:inner-product-in-bargmann}

One of the many advantages of the Bargmann representation is that all inner products between Gaussian objects are defined through a complex Gaussian integral, regardless of the nature of the objects involved. In this section, we provide a general formula for computing such inner products and we give some examples.

As an example of combining two Gaussian objects, we consider the action of an $m$-mode unitary on the first $m$ modes of an $(m+n)$-mode ket. Then the integral is to be computed as follows:
\begin{align}\label{eq:Upsi_integral}
F_{(U\otimes I)|\psi\rangle}(z,v) &=\int_{\mathbb{C}^m}\frac{d^{2m}w}{\pi^m}e^{-\|w\|^2}F_U(z,w^*)F_{|\psi\rangle}(w,v)\\
&=\int_{\mathbb{C}^m}\frac{d^{2m}w}{\pi^m}e^{\frac12\left(\begin{smallmatrix}w^*\\w\end{smallmatrix}\right)^{\mathrm{T}}\left(\begin{smallmatrix}0&-\mathbb{1}\\-\mathbb{1}&0\end{smallmatrix}\right)\left(\begin{smallmatrix}w^*\\w\end{smallmatrix}\right)}c_U\exp\left[\frac{1}{2}\left(\begin{smallmatrix}z\\w^*\end{smallmatrix}\right)^{\mathrm{T}}A_U\left(\begin{smallmatrix}z\\w^*\end{smallmatrix}\right)+ \left(\begin{smallmatrix}z\\w^*\end{smallmatrix}\right)^{\mathrm{T}}b_U\right]\\
\nonumber&\quad\quad\quad\quad\quad\quad\quad\quad\quad\quad\quad\quad\quad\quad\quad\times c_{|\psi\rangle}\exp\left[\frac{1}{2}\left(\begin{smallmatrix}w\\v\end{smallmatrix}\right)^{\mathrm{T}}A_{|\psi\rangle}\left(\begin{smallmatrix}w\\v\end{smallmatrix}\right)+ \left(\begin{smallmatrix}w\\v\end{smallmatrix}\right)^{\mathrm{T}}b_{|\psi\rangle}\right]\\\label{eq:gaussian_integral_final}
&=c_Uc_{|\psi\rangle}\int_{\mathbb{C}^m}\frac{d^{2m}w}{\pi^m}\exp\left[\frac{1}{2}\left(\begin{smallmatrix}z\\w^*\\w\\v\end{smallmatrix}\right)^{\mathrm{T}}A\left(\begin{smallmatrix}z\\w^*\\w\\v\end{smallmatrix}\right)+ \left(\begin{smallmatrix}z\\w^*\\w\\v\end{smallmatrix}\right)^{\mathrm{T}}b\right],
\end{align}
where the $A$ matrix and $b$ vector can be expressed in block form:
\begin{align}\label{eq:Abc_Upsi}
    A = \begin{bmatrix}
        A^\mathrm{out}_U & \Gamma_U^{\mathrm{T}} & 0 & 0\\
        \Gamma_U & A^\mathrm{in}_U & -\mathbb{1} & 0\\
        0 & -\mathbb{1} & A_{|\psi\rangle}^w & \Gamma_{|\psi\rangle}^{\mathrm{T}}\\
        0 & 0 & \Gamma_{|\psi\rangle} & A_{|\psi\rangle}^{v}
    \end{bmatrix},
    \quad
    b = \begin{bmatrix}
        b^\mathrm{out}_U\\
        b^\mathrm{in}_U\\
        b_{|\psi\rangle}^w\\
        b_{|\psi\rangle}^v
    \end{bmatrix}.
\end{align}
Note that we incorporated the term $e^{-||w||^2}$ from the measure into the $A$ matrix through the off-diagonal identity blocks.

Eq.~\eqref{eq:gaussian_integral_final} is a well-known complex Gaussian integral with leftover variables. The solution is yet another exponential of a quadratic polynomial with a new triple of Abc parameters:
\begin{align}\label{eq:gauss_intA}
    A &= \begin{bmatrix}
        A^\mathrm{out} & 0\\
        0 & A_{|\psi\rangle}^v
    \end{bmatrix}-\begin{bmatrix}\Gamma_U^{\mathrm{T}}&0\\0&\Gamma_{|\psi\rangle}\end{bmatrix}M^{-1}\begin{bmatrix}\Gamma_U&0\\0&\Gamma_{|\psi\rangle}^{\mathrm{T}}\end{bmatrix},\\\label{eq:gauss_intb}
    b &= \begin{bmatrix}b^\mathrm{out}\\b_{|\psi\rangle}^v\end{bmatrix}-\begin{bmatrix}\Gamma_U^{\mathrm{T}}&0\\0&\Gamma_{|\psi\rangle}\end{bmatrix}M^{-1}\begin{bmatrix}b^\mathrm{in}_U\\b_{|\psi\rangle}^w\end{bmatrix},\\\label{eq:gauss_intc}
    c &= c_Uc_{|\psi\rangle}\frac{\exp(-\frac{1}{2}
    \begin{bmatrix}b^\mathrm{in}_U&b_{|\psi\rangle}^w
    \end{bmatrix}
    M^{-1}
    \begin{bmatrix}
    b^\mathrm{in}_U\\b_{|\psi\rangle}^w\end{bmatrix})}{\sqrt{\det(iM)}},
\end{align}
where $M = \begin{bmatrix}A_U^\mathrm{in}&-\mathbb{1}\\-\mathbb{1}&A^w_{|\psi\rangle}\end{bmatrix}$ is the block that includes the off-diagonal identities coming from the measure. 

All inner products between two Gaussian objects take this form, with the appropriate blocks replaced. The only ``exception'' is the tracing operation, which does not involve two initial Bargmann functions, but rather only one. In that case one can start directly by subtracting the off-diagonal identities from the $A$ matrix of the Bargmann function on the corresponding integration variables, and then proceed with the Gaussian integral. As an example, we show how to calculate the partial trace of a two-mode Gaussian density matrix (note we use mode-wise order):
\begin{align}
\begin{split}
F_{\mathrm{Tr}_0[\rho]}(w,v) &= \int_\mathbb{C}\frac{d^2z}{\pi}e^{-\|z\|^2}F_\rho(z^*,z,w,v)\\
&=c_\rho\int_\mathbb{C}\frac{d^2z}{\pi}\exp{\biggl[\frac12\left(\begin{smallmatrix}
    z^*\\z\\w\\v
\end{smallmatrix}\right)^{\mathrm{T}}\begin{bmatrix}
    A_0-X&R\\R^{\mathrm{T}}&A_1
\end{bmatrix}\left(\begin{smallmatrix}
    z^*\\z\\w\\v
\end{smallmatrix}\right) + \left(\begin{smallmatrix}
    z^*\\z\\w\\v
\end{smallmatrix}\right)^{\mathrm{T}}\begin{bmatrix}
    b_0\\b_1
\end{bmatrix}\biggr]}
\end{split}
\end{align}
So $M=A_0-X$ and we can apply the standard formulas for Gaussian integrals to obtain the triple that parametrizes the result:
\begin{align}\label{eq:inner-prods}
\begin{split}
    A &= A_1 - R^{\mathrm{T}}(A_0-X)^{-1}R\\
    b &= b_1 - R^{\mathrm{T}}(A_0-X)^{-1}b_0\\
    c &= c_\rho\frac{\exp[-\frac12b_0^{\mathrm{T}}(A_0-X)^{-1}b_0]}{\sqrt{-\det(A_0-X)}}
\end{split}
\end{align}
A useful way to show the formulas for the Gaussian integral is to think of matrices in block form and then to manipulate the blocks. So after subtracting the off-diagonal identities due to the integration measure, we permute rows and columns until the integration variables come before all the other ones. Then the matrix and the vector that parametrize the integrand will have the following block form:

\begin{center}
\begin{tikzpicture}[scale=0.5]
\fill[red!30] (0,2) rectangle (1,3);  
\fill[blue!30] (1,2) rectangle (3,3);  
\fill[blue!30] (0,0) rectangle (1,2);  
\fill[green!30] (1,0) rectangle (3,2); 
\draw[thick] (0,0) rectangle (3,3);
\draw[thick] (1,0) -- (1,3);
\draw[thick] (0,2) -- (3,2);
\node at (0.5,2.5) {$M$};
\node at (-1,1.5) {$A=$};
\end{tikzpicture}
\end{center}

\begin{center}
\begin{tikzpicture}[scale=0.5]
\fill[violet!60] (0,2) rectangle (0.5,3);  
\draw[thick] (0,2) rectangle (0.5,3);
\fill[lime] (0,0) rectangle (0.5,2);  
\draw[thick] (0,0) rectangle (0.5,2);
\node at (-1,1.5) {$b=$};
\end{tikzpicture}
\end{center}

Then, the Abc triple that parametrizes the result of the integral is given by the following block formulas:
\begin{center}
\begin{tikzpicture}[scale=0.6]
\fill[green!30] (0,0) rectangle (2,2);
\draw[thick] (0,0) rectangle (2,2);

\draw[thick] (2.7,1) -- (3.3,1);  

\fill[blue!30] (4,0) rectangle (5,2);
\draw[thick] (4,0) rectangle (5,2);
\fill[red!30] (5.1,1) rectangle (6.1,2);
\draw[thick] (5.1,1) rectangle (6.1,2);
\fill[blue!30] (6.2,1) rectangle (8.2,2);
\draw[thick] (6.2,1) rectangle (8.2,2);

\node at (-1,1) {$A=$};
\node at (5.6,1.5) {\scalebox{0.7}{$M^{\scriptscriptstyle{-1}}$}};

\end{tikzpicture}
\end{center}

\begin{center}\hspace{-2.25em}
\begin{tikzpicture}[scale=0.6]
\fill[lime] (0.75,0) rectangle (1.17,2);
\draw[thick] (0.75,0) rectangle (1.17,2);

\draw[thick] (2.7,1) -- (3.3,1);  

\fill[blue!30] (4,0) rectangle (5,2);
\draw[thick] (4,0) rectangle (5,2);
\fill[red!30] (5.1,1) rectangle (6.1,2);
\draw[thick] (5.1,1) rectangle (6.1,2);
\fill[violet!60] (6.2,1) rectangle (6.62,2);
\draw[thick] (6.2,1) rectangle (6.62,2);

\node at (-1,1) {$b=$};
\node at (5.6,1.5) {\scalebox{0.7}{$M^{\scriptscriptstyle{-1}}$}};
\end{tikzpicture}
\end{center}

\begin{center}\hspace{-3.25em}
\begin{tikzpicture}[scale=0.6]
\def\x{1.0}  
\def\y{0}  

\draw[thick] (-0.2,1) -- (6,1);  

\fill[violet!60] ({1.3+\x},{1.78+\y}) rectangle ({2.3+\x},{2.2+\y});
\draw[thick] ({1.3+\x},{1.78+\y}) rectangle ({2.3+\x},{2.2+\y});

\fill[red!30] ({2.4+\x},{1.2+\y}) rectangle ({3.4+\x},{2.2+\y});
\draw[thick] ({2.4+\x},{1.2+\y}) rectangle ({3.4+\x},{2.2+\y});

\fill[red!30] (2.8,-0.3) rectangle (3.8,0.7);
\draw[thick] (2.8,-0.3) rectangle (3.8,0.7);
\node at (3.3,0.2) {\scalebox{0.8}{$iM$}};

\fill[violet!60] ({3.5+\x},{1.2+\y}) rectangle ({3.92+\x},{2.2+\y});
\draw[thick] ({3.5+\x},{1.2+\y}) rectangle ({3.92+\x},{2.2+\y});

\node at (-1,1) {$c=$};
\node at (3.9,1.7) {\scalebox{0.7}{$M^{\scriptscriptstyle{-1}}$}};
\node at (1.0,1.7) {$\exp\bigl(-\frac12$};
\node at (5.5,1.7) {$\bigr)$};
\node at (2.5, 0.3) {$\sqrt{\det\phantom{\bigl(}\hspace{0.5cm}}$};
\end{tikzpicture}
\end{center}

One special case where the Bargmann representation is particularly useful is the inner product with the vacuum. The Bargmann function of the $n$-mode vacuum is the constant function $F_{|0\rangle}(z) = 1$, which leads to the useful fact that the leftover Bargmann function is obtained by evaluating the original Bargmann function at zero on the variables involved in the inner product, or alternatively by discarding the corresponding rows and columns from the Abc parametrization. 

This gives a direct interpretation of the $c$ part of the Abc parametrization. For a ket it is $c_{|\psi\rangle} = F_{|\psi\rangle}(0) = \langle 0|\psi\rangle$, i.e., the vacuum amplitude of $|\psi\rangle$. For a unitary it is $c_U = F_U(0,0) = \langle0|U|0\rangle$, i.e., the vacuum-to-vacuum transition amplitude. For a density matrix it is $c_\rho = F_\rho(0,0) = \langle0|\rho|0\rangle$, i.e., the vacuum probability. And so on depending on the nature of the object described by the Bargmann function.

An instructive example of this is the Fock damping operator, which is equivalent to a beamsplitter with vacuum in one of its inputs and in one of its outputs. In fact, the Bargmann matrix of the Fock damping operator can be computed by starting from the Bargmann matrix of a beam splitter and then deleting the two rows and columns that correspond to one input and one output.

In what follows, we describe how one can use inner product formulas above to compute the inverse of a Gaussian operation.

\subsection{Convergence of Gaussian integrals} As a remark, we highlight that the contraction formulas introduced above assume the Gaussian integral converges. This is often the case, as the application of Gaussian unitaries and channels on Gaussian states is well-defined. However, convergence of a Gaussian integral is not always guaranteed, and we should be careful when contracting non-physical Gaussian components. Assume we are contracting two Gaussian objects $G_1$ and $G_2$, where $G_1$ is defined on the joint space $MN$ and $G_2$ is defined on the joint space $NP$. The contraction is on the space $N$. This is the most general form of a contraction. Let $A_{N}^{(i)}$ be the $N$ block of $A_{G_i}$ for $i=1,2$. This means
\begin{align}
A_{G_1} = \begin{bmatrix}
\ast & \ast\\
\ast & A_{N}^{(1)}
\end{bmatrix}, \quad A_{G_2} = \begin{bmatrix}
A_{N}^{(2)} & \ast\\
\ast & \ast
\end{bmatrix}.
\end{align}
We are interested in checking the convergence of the following integral
\begin{align}
\int_{z\in\mathbb C^{|N|}} e^{-\norm{z}^2} F_{G_1}(w,z^\ast) F_{G_2}(z,v)\, \mathrm d^{2|N|}z
\end{align}
for all $w\in\mathbb C^{|M|}$ and $v\in\mathbb C^{|P|}$.
We show that
\begin{proposition}\label{prop:contraction-condition}
The contraction of $G_1$ and $G_2$ is well-defined if
\begin{align}
\norm{A_N^{(1)} + A_{N}^{(2)}{}^\ast} < 2,
\end{align}
where $\norm{\cdot}$ denotes the operator norm (i.e., the largest singular value).
\end{proposition}
\begin{proof}
The convergence of the contraction requires
the convergence of the integral
\begin{align}
\int_{z\in\mathbb C^{|N|}} \exp\left(\frac12 [z^\ast{}^{\mathrm T} z^{\mathrm T}] M \begin{bmatrix}
z^\ast\\
z
\end{bmatrix}
+ \zeta^{\mathrm T} \begin{bmatrix}
z^\ast\\
z
\end{bmatrix}
\right)\, \mathrm d^{2|N|}z,
\end{align}
where
\begin{align}
M = \begin{bmatrix}
A_{N}^{(1)} & -\mathbb 1\\
-\mathbb 1 & A_N^{(2)}
\end{bmatrix}
\end{align}
Here, $\zeta$ is some vector that depends on the other blocks of the Abc parametrizations of $G_1$ and $G_2$ and extra variables $w,v$, but it will not impact the convergence of the integral. We can rewrite the integral in terms of real-valued parameters by letting $z = \frac{x+iy}{\sqrt 2}$ and noticing that
\begin{align}\label{eq:M-prime}
\begin{bmatrix}
z^\ast{}^{\mathrm T} & z^{\mathrm T}
\end{bmatrix}
M
\begin{bmatrix}
z^\ast\\
z
\end{bmatrix}
=
\begin{bmatrix}
x^{\mathrm T} & y^{\mathrm T}
\end{bmatrix}
\underbrace{
\begin{bmatrix}
-\mathbb 1 + \frac{A_N^{(1)}  + A_N^{(2)}}{2} & \frac{iA_N^{(1)} -i A_N^{(2)}}{2}\\
\frac{iA_N^{(1)}  -i A_N^{(2)}}{2} & -\mathbb 1 - \frac{A_N^{(1)}  + A_N^{(2)}}{2}
\end{bmatrix}}_{M'}
\begin{bmatrix}
x\\
y
\end{bmatrix}.
\end{align}
As the Gaussian integral
\begin{align}
\int_{x,y\in\mathbb R^{|N|}} \exp\left(\frac12 [x^{\mathrm T} y^{\mathrm T}] M' \begin{bmatrix}
x\\
y
\end{bmatrix}
+ v^T \begin{bmatrix}
x\\
y
\end{bmatrix}
\right)\, \mathrm d^{|N|}x\mathrm d^{|N|}y
\end{align}
is convergent if $\mathrm{Re}(M') < 0$ (see \cite[Section 1.2]{zinn2010path}). Using \eqref{eq:M-prime} for the expression of $M'$, we get
\begin{align}\label{eq:contraction-condition-int}
\begin{bmatrix}
\Re(A_N^{(1)} + A_N^{(2)}) & -\Im(A_N^{(1)} - A_N^{(2)})\\
-\Im(A_N^{(1)} - A_N^{(2)}) & -\Re(A_N^{(1)} + A_N^{(2)})
\end{bmatrix} < 2\cdot\mathbb 1.
\end{align}
We then use the following lemma. In the following lemma, we use $(\lambda_{j})_j$ to denote the eigenvalues, and we use $\sigma_{\max}$ to denote the largest singular value.
\begin{lemma}\label{lem:helper-2}
The symmetric matrix constructed as
\begin{align}\label{eq:composite-matrix}
C = \begin{bmatrix}
D & E\\
E & -D
\end{bmatrix}
\end{align}
by real-valued symmetric matrices $D,E \in \mathbb R^n$ satisfies $\max_{j=1,\cdots,2n} |\lambda_{j}(C)| = \sigma_{\max}(D+iE)$.
\end{lemma}
\begin{proof}
Let $W = \frac1{\sqrt 2} \begin{bmatrix}
\mathbb 1_n & i \mathbb 1_n\\
\mathbb 1_n & -i \mathbb 1_n
\end{bmatrix}$ and notice that
\begin{align}
W C W^\dag = \begin{bmatrix}
0 & D + iE\\
D - iE & 0
\end{bmatrix}.
\end{align}
Let $H = W C W^\dagger$. As $H$ is Hermitian, we have that $\max_j |\lambda_j(H)| = \max_j \sqrt{\lambda_j(H^2)}$. Furthermore, we have
\begin{align}
H^2 = \begin{bmatrix}
(D+iE)(D-iE) & 0\\
0 & (D-iE)(D+iE)
\end{bmatrix},
\end{align}
and as eigenvalues of $(D+iE)(D-iE)$ and $(D-iE)(D+iE)$ are the squared singular values of $D+iE$, we conclude that $\max_j |\lambda_j(H^2)| = \sigma_{\max}(D+iE)^2$. The final result follows as $\lambda_j(H) = \lambda_j(C)$ since they are related through a conjugation by a unitary ($W$).
\end{proof}
We finish the proof, by noting that according to \protect\cref{lem:helper-2}, the largest eigenvalue on the left-hand side of \eqref{eq:contraction-condition-int} is $\norm{\Re(A_N^{(1)} + A_N^{(2)}) - i \Im (A_N^{(1)} - A_N^{(2)})} = \norm{A_N^{(1)}{}^\ast +A_N^{(2)}}$.
\end{proof}

To give example use cases of \protect\cref{prop:contraction-condition}, we note that the condition for a Gaussian state $\ket{\psi}$ (with $A_\psi$) to be normalizable is that it should be able to contract it with $\bra{\psi}$ (which has the A-matrix $A_\psi^\ast$). According to \protect\cref{prop:contraction-condition}, this is possible if $\norm{A_\psi + (A_\psi^\ast)^\ast} < 2$ which is equivalent to the known condition $\norm{A_\psi} <1 $. As another example we consider inversion of a Gaussian operator. Let $G$ be a Gaussian operator, with Abc parametrization in the output-input order written as
\begin{align}
A_G = \begin{bmatrix}
A_{\mathrm{out}} & R^T\\
R & A_{\mathrm{in}}
\end{bmatrix},
\end{align}
where $A_{\mathrm{out}}, A_{\mathrm{in}}, R \in \mathbb C^{n\times n}$. We have that $G$ is invertible if $\norm{A_{\mathrm{in}}^\ast + (A_{\mathrm{in}} - R A_{\mathrm{out}}^{-1} R^T)^{-1}  } \leq 2$, and its inverse $G^{-1}$ is parametrized by
\begin{align}
A_{G^{-1}} = X A_{G}^{-1} X,\quad b_{G^{-1}} = -XA_{G}^{-1} b_G,
\end{align}
where $X = \begin{bmatrix}
    0 & \mathbb 1_n\\
    \mathbb 1_n & 0
\end{bmatrix}$. This can be verified directly by using the Gaussian contraction formulas and checking \protect\cref{prop:contraction-condition}.

\section{Physicality requirements}\label{app:physicality}

In this section we discuss conditions under which a CV Gaussian operator is Hermitian, positive, and lastly conditions under which the operator corresponds to a CV density matrix. We then continue by extending these conditions to CV maps, using their Choi–Jamiołkowski operator.

\subsection{Operators}

Recall that every Gaussian operator that acts on a Hilbert space of $n$ particles (e.g., a Gaussian mixed state, or a Gaussian unitary) has a Gaussian Bargmann function \cite{yao2024riemannian} defined as
\begin{align}\label{eq:rho-in-bargmann}
\begin{split}
F(z, z') &= \exp(\frac{\norm{z}^2 + \norm{z'}^2}{2}) \bra{z^\ast} \rho \ket{z'}\\
&= c  \exp(\zeta^{\mathrm{T}} A \zeta + \zeta^{\mathrm{T}}b),
\end{split}
\end{align}
with $\zeta = \begin{bmatrix}z\\z'\end{bmatrix}\in\mathbb{C}^{2n}$.
Let $G$ represent such an operator. Note that $ A\in\mathbb C^{2n\times 2n}$ and $ b \in \mathbb C^{2n}$. If the operator is a Hermitian CV operator, then its Ab parametrization has to satisfy certain block structures, determined in the next lemma.

\begin{lemma}\label{lem:hermiticity}
Let $A_G$ and $b_G$ represent the Ab parameters of a Hermitian operator in type-wise ordering. Then, they admit the following block structures
\begin{align}\label{eq:block-structures}
A_G = \begin{bmatrix}
\Lambda^\dagger & \Gamma\\
\Gamma^{\mathrm{T}} & \Lambda
\end{bmatrix},
\quad
b_G = \begin{bmatrix}
\beta^\ast\\
\beta
\end{bmatrix}
\end{align}
with blocks satisfying
\begin{align}\label{eq:hermitian-conditions}
\Lambda = \Lambda^{\mathrm{T}}, \quad \Gamma = \Gamma^\dag,
\end{align}
\end{lemma}
\begin{proof}
Note that $A_G$ represents a quadratic form, and hence, can always be taken to be complex symmetric. Moreover, as $G$ is Hermitian, we get
\begin{align}
\bra{ z}G\ket{ w} = \bra{w}G\ket{ z}^\ast.
\end{align}
Using the formulation in \eqref{eq:rho-in-bargmann}, this gives $ X  A_G  X =  A^\ast_G$ and $ X  b =  b^\ast$ with $ X = (\begin{smallmatrix}
    0&\mathbb{1}\\\mathbb{1} & 0
\end{smallmatrix})$, which (using the fact that $ A_G$ is symmetric) gives the conditions in Eq.~\eqref{eq:hermitian-conditions} and the block structures in Eq. \eqref{eq:block-structures}.
\end{proof}

We now introduce conditions under which $G$ corresponds to a positive CV operator.

\begin{lemma}\label{lem:positivity}
The operator $G$ is positive semi-definite, if and only if
\begin{align}
\Gamma \geq 0, \text{and }c\geq 0.
\end{align}
\end{lemma}
\begin{proof}
First, note that $c\geq 0$ must hold as
\begin{align}
\bra{0} G \ket{0} = c,
\end{align}
and that $\bra{0} G \ket{0}\geq0$ is imposed by positivity of $G$.
We show the rest this in two steps:
\begin{itemize}
    \item $G \geq 0 \Rightarrow \Gamma\geq 0$: To see this, let $\ket{1_i}$ represent the Fock state of $n$ modes that has one photon on the $i$-th mode and is vacuum elsewhere, i.e.,
    \begin{align}
    \ket{1_i} := |0\cdots0
    \underbrace{1}_{\text{\small $i$-th position}}
    0\cdots 0\rangle.
\end{align}
We then have
\begin{align}
\bra{1_i} \rho \ket{1_j} = c\cdot(\Gamma + \beta^\ast\beta^{\mathrm{T}})_{ij}.
\end{align}
Note that $(\bra{1_i} \rho \ket{1_j})_{ij}$ represents the principal submatrix of the Fock representation where the total photon number is $1$, it should be a positive semi-definite matrix. As the b vector can be made zero through conjugation by displacement operators (which does not affect the A matrix), we conclude that positivity of $G$ implies $c\cdot \Gamma \geq 0$.

\item $\Gamma\geq 0 \Rightarrow G\geq 0$: Here, we use the following decomposition
\begin{align}
G = T^\dagger T 
\end{align}
with $T$ being a Gaussian operator acting on Hilbert space of $n$ modes with the following Ab matrices
\begin{align}
A_T = \begin{bmatrix}
\Lambda & \sqrt\Gamma^{\mathrm{T}}\\
\sqrt\Gamma & 0
\end{bmatrix},
\quad
b_T = \begin{bmatrix}
\beta\\
0
\end{bmatrix},
\quad
c_T = \sqrt{c}.
\end{align}
One can readily verify that $ G = T^\dag T$ through the inner product formulas introduced in \protect\cref{app:inner-product-in-bargmann}. We also highlight that $T$ and $T^\dagger$ can be contracted as the contraction condition of \protect\cref{prop:contraction-condition} (in this case $\norm{0+0}=0<2$) is satisfied.
\end{itemize}
\end{proof}

We now introduce a physicality requirement for mixed states and quantum channels. Recall that the only requirement for $A_{|\psi\rangle}$ to represent a physical pure state $\ket\psi$ was to be complex symmetric with eigenvalues in the unit disk \cite{chabaud2021continuous}, which simply corresponds to $\ket\psi$ having a finite norm. We now proceed to introduce physicality requirements for a mixed Gaussian state. To this end, we first need an intermediary lemma:

\begin{lemma}\label{lem:block-pos}
We have the following equivalence
\begin{align}\label{eq:equality}
    \begin{bmatrix}
     P &  R\\ R^\dagger & Q
\end{bmatrix}\geq 0 \Leftrightarrow \begin{cases} P\geq0\\ Q\geq  R^\dagger  P^{+}  R.\end{cases}
\end{align}
and similarly,
\begin{align}\label{eq:strict-inequality}
    \begin{bmatrix}
     P &  R\\ R^\dagger & Q
\end{bmatrix} > 0 \Leftrightarrow \begin{cases} P>0\\ Q >  R^\dagger  P^{-1}  R.\end{cases}
\end{align}
\end{lemma}
\begin{proof}
Let us first show \eqref{eq:equality}. Note that we can diagonalize $P$ by conjugating the big block matrix by $(\begin{smallmatrix}
V & 0\\
0 & \mathbb 1
\end{smallmatrix})$, where $V$ diagonalizes $P$ i.e., $V P V^\dagger$ is diagonal. Now, note that if any entry on the diagonal of $VPV^\dag$ is zero, the corresponding row in $V R$ must be entirely zero. Therefore, if the projector onto the image of $P$ is denoted by $\Pi_P$, we have $R = \Pi_P R$. Also, let us denote the pseudo inverse of a matrix, say $M$, by $M^+$. We then proceed to prove each direction of \eqref{eq:equality} in the following.
\begin{itemize}
\item To prove the backward ($\Leftarrow$) direction, note that
\begin{align}\label{eq:lemma-backward}
\begin{split}
\begin{bmatrix}
 P &  R\\
 R^\dagger &  Q
\end{bmatrix}
= \begin{bmatrix}
\mathbb 1 & 0\\
 R^\dagger  P^{+} & \mathbb 1
\end{bmatrix}
\begin{bmatrix}
 P & 0\\
0 &  Q -  R^\dagger  P^{+}  R
\end{bmatrix}
\begin{bmatrix}
\mathbb 1 &  P^{+} R\\
0 & \mathbb 1
\end{bmatrix}.
\end{split}
\end{align}
To see the equality, it is important to note that $R^\dag P^+ P = R^\dag \Pi_P = R^\dag$.
\item For the forward ($\Rightarrow$) direction, we first note that
\begin{align}
\begin{bmatrix}
\mathbb 1 &  A\\
0 & \mathbb 1
\end{bmatrix}^{-1} = 
\begin{bmatrix}
\mathbb 1 & - A\\
0 & \mathbb 1
\end{bmatrix},
\end{align}
which, based on Eq.~\eqref{eq:lemma-backward} implies
\begin{align}
\begin{split}
\begin{bmatrix}
 P & 0\\
0 &  Q -  R  P^{+}  R^\dagger
\end{bmatrix}= \begin{bmatrix}
\mathbb 1 & 0\\
- R^\dagger  P^{+} & \mathbb 1
\end{bmatrix}
\begin{bmatrix}
 P &  R\\
 R^\dagger &  Q
\end{bmatrix}
\begin{bmatrix}
\mathbb 1 & - P^{+} R\\
0 & \mathbb 1
\end{bmatrix}.
\end{split}
\end{align}
and proves the forward ($\Rightarrow$) direction.
\end{itemize}
For proving \eqref{eq:strict-inequality}, note that all principal submatrices of a positive definite matrix, are positive definite. This ensures the existence of $P^{-1}$. The rest of the proof is analogous to the proof of \eqref{eq:equality}, where one should note that the matrices of the form $(\begin{smallmatrix}
\mathbb 1 & A\\
0 & \mathbb 1
\end{smallmatrix})$ are full-rank and therefore, if $M = (\begin{smallmatrix}
    \mathbb 1 & 0\\
    A^\dag & 0
\end{smallmatrix}) N (\begin{smallmatrix}
\mathbb 1 & A\\
0 & \mathbb 1
\end{smallmatrix})$ with $N>0$, we conclude that $M>0$.
\end{proof}

\begin{proposition}\label{prop:dm-physicality}
Let $\rho$ be a Gaussian mixed state with Ab parameters $A_\rho,  b_\rho$ written in type-wise order as
\begin{align}
A_\rho = \begin{bmatrix}
\Lambda^\dag & \Gamma\\
\Gamma^{\mathrm{T}} & \Lambda
\end{bmatrix}
\end{align}
Then, we have
\begin{align}\label{eq:dm-physicality-conditions}
\begin{split}
&0\leq\Gamma<\mathbb{1},\\
&\Lambda^{\dagger}(\mathbb{1} - \Gamma^{\mathrm{T}})^{-1}\Lambda < \mathbb{1} - \Gamma,\\
&c_\rho = \sqrt{\mathrm{det}(\mathbb 1 - XA_\rho)}\, \exp(\frac12 b_\rho^{\mathrm{T}} (\mathbb 1 - X A_\rho)^{-1} b_\rho^\ast).
\end{split}
\end{align}
These conditions are both necessary and sufficient for a Bargmann function to represent a Gaussian mixed state.
\end{proposition}
\begin{proof}
We provide two separate proofs here. 

\paragraph{First proof} An operator $\rho$ corresponds to a valid density matrix if and only if $\rho$ is a positive semi-definite operator and $\tr(\rho) = 1$. The positivity condition, according to \protect\cref{lem:positivity} is equivalent to
\begin{align}\label{eq:positivity-first-pf}
\Gamma\geq 0.
\end{align}
We now need to check $\tr(\rho) = 1$. To this end, we write
\begin{align}\label{eq:dm-finite-trace}
\begin{split}
\tr(\rho) &= \frac{1}{\pi^m} \int_{\mathbb{C}^m} \bra{z}\rho \ket{z} \, \mathrm d^{2m}z\\
&= \frac{1}{\pi^m} \int_{\mathbb{C}^m} F_\rho(z^\ast, z)\exp(-\norm{z}^2) \,\mathrm d^{2m}z\\
&= \frac{c_\rho}{\pi^m} \int_{\mathbb{C}^m} \exp(\frac12
\begin{bmatrix}
z^\ast{^{\mathrm{T}}} & z^{\mathrm{T}}
\end{bmatrix}
A
\begin{bmatrix}
z^\ast\\
z
\end{bmatrix} + b^{\mathrm{T}}
\begin{bmatrix}
z^\ast\\
z
\end{bmatrix}) \exp(-\norm{z}^2) \, \mathrm d^{2m}z\\
&= \frac{c_\rho}{\pi^m} \int_{\mathbb{R}^{2m}} \exp\left( -\begin{pmatrix}
 x^{\mathrm{T}} &  y^{\mathrm{T}}
\end{pmatrix}  T^\dagger (\mathbb 1 - XA) T
\begin{pmatrix}
x\\
y
\end{pmatrix}\right) \, \mathrm d^mxd^my,
\end{split}
\end{align}
where $x,y$ are real vectors representing the real and imaginary part of the complex vector $z$, and $T = \begin{pmatrix}
\mathbb 1 & -i\mathbb 1\\
\mathbb 1 & i\mathbb 1
\end{pmatrix}$ is proportional to a unitary matrix. Finally, note that
\begin{align}\label{eq:finite-trace-int-form}
\int_{\mathbb{R}^{2m}} \exp\left( -\begin{pmatrix}
 x^{\mathrm{T}} &  y^{\mathrm{T}}
\end{pmatrix}  T^\dagger (\mathbb 1 - XA) T
\begin{pmatrix}
x\\
y
\end{pmatrix} + 
\begin{pmatrix}
\beta_r &  \beta_i
\end{pmatrix}
\begin{pmatrix}
x\\
y
\end{pmatrix}
\right) \, \mathrm d^mxd^my < \infty \Leftrightarrow \mathbb 1 - XA > 0,
\end{align}
We can write the if and only if in \eqref{eq:finite-trace-int-form} since $\mathbb 1 - XA$ is always Hermitian as $\rho$ is a Hermitian operator. Finally, applying \protect\cref{lem:block-pos}, we obtain that
\begin{align}
\mathbb 1 - XA > 0 \Leftrightarrow \begin{cases}
\Gamma < \mathbb 1,\\
\Lambda^\dag (\mathbb 1 - \Gamma^{\mathrm{T}})^{-1} \Lambda < \mathbb 1 - \Gamma.
\end{cases}
\end{align}
Moreover, note that using $\tr(\rho) = 1$ and using the Gaussian integral above, we get
\begin{align}
c = \sqrt{\mathrm{det}(\mathbb 1 - XA_\rho)}\, \exp(\frac12 b_\rho^{\mathrm{T}} (\mathbb 1 - X A_\rho)^{-1} b_\rho^\ast),
\end{align}
which together with \eqref{eq:positivity-first-pf} form all our physicality conditions.

\paragraph{Second proof} The $ A_\rho$ matrix of a mixed state is related to the Husimi covariance matrix via $\Sigma = (\mathbb{1} -  X A_\rho)^{-1}$. The physicality constraints using the Husimi covariance are $\Sigma \geq \left(\begin{smallmatrix}
    0&0\\0&\mathbb{1}
\end{smallmatrix}\right)$, meaning that a state is physical if and only if its Husimi covariance matrix satisfies such inequality. Therefore, we need to convert this requirement into conditions that apply to $ A_\rho$.

Recall that we want to find conditions on $ A_\rho$ such that
\begin{align}\label{eq:phys-initial-form}
\left( \mathbb 1 -  X  A_\rho \right)^{-1} \geq \begin{bmatrix}
0 & 0\\
0 & \mathbb 1
\end{bmatrix}.
\end{align}
Using Schur's complement formula, we can rewrite the LHS as 
\begin{align}
\begin{bmatrix}
\mathbb 1 - \Gamma ^{\mathrm{T}} & -\Lambda\\
-\Lambda^\dagger & \mathbb 1 - \Gamma
\end{bmatrix}^{-1}
=
\begin{bmatrix}
(\mathbb 1 - \Gamma^{\mathrm{T}})^{-1} + (\mathbb 1 - \Gamma^{\mathrm{T}})^{-1}\Lambda {H}^{-1} \Lambda^\dag (\mathbb 1 - \Gamma^{\mathrm{T}})^{-1} & (\mathbb 1 - \Gamma^{\mathrm{T}})^{-1} \Lambda {H}^{-1}\\
{H}^{-1} \Lambda^\dag (\mathbb 1 - \Gamma^{\mathrm{T}})^{-1} & {H}^{-1}
\end{bmatrix},
\end{align}
where ${H} = \mathbb 1 - \Gamma - \Lambda^\dag(\mathbb 1 - \Gamma^{\mathrm{T}})^{-1} \Lambda$. Therefore, the condition \eqref{eq:phys-initial-form} combined with \protect\cref{lem:block-pos} implies
\begin{itemize}
\item The first condition on the RHS of the lemma reads as
\begin{align}\label{eq:R-condition}
{H}^{-1} - \mathbb 1 \geq 0 \Leftrightarrow 0 < {H} \leq \mathbb 1.
\end{align}
\item The second condition on the RHS of the lemma gives
\begin{align}
(\mathbb 1 - \Gamma^{\mathrm{T}})^{-1} + (\mathbb 1 - \Gamma^{\mathrm{T}})\Lambda {H}^{-1} \Lambda^\dag (\mathbb 1 - \Gamma^{\mathrm{T}})^{-1} \geq (\mathbb 1 - \Gamma^{\mathrm{T}})^{-1} \Lambda {H}^{-1}({H}^{-1}-\mathbb 1)^{-1} {H}^{-1}\Lambda^\dag (\mathbb 1 - \Gamma^{\mathrm{T}})^{-1}.
\end{align}
Note that the above inequality implies $\Gamma < \mathbb 1$. Using this fact, we get the following equivalent form
\begin{align}
(\mathbb 1 - \Gamma^{\mathrm{T}})^{-1} \geq (\mathbb 1 - \Gamma^{\mathrm{T}})^{-1} \Lambda (\mathbb 1 - {H})^{-1} \Lambda^\dag (\mathbb 1 - \Gamma^{\mathrm{T}})^{-1},
\end{align}
which then finally gives
\begin{align}
\mathbb 1 - \Gamma^{\mathrm{T}} \geq \Lambda (\mathbb 1 - {H})^{-1} \Lambda^\dag.
\end{align}
Now, noticing that this is simply equivalent to $\mathbb 1 - {H} \geq \Lambda^\dag (\mathbb 1 - \Gamma^{\mathrm{T}})^{-1}\Lambda$, and plugging in ${H} = \mathbb 1 - \Gamma - \Lambda^\dag(\mathbb 1 - \Gamma^{\mathrm{T}})^{-1}\Lambda$, we get that this all becomes equivalent to $\Gamma \geq 0$. Overall, this condition gave us the following
\begin{align}
0 \leq \Gamma < \mathbb 1.
\end{align}
\end{itemize}
We point out that plugging in the expression for ${H}$ in \eqref{eq:R-condition} and utilizing the condition $0\leq \Gamma < \mathbb 1$, we get the other physicality condition $\Lambda^{\dagger}(\mathbb{1} - \Gamma^{\mathrm{T}})^{-1}\Lambda < \mathbb{1} - \Gamma$. Lastly, we point out that $c$ can be computed as in the first proof we presented.
\end{proof}

\begin{remark}\label{rem:physicality}
Note that $(\mathbb 1 - XA_\rho) > 0$ is imposed by the first two lines of inequalities in \eqref{eq:dm-physicality-conditions}, which then imposes that $c_\rho \geq 0$ via the last equality there. To see $(\mathbb 1 - XA_\rho) > 0$, note that following our second proof, we have obtained that \eqref{eq:phys-initial-form} is identical to the first two lines of inequalities in \eqref{eq:dm-physicality-conditions}. This is while \eqref{eq:phys-initial-form} indicates $\mathbb 1-X A_\rho$ is invertible and positive semi-definite, which then results in positive definiteness of $\mathbb 1-X A_\rho$ 
\end{remark}

\subsection{Maps}
We now turn to prove complete positivity and trace-preserving conditions on Gaussian maps. Firstly, we point out that a generic Gaussian CP map has an Abc triple $(A_\Phi,b_\Phi,c_\Phi)$, with the following block form
\begin{align}\label{eq:channel-block-structure}
A_\Phi = 
\begin{bmatrix}
\Lambda^\dag_\Phi & \Gamma_\Phi\\
\Gamma_\Phi^{\mathrm{T}} & \Lambda_\Phi
\end{bmatrix},
\quad
b_\Phi = \begin{bmatrix}
\beta_\Phi^\ast\\
\beta_\Phi
\end{bmatrix}.
\end{align}
Note that as $A$ matrix is symmetic, we have $\Lambda_\Phi = \Lambda^T_\Phi$. Also, as we will show later in \cref{rem:cp-map-block-structure}, for any CP map $\Phi$, it is the case that $\Gamma_\Phi = \Gamma_\Phi^\dag$. Hence, \eqref{eq:channel-block-structure} matches the block structure in \eqref{eq:CP-map}. We recall that the channel has the following blocks in output-input ordering
\begin{align}
\begin{bmatrix}
A_\Phi^{\mathrm{out}} & R^{\mathrm{T}}_\Phi\\
R_\Phi & A_\Phi^{\mathrm{in}}
\end{bmatrix}.
\end{align}

We now prove the lemma introduced in the main text, determining the conditions for completely positive (CP) maps.

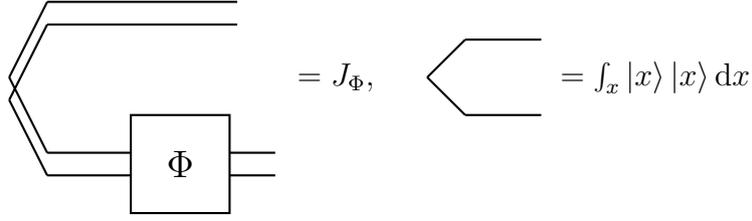
\begin{figure}
    \centering
\begin{tikzpicture}
\draw[thick, black] (-1,0) -- (1.5,0);
\draw[thick, black] (-1,-.3) -- (1.5,-.3);
\draw[thick, black] (-1,-2) -- (0.1,-2);
\draw[thick, black] (-1,-2.3) -- (.1, -2.3) ;

\draw[thick, black] (-1,0) -- (-1.5,-1);
\draw[thick, black] (-1.5, -1) -- (-1,-2);
\draw[thick, black] (-1, -.3) -- (-1.5, -1.3);
\draw[thick, black] (-1.5,-1.3) -- (-1, -2.3);

\draw[thick, black] (0.1, -1.5) rectangle (1.4, -2.8);
\node at (.75, -2.15) {\Large $\Phi$};
\draw[thick, black] (1.4, -2) -- (2, -2);
\draw[thick, black] (1.4, -2.3) -- (2, -2.3);

\node at (2.8, -1) {\large $= J_\Phi,$};

\draw[thick, black] (4.5,-.5) -- (5.5,-0.5);
\draw[thick, black] (4.5,-1.5) -- (5.5,-1.5);
\draw[thick, black] (4, -1) -- (4.5, -.5);
\draw[thick, black] (4, -1) -- (4.5, -1.5);

\node at (7, -1) {\large $= \int_x \ket x\ket x \mathrm dx$};
\end{tikzpicture}
\caption{The Choi state ($J_\Phi$) corresponding to a channel $\Phi$. As it can be observed, the Choi state is obtained by merely `bending the wires' of $\Phi$. This simply corresponds to rearranging the indices of the Bargmann representation. Interestingly, the order of indices does not change, if written in type-wise ordering (introduced in \protect\cref{app:ordering-conventions}).}
\label{fig:choi}
\end{figure}
\begin{lemma}[\protect\cref{prop:channels-in-bargmann} in the main text]\label{lem:cp-condition-in-appendix}
A Gaussian map, say $\Phi$, represented by a Gaussian Bargmann function as in \eqref{eq:CP-map}, is completely positive if and only if
\begin{align}
0\leq \Gamma_\Phi, \text{and } 0\leq c_\Phi.
\end{align}
\end{lemma}
\begin{proof}
The proof relies on utilizing the Choi state representation of the channel. Recall that the Choi state, represented by $J_\Phi$ is defined as 
\begin{align}
J_\Phi = (\mathcal I \otimes \Phi)[\ket{\mathrm{EPR}}\!\bra{\mathrm{EPR}}],
\end{align}
and as shown in \protect\cref{fig:choi}, this transformation simply `bends the wires' of $\Phi$ (here $\ket{\mathrm{EPR}} = \int_x \ket{x}\ket{x} \mathrm dx = \sum_n \ket{n}\ket{n}$ represents the CV maximally entangled state). This means that by properly arranging the indices of $A_\Phi$, it will transform into the A matrix representing the Choi state. We recall that a channel is completely positive (CP) if and only if $J_\Phi$ corresponds to a positive operator \cite{watrous2018theory}, and hence, by applying \protect\cref{lem:positivity}, we get the CP requirement as
\begin{align}
\Gamma_\Phi \geq 0, \text{and } c_\Phi\geq 0.
\end{align}
\end{proof}

\begin{remark}\label{rem:cp-map-block-structure}
From the proof above, we can see that the Choi state of a Gaussian channel is a Gaussian state with the same Abc parameters as the channel itself (in the type-wise ordering). Hence, by \cref{lem:hermiticity}, we can see that for any CP map $\Phi$, it is the case that $\Gamma_\Phi = \Gamma_\Phi^\dag$. This was also implied by the condition $\Gamma_\Phi \geq 0$ from \cref{lem:cp-condition-in-appendix}.
\end{remark}

We can also derive trace-preserving conditions for a map, and therefore, summarize the conditions under which a map is CPTP.

\begin{proposition}\label{prop:TP}
A map $\Phi$ is trace-preserving if and only if
\begin{align}\label{eq:channel-tp}
\begin{split}
&\mathbb 1 - X A_\Phi^{\mathrm{out}} > 0,\\
&A_\Phi^{\mathrm{in}} = R_\Phi ( A_\Phi^{\mathrm{out}} - X)^{-1} R_\Phi^{\mathrm{T}} + X,\\
&c_\Phi = \sqrt{\mathrm{det}(\mathbb 1 - XA_\Phi^{\mathrm{out}})} \, \exp(\frac12 (b_\Phi^{\mathrm{out}})^{\mathrm{T}} (\mathbb 1 - XA_\Phi^{\mathrm{out}}) (b_\Phi^{\mathrm{out}})^\ast).
\end{split}
\end{align}
As a result, $\Phi$ corresponds to a CPTP map if and only if other than \eqref{eq:channel-tp} it satisfies $\Gamma_\Phi \geq 0$.
\end{proposition}
\begin{proof}
First, note that sending vacuum into the channel, we get a Gaussian state with A matrix of $A_\Phi^{\mathrm{out}}$. Requiring the output state to have a finite trace we get (see \eqref{eq:dm-finite-trace})
\begin{align}
\mathbb 1 - XA_\Phi^{\mathrm{out}} > 0.
\end{align}
Next, we note that $\Phi$ is trace preserving if and only if its dual is unital, i.e.
\begin{align}\label{eq:channel-tp-def}
\Phi^\dagger(\mathbb 1) = \mathbb 1.
\end{align}
Applying contraction formulas of \eqref{app:inner-product-in-bargmann}, we can re-write \eqref{eq:channel-tp-def} as
\begin{align}
X = A_\Phi^{\mathrm{in}}  - R_\Phi(A_\Phi^{\mathrm{out}} - X)^{-1} R_\Phi^{\mathrm{T}},
\end{align}
which is the other condition for $\Phi$ to be TP. The equation for $c_\Phi$ is also obtained from equating the `c' parameter of the equation $\Phi^\dag(\mathbb 1) = \mathbb 1$.
\end{proof}

\subsection{Variation to Gaussian core states}\label{app:variation}

We highlight an interesting fact about our definition of core states. In particular, we show that restricting to $k_1=k_2$ in \protect\cref{def:core} results in the same set of core states.

\begin{proposition}\label{prop:equivalent-core}
A Gaussian state $\rho$ defined on a bipartite system $MN$ is a Gaussian core state if and only if for all $k\in\mathbb N^{n}$, the conditional state
\begin{align}\label{eq:alternate-core}
(\mathbb 1 \otimes \bra{k}) \rho (\mathbb 1\otimes \ket{k})
\end{align}
has finite support in the Fock basis.
\end{proposition}

This proposition allows us to interpret Gaussian core states as those that result in core state when we measure a subsystem (i.e., $N$) in the number basis.

\begin{proof}
It is transparent that any Gaussian core state (as defined in \protect\cref{def:core}) satisfies the provided constraint. We need to show that any $\rho$ that satisfies \protect\cref{eq:alternate-core}, is a Gaussian core state. To this end, we use the constraint with $k=0$ to get $(\mathbb 1\otimes \bra{0})\rho(\mathbb 1 \otimes \ket{0}) \propto \ket{0}\bra{0}$ as vacuum is the only Gaussian state with finite Fock support. Moreover, note that $(\mathbb 1\otimes \bra{0})\rho(\mathbb 1 \otimes \ket{0}) \propto \ket{0}\bra{0}$ implies that the top left block of $A_\rho$ and the first block of $b_\rho$ in mode-wise ordering must be zero. This, together with \protect\cref{prop:core-states-in-bargmann} implies that $\rho$ is a Gaussian core state (as defined in \protect\cref{def:core}).
\end{proof}

\section{Proofs}
Here we provide the technical proofs for the claims in the main text.

\subsection{The \texorpdfstring{$b\neq0$}{b != 0} case}

In the following, we show that if a state with Ab parameters $A_1,b_1=0$ (either pure or mixed) satisfies a stellar decomposition, so does any state with $A_1, b_2 \neq 0$. This implies that the existence of a stellar decomposition solely depends on the quadratic exponent described by the A matrix.

\begin{lemma}[Disregarding displacements]
Whether or not a state (either pure or mixed) accepts a stellar decomposition is solely dependent on $A$. In other words, a state $\rho$ with Ab parameters $(A,b)$ admits a stellar decomposition iff $\sigma$ with Ab parameters $(A,0)$ does so.
\end{lemma}
Let the stellar decomposition of $\sigma$ be $\sigma = (\Phi' \otimes \mathcal I) (\sigma_{\mathrm{core}})$, where the parametrization of $\Phi'$ and $\sigma_{\mathrm{core}}$ are $A_\Phi$ and $A_{\mathrm{core}}$. It is straightforward to show that $\rho = (\Phi'\otimes \mathcal I)(\rho')$ where $\rho'$ has the same A matrix as $\sigma_{\mathrm{core}}$ i.e., $A_{\mathrm{core}}$ but has potentially non-zero b vector, which we denote by $b'$. Recall that we need the part of $b'$ corresponding to the first $m$ modes to be zero in order for $\rho'$ to be a core state. In what follows, we show that by absorbing a displacement $D(\gamma)$ into $\rho'$ we can make it a Gaussian core state. We refer to \protect\cref{fig:disregard-displacements} for a schematic representation of the argument.

A straightforward calculation shows that $D(\gamma)\otimes \mathbb 1$ affects the b vector of a state $\rho(A,b)$ as
\begin{align}\label{eq:displacement-evolution}
D(\gamma)\otimes\mathbb 1: b \mapsto \begin{bmatrix}
\mathbb 1 - X_mA_m & -X_mR^{\mathrm{T}}\\
-X_n R & \mathbb 1 - X_n A_n
\end{bmatrix} 
\begin{bmatrix}
\tilde\gamma\\
0
\end{bmatrix}
+ b,
\end{align}
where the equation is written in mode-wise ordering with $X_n:= (\begin{smallmatrix}
0 & \mathbb 1_n\\
\mathbb 1_n & 0
\end{smallmatrix})$ being an X Pauli of size $2n$. Recall that our goal is to bring the b vector into the form $(\begin{smallmatrix}
0\\
*
\end{smallmatrix})$. By restricting \eqref{eq:displacement-evolution} to the first entries (i.e., the ones corresponding to the first $m$ modes), and recalling that the state $\rho'$ has the A matrix of a core state (i.e., $(A_m)_{\rho'}=0$) we find the displacement that allows transformation of $\rho'$ into an actual core state
\begin{align}
\gamma = -(b'_m)_k,
\end{align}
where $(b'_m)_k$ is the part of the b vector of $\rho'$ (i.e., $b'$) that corresponds to the ket part of the first $m$ modes of $\rho$. 

\begin{figure}[t]
\centering
\begin{tikzpicture}

\draw[thick, black] 
(-1, -1.6) arc[start angle=270, end angle=90, radius=1.2] -- (-1, 0) -- (-1, -1.6) -- cycle;
\node at (-1.59, -0.4) {$\rho(A,b)$};

\draw[thick, black] (-1, 0) -- (0, 0) node[ above] {\small $m$ modes};
\draw[thick, black] (-1, -1) -- (0, -1) node[below] {\small $n$ modes};
\node at (.5, -.5) {$=$};

\draw[thick, black] 
(2.4, -2) arc[start angle=270, end angle=90, radius=1.55] -- (2.4, 0.4) -- (2.4, -1.6) -- cycle;
\node at (1.65, -0.45) {$\rho'{(A_c,b')}$};

\draw[thick, black] (3, -.4) rectangle (3.8, 0.4);
\node at (3.4, 0) {\large $\Phi'$};

\draw[thick, black] (2.4, 0) -- (3, 0);
\draw[thick, black] (2.4, -1) -- (4.4, -1);

\draw[thick, black] (3.8, 0) -- (4.4, 0);

\node at (.5, -3.7) {$=$};

\draw[thick, black] 
(2, -4.8) arc[start angle=270, end angle=90, radius=1.1] -- (2, -2.8) -- (2, -2.8) -- cycle;
\node at (1.52, -3.7) {$\rho'$};

\draw[thick, black] (2, -3.2) -- (2.4, -3.2);
\draw[thick, black] (2, -4.2) -- (4.2, -4.2);

\draw[thick, black] (2.4, -3.6) rectangle (3.4, -2.8);
\node at (2.9, -3.2) {$D(\gamma)$};

\draw[thick, black] (3.4, -3.2) -- (3.8, -3.2);

\draw[thick, black] (3.8, -3.6) rectangle (5, -2.8);
\node at (4.4, -3.2) {$D(\gamma)^\dagger$};
\draw[thick, black] (5, -3.2) -- (5.5, -3.2);

\draw[thick, black] (5.5, -3.6) rectangle (6.3, -2.8);
\node at (5.9, -3.2) {\large $\Phi'$};
\draw[thick, black] (6.3, -3.2) -- (6.8, -3.2);


\node at (.5, -6.3) {$=$};

\draw[thick, black] 
(2, -7.4) arc[start angle=270, end angle=90, radius=1.1] -- (2, -5.4) -- (2, -5.4) -- cycle;
\node at (1.52, -6.3) {$\rho_{\mathrm{core}}$};

\draw[thick, black] (2, -5.8) -- (2.5, -5.8);
\draw[thick, black] (2, -6.8) -- (3.8, -6.8);

\draw[thick, black] (2.5, -6.2) rectangle (3.3, -5.4);
\node at (2.9, -5.8) {\large $\Phi$};
\draw[thick, black] (3.3, -5.8) -- (3.8, -5.8);
\end{tikzpicture}
\caption{Circuit diagram of the method for dealing with states with non-zero $b$. We start by applying the stellar decomposition as if $b$ was zero. This gives us the first equation, where $\rho'$ has the right A matrix ($A_{\mathrm{core}}$ represented by $A_c$ in the figure) but the it could be that $b'\neq0$. In this case (i.e., $b'\neq 0$), we consider absorbing a displacement into $\rho'$. As explained in this Appendix, we can always find $D(\gamma)$ such that $\rho_{\mathrm{core}} = (D(\gamma)\otimes \mathbb 1)\rho'$ has the right b vector $b_{\mathrm{core}}$.}\label{fig:disregard-displacements}
\end{figure}
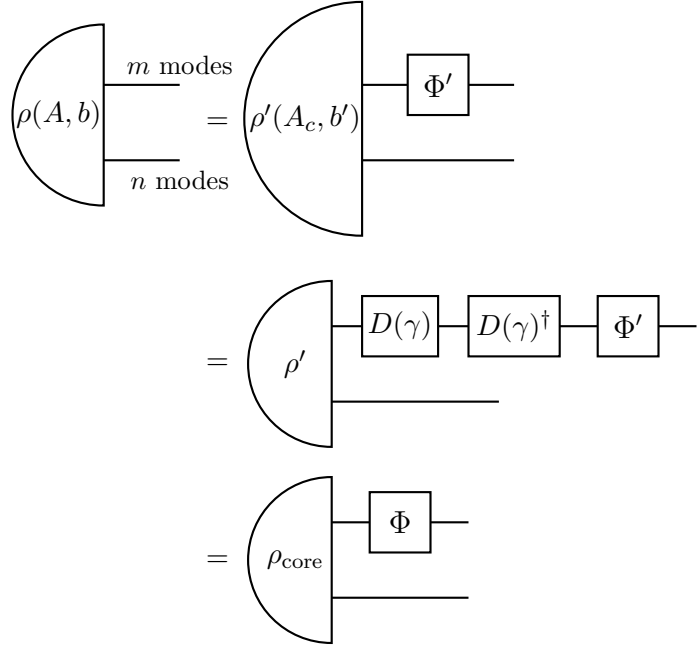


\subsection{Proof of \texorpdfstring{\protect\cref{prop:core-states-in-bargmann}}{}}\label{app:proof-of-ket-decomposition}

We first show that a Gaussian core state must satisfy \eqref{eq:core-state-block-form}. To see this, let $F_\rho(z_1,w_1,z_2,w_2)$ be the Bargmann function of $\rho$ in mode-wise order, and note that $F_\rho(z_1,w_1,0,0)$ defines the Bargmann function of $(\mathbb 1 \otimes \bra0) \rho (\mathbb 1 \otimes \ket0)$. Note that $(\mathbb 1 \otimes \bra{0})\rho(\mathbb 1 \otimes \ket0) \propto \ket0\bra0$ as vacuum is the only Gaussian state with finite Fock support, and therefore, we get that the top-left block of $A_\rho$ must be zero (note that the top-left block of $A$ is the A matrix of an object where all other modes are projected into zero).

For the other direction, we need to show that the block parametrization of \eqref{eq:core-state-block-form} implies the state is a Gaussian core state. We do this in two steps:
\begin{itemize}
    \item The pure state case: We consider $\ket{\psi_{\mathrm{core}}}$ with the following Abc parametrization:
\begin{align}\label{eq:Acore_def}
    A_{|\psi_\mathrm{core}\rangle}&=\begin{bmatrix}
    0 & R^{\mathrm{T}}\\
    R & A
    \end{bmatrix},\\\label{eq:bcore_def}
    b_{|\psi_\mathrm{core}\rangle}&=\begin{bmatrix}
    0\\
    b
    \end{bmatrix},
\end{align}
for some blocks $R$, $A$ and $b$.
The zero blocks guarantee that the Taylor coefficients of the Bargmann function parametrized by $A_{|\psi_\mathrm{core}\rangle}$, $b_{|\psi_\mathrm{core}\rangle}$ and $c_{|\psi_\mathrm{core}\rangle}$ terminate whenever the sum of the orders of the derivatives on the first $m$ variables is larger than the sum of the orders of the derivatives on the last $n$ variables \cite{miatto2020fast}. To see this, recall that the Fock amplitudes are
\begin{align}
G_{\ell} := \langle \ell | \psi_{\mathrm{core}} \rangle = \frac{1}{\sqrt{\ell!}}\partial^\ell F(0).
\end{align}
{where $\ket{\ell}$ is a Fock state for any tuple of integers $\ell\in\mathbb N^{m+n}$. In other words, we have $\ket{\psi_{\mathrm{core}}} = \sum_{\ell\in\mathbb N^{m+n}} G_\ell \ket{\ell}$}. Interestingly, these Fock amplitudes {$G_\ell$} satisfy a simple recursion relation detailed in \eqref{eq:recrel}. We exploit this recursion relation to complete the proof. {Let us write $\ell = (j,k)$ for $j\in\mathbb N^{m}, k\in\mathbb N^{n}$}. We claim that
\begin{align}
G_{j,k} = 0, \;\text{if} \; \norm{j}_1 > \norm{k}_1.
\end{align}
To see this, note that whenever $\|j\|_1 = \|k\|_1$, taking any more derivative with respect to the first $m$ variables (i.e., when $i \leq m$) results in a zero Fock amplitude:
\begin{align}\label{eq:fock_edge}
    G_{j+1_i,k} = \frac{1}{\sqrt{j_i+1}}\biggl[\underbrace{b_i}_0G_{j,k}+\sum_{p=1}^n\sqrt{k_{p}}R_{ip}\underbrace{G_{j,k-1_{p}}}_{0}\biggr],
\end{align}
where $G_{jk}$ is the Fock amplitude at the point $(j,k)$ in the $(m+n)$-dimensional Fock lattice.
The first contribution is zero by definition ($b_i = 0$ whenever $i \leq m$ because of the zero block in \eqref{eq:bcore_def}). That the second part is zero can be seen recursively for all values of $\|k\|_1$. Starting from $\|j\|_1=0$ and $\|k\|_1=0$ we have that $\sqrt{k_p}=0$ for all $p$ and so $G_{jk}=0$ for $i\leq m$, $\|j\|_1 = 1$ and $\|k\|_1=0$. The same holds for larger values of $\|j\|_1$, i.e. $G_{jk} = 0$ if $\|j\|_1 > 0$ and $\|k\|_1=0$. Then consider $G_{j+1_i,k}$ with $\|j\|_1=1$ and $\|k\|_1=1$. Again, the first contribution in \eqref{eq:fock_edge} is zero because $b_i=0$ for $i\leq m$ and the second contribution is zero because we just proved that $G_{jk} = 0$ when $\|j\|_1>0$ and $\|k\|_1=0$. As before, this holds for larger values of $\|j\|_1$, proving $G_{jk} = 0$ with $\|j\|_1 > 1$ and $\|k\|_1=1$. 
This argument can be repeated for all the subsequent values of $\|k\|_1$, proving that $G_{jk} = 0$ whenever $\|j\|_1>\|k\|_1$ in general.

\item Mixed state case: Note that by vectorizing our core state $\rho_{\mathrm{core}}$ (in the type-wise ordering), we have (see \cite{watrous2018theory})
\begin{align}
\mathrm{vec}\left( (\mathbb 1 \otimes \bra{k_1}) \rho_{\text{{core}}} (\mathbb 1 \otimes \ket{k_2})\right) = (\mathbb 1 \otimes \bra{k_1}\otimes \mathbb 1\otimes \bra{k_2}) \mathrm{vec}(\rho_{\text{{core}}}) \qquad \text{(type-wise ordering)}
\end{align}
Rewriting the above equation in mode-wise ordering, would simply swap the middle two spaces (see e.g. \eqref{eq:example-type-wise} and \eqref{eq:example-mode-wise}), and hence, we get
\begin{align}\label{eq:rho-core-pf}
\mathrm{vec}\left( (\mathbb 1 \otimes \bra{k_1}) \rho_{\mathrm{core}} (\mathbb 1 \otimes \ket{k_2})\right) = (\mathbb 1 \otimes \bra{k_1,k_2}) \mathrm{vec}(\rho_{\mathrm{core}}) \qquad \text{(mode-wise ordering)}
\end{align}
and hence, we would like to show that $\mathrm{vec}(\rho)$, with the Ab parametrization
\begin{align}
A_{\rho_{\mathrm{core}}} = \begin{bmatrix}
0 & R^{\mathrm{T}}\\
R & A
\end{bmatrix},
\quad
b_{\rho_{\mathrm{core}}} = \begin{bmatrix}
0\\
b
\end{bmatrix}
\end{align}
has the property that the right-hand-side of \eqref{eq:rho-core-pf} has a finite support in the Fock basis.\footnote{Note that $R$ is a $2n\times 2m$ and $A$ is a $2n\times 2n$ matrix.} This is guaranteed by our proof above for the pure states (note that we did not use the property that $\ket{\psi_{\mathrm{core}}}$ is of unit norm, and the only property we needed was the zero blocks \eqref{eq:Acore_def} and \eqref{eq:bcore_def}).
\end{itemize}

Lastly, we highlight that, as shown in our proof for pure states, the state
\begin{align}
(\mathbb 1 \otimes \bra{k})\ket{\psi_{\mathrm{core}}}
\end{align}
has Fock support at most on the set
\begin{align}
J_{\norm{k}_1}^{m}:=\{\ell\in\mathbb N^m: \norm{\ell}_1 \leq \norm{k}_1\}.
\end{align}
Similarly, the operator
\begin{align}
\tilde\rho = (\mathbb 1 \otimes \bra{k_1}) \rho_{\mathrm{core}}(\mathbb 1\otimes \ket{k_2})
\end{align}
satisfies
\begin{align}
\bra{\ell_1}\tilde\rho\ket{\ell_2} = 0
\end{align}
whenever
\begin{align}
\norm{\ell_1 + \ell_2}_1 \geq \norm{k_1+k_2}_1.
\end{align}
We have provided a pictorial representation of this fact for pure states in \protect\cref{fig:finite-cutoff-of-core-state}.

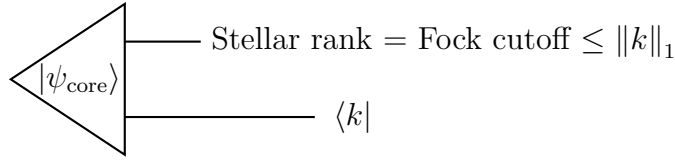
\begin{figure}[t]
\centering
\begin{tikzpicture}
    \draw[thick] (3.5,0) -- (5,1) -- (5,-1) -- cycle;
    \node at (4.4,0) {\large $\ket{\psi_{\text{core}}}$};
    \draw[thick] (5,0.5) -- (6,0.5);
    \draw[thick] (5,-0.5) -- (7.5,-0.5);
    \node at (8, -.5) {\large $\bra{k}$};
    \node at (8,.5) {\large Stellar rank $\leq \norm{k}_1$};
\end{tikzpicture}
\caption{A Gaussian core state over $MN$ with $|M| =m$ and $|N|=n$. By definition (\protect\cref{def:core}), such states have finite support in the photon number basis, whenever we perform a number-basis measurement on modes $N$. Moreover, we can show that measuring a Fock pattern $k\in\mathbb{N}^n$ on subsystem $N$ leaves the state on subsystem $M$ with Fock support on the set $J_{\norm{k}_1}^{m}$ of Fock states (see \eqref{eq:J-notation} for the definition of such sets). Recall that this means the stellar rank (or equivalently the degree of the polynomial stellar function associated to the state on subsystem $M$) is upper bounded by $\norm{k}_1$ i.e., the total number of measured photons. This figure is reported in the main text (Fig.~3).}
\label{fig:finite-cutoff-of-core-state}
\end{figure}

\subsection{Proof of \texorpdfstring{\protect\cref{prop:pure-state-case}}{}}\label{app:pf-of-pure-state}
Our proof is structured as follows: we start from a generic $(m+n)$-mode pure Gaussian state, and we construct a Gaussian core state $|\psi_\mathrm{core}\rangle$ and a Gaussian unitary $U$ such that when the unitary is applied to the first $m$ modes of the Gaussian core state we obtain the original Gaussian state.

Consider then an initial Gaussian $(m+n)$-mode state $|\psi\rangle$ with Abc parametrization
\begin{align}\label{eq:initial_ketA}
    A_{|\psi\rangle}&=\begin{bmatrix}
    A_\psi^{(m)} & R_\psi^{\mathrm{T}}\\
    R_\psi & {A_\psi^{(n)}}
    \end{bmatrix},\\\label{eq:initial_ketb}
    b_{|\psi\rangle}&=\begin{bmatrix}
    {b_\psi^{(m)}}\\
    {b_\psi^{(n)}}
    \end{bmatrix},
\end{align}
where we have specified the separation between the first $m$ modes and the remaining $n$ modes using blocks of appropriate size: $A_\psi^{(m)}$ is $m\times m$, $R_\psi$ is $n\times m$ and ${A_\psi^{(n)}}$ is $n\times n$. ${b_\psi^{(m)}}$ is an $m$-dimensional complex vector and ${b_\psi^{(n)}}$ is an $n$-dimensional complex vector. We omit $c_{|\psi\rangle}$ as it is determined by normalization.

Recall that if we project the last $n$ modes of $|\psi_\mathrm{core}\rangle$ onto $|k\rangle$, the corresponding marginal $|\mathrm{core}_k\rangle$ on the first $m$ modes will have Fock support on $J_{\norm{k}_1}^{m}$. In particular, if we project the last $n$ modes onto vacuum, the marginal on the first $m$ modes must be the $m$-mode vacuum. It follows that the unitary $U$ that we seek must map the $m$-mode vacuum to the pure Gaussian state that results from projecting the last $n$ modes of the original state $|\psi\rangle$ onto vacuum. In other words, we want $U$ to satisfy
\begin{align}
(\mathbb 1 \otimes \bra{0^n})\ket{\psi} = U(\mathbb 1 \otimes \bra{0^n}) \ket{\psi_{\mathrm{core}}}.
\end{align}
Since $\ket{\psi_{\mathrm{core}}}$ is a Gaussian core state, we have that
\begin{align}
(\mathbb 1 \otimes \bra{0^n})\ket{\psi_{\mathrm{core}}} \propto \ket{0^m}.
\end{align}
As a result, we have
\begin{align}\label{eq:elementary-building-block}
(\mathbb 1 \otimes \bra{0^n})\ket{\psi} \propto U\ket{0^m}.
\end{align}
The fact that vacuum projections are elementary operations in the Bargmann formalism (see \protect\cref{app:inner-product-in-bargmann}), allows us to use \eqref{eq:elementary-building-block} to compute the Gaussian unitary $U$ straightforwardly. 

\begin{figure}[htbp]
\centering
\begin{tikzpicture}
    \draw[thick] (3,0) -- (4.5,1) -- (4.5,-1) -- cycle;
    \node at (4,0) {\large $\ket{\psi}$};
    \draw[thick] (4.5,0.5) -- (5.5,0.5);
    \draw[thick] (4.5,-0.5) -- (5.5,-0.5);
    \node at (6, -0.5) {\large $\bra{0^n}$};
    \node at (7, 0.0) {\large $=$};
    \draw[thick] (7.5,0) -- (9,1) -- (9,-1) -- cycle;
    \node at (8.4,0) {\large $\ket{\psi_\mathrm{core}}$};
    \draw[thick] (9,0.5) -- (10,0.5);
    \draw[thick] (9,-0.5) -- (10,-0.5);
    \node at (10.5, -0.5) {\large $\bra{0^n}$};
    \draw[thick] (10,1) -- (11,1) -- (11,0) -- (10,0) -- cycle;
    \node at (10.5, 0.5) {\large $U$};
    \draw[thick] (11,0.5) -- (12,0.5);
    \node at (12.5, 0.0) {\large $=$};
    \node at (13.5, 0.5) {\large $|0^m\rangle$};
    \draw[thick] (14,0.5) -- (15,0.5);
    \draw[thick] (15,1) -- (16,1) -- (16,0) -- (15,0) -- cycle;
    \node at (15.5, 0.5) {\large $U$};
    \draw[thick] (16,0.5) -- (17,0.5);
\end{tikzpicture}
\caption{A visual proof that the unitary of the pure stellar decomposition maps the $m$-mode vacuum to the heralded state $\langle0^n|\psi\rangle$ corresponding to measuring vacuum on the last $n$ modes of $|\psi\rangle$ (Fig. \ref{fig:proof_of_U} in the main text).} 
\end{figure}

\paragraph{Gaussian unitary.}
We now find the Abc parametrization of such $U$, and note that the requirement uniquely identifies $U$ up to an $m$-mode passive Gaussian transformation that preserves the photon number (interferometer). To this end, note that the A matrix of the left-hand-side of \eqref{eq:elementary-building-block} is $A_\psi^{(m)}$ (i.e., the top-left block in \eqref{eq:initial_ketA}). Let $A_U$ and $b_U$ be parametrized as follows:
\begin{align}\label{eq:AU}
    A_U &= \begin{bmatrix}
        A_U^\mathrm{out} & \Gamma_U^{\mathrm{T}}\\\Gamma_U & A_U^\mathrm{in}
    \end{bmatrix}=\begin{bmatrix}
        x & \sqrt{\mathbb{1} - xx^*}\\\sqrt{\mathbb{1}-x^*x} & -x^*
    \end{bmatrix},\\\label{eq:bU}
    b_U &= \begin{bmatrix}
        b_U^\mathrm{out}\\
        b_U^\mathrm{in}
    \end{bmatrix}=\begin{bmatrix}
        \Gamma_U^{\mathrm{T}}\gamma\\
        A_U^\mathrm{in}\gamma - \gamma^*
    \end{bmatrix},
\end{align}
for some matrix $x\in\mathbb C^{m\times m}$ and complex vector $\gamma\in\mathbb C^m$.
We can write the four $m\times m$ blocks in \eqref{eq:AU} this way because the $A$ matrix of a Gaussian unitary is itself unitary and symmetric \cite{miatto2020fast}. The blocks in \eqref{eq:bU} are obtained by composing $U$ with a displacement by $\gamma\in\mathbb{C}^m$, which will be chosen such that $U$ applied to $|\psi_\mathrm{core}\rangle$ reproduces the displacement of the original state.

Therefore the condition that $U$ maps vacuum to the state with A matrix of $A_\psi^{(m)}$ and $b$ of ${b_\psi^{(m)}}$ implies 
\begin{align}
A_U^\mathrm{out}=x=A_\psi^{(m)}, \quad b_U^\mathrm{out}=\Gamma_U^{\mathrm{T}}\gamma = {b_\psi^{(m)}}.
\end{align}
This is enough to identify $U$:
\begin{align}\label{eq:AU_final}
\begin{split}
    A_U &= \begin{bmatrix}
        A_\psi^{(m)} & \Gamma_U^{\mathrm{T}}\\\Gamma_U & -{A_\psi^{(m)}}^*
    \end{bmatrix},\\
    b_U &= \begin{bmatrix}
        {b_\psi^{(m)}}\\
        -{A_\psi^{(m)}}^*\Gamma_U^{-1}{b_\psi^{(m)}} - \Gamma_U^{-T}{b_\psi^{(m)}}^*
    \end{bmatrix},
\end{split}
\end{align}
with $\Gamma_U = \sqrt{\mathbb{1} - {A_\psi^{(m)}}^*A_\psi^{(m)}}$ and $c_U$ determined by the unitarity condition (up to a global phase).

\paragraph{Gaussian core state} Following \protect\cref{prop:core-states-in-bargmann}, we can parametrize the Abc representation of our Gaussian core state $\ket{\psi_{\mathrm{core}}}$ as
\begin{align}
A_{\ket{\psi_{\mathrm{core}}}} = \begin{bmatrix}
0 & R_c^{\mathrm{T}}\\
R_c & A_c^{(n)}
\end{bmatrix},
\quad b_{\ket{\psi_{\mathrm{core}}}} = \begin{bmatrix}
0\\
b_c^{(n)}
\end{bmatrix},
\end{align}
where $R_c$ is an $n\times m$ and $A_c^{(n)}$ is an $n\times n$ matrix, with $b_c^{(n)}$ being a vector of size $n$.

There are two ways to find the unknown $R_c$, $A_c^{(n)}$ and $b_c^{(n)}$. One is to apply $U^\dagger$ to the first $m$ modes of $|\psi\rangle$. The other is to apply $U$ to the first $m$ modes of $|\psi_\mathrm{core}\rangle$ and set the result equal to $|\psi\rangle$. We follow the second approach, but both are possible.
Using \eqref{eq:gauss_intA} and \eqref{eq:gauss_intb} to compute the action of $U$ on the first $m$ modes of $|\psi_\mathrm{core}\rangle$ and we obtain
\begin{align}
\begin{split}
A_{U|\psi_\mathrm{core}\rangle} &= \begin{bmatrix}
    A_U^\mathrm{out} & \Gamma_U^{\mathrm{T}}R_c^{\mathrm{T}}\\R_c\Gamma_U & A_c^{(n)} + R_c A_U^\mathrm{in}R_c^{\mathrm{T}}
\end{bmatrix},\\
b_{U|\psi_\mathrm{core}\rangle} &= \begin{bmatrix}
    b_U^\mathrm{out}\\ R_c b_U^\mathrm{in} + b_c^{(n)}
\end{bmatrix},\\
c_{U|\psi_\mathrm{core}\rangle} &= c_Uc_{|\psi_\mathrm{core}\rangle}.
\end{split}
\end{align}
This ought to be equal to the Abc parametrization of the original state $|\psi\rangle$, which implies $R_c=R_\psi\Gamma_U^{-1}$, $A_c^{(n)}={A_\psi^{(n)}} + R_cA_m^*R_c^{\mathrm{T}}$ and $b_c^{(n)} = {b_\psi^{(n)}} - R_c b_U^\mathrm{in}$, yielding
\begin{align}\label{eq:Acore}
    A_{|\psi_\mathrm{core}\rangle} &= \begin{bmatrix}
        0 & \Gamma_U^{-T}R_\psi^{\mathrm{T}}\\
        R_\psi\Gamma_U^{-1} & {A_\psi^{(n)}} + R_\psi({{A_\psi^{(m)}}^*}^{-1} - {A_\psi^{(m)}})^{-1}R_\psi^{\mathrm{T}}
    \end{bmatrix},\\
    b_{|\psi_\mathrm{core}\rangle} &= \begin{bmatrix}
        0\\
        {b_\psi^{(n)}} + R_\psi\Gamma_U^{-1}\left( A_\psi^{(m)}{}^\ast \Gamma_U^{-1} b^{(m)}_\psi + \Gamma_U^{-\mathrm{T}} b_\psi^{(m)} {}^\ast\right)
    \end{bmatrix},
\end{align}
with $\Gamma_U = \sqrt{\mathbb{1}-{A_\psi^{(m)}}^*{A_\psi^{(m)}}}$, $b_U^\mathrm{in}=-{A_\psi^{(m)}}^*\Gamma_U^{-1}{b_\psi^{(m)}} - \Gamma_U^{-T}{b_\psi^{(m)}}^*$ and $c_{|\psi_\mathrm{core}\rangle}$ set by normalization. We have the freedom to distribute the phase of $c_{|\psi\rangle}$ between $c_U$ and $c_{|\psi_\mathrm{core}\rangle}$ as we wish.


\subsection{Proof of \texorpdfstring{\protect\cref{prop:mixed-state-case}}{}}\label{app:pf-of-mixed-state}
The initial state has Abc parameters in mode-wise order (see \protect\cref{app:ordering-conventions}):
\begin{align}
\begin{split}
A_\rho &= \begin{bmatrix}
A_\rho^{(m)} & R_\rho^{\mathrm{T}}\\
R_\rho & A_\rho^{(n)}
\end{bmatrix},\\
b_\rho &= \begin{bmatrix}
{b_\rho^{(m)}} \\{b_\rho^{(n)}}
\end{bmatrix}.
\end{split}
\end{align}

Note that ${A_\rho^{(m)}}$ is $2m\times 2m$, $R_\rho$ is $2n\times 2m$, ${A_\rho^{(n)}}$ is $2n\times 2n$, ${b_\rho^{(m)}}$ is a $2m$-vector and ${b_\rho^{(n)}}$ is a $2n$-vector.

We start with a core density matrix and a Gaussian channel with Bargmann matrices in mode-wise order:
\begin{align}\label{eq:Arho_core}
    A_{\rho_\mathrm{core}} &= \begin{bmatrix}
        0 & R_c^{\mathrm{T}}\\
        R_c & A_c^{(n)}
    \end{bmatrix},\\
    A_\Phi &= \begin{bmatrix}
        A_\Phi^\mathrm{out} & \Gamma_\Phi^{\mathrm{T}}\\
        \Gamma_\Phi & A_\Phi^\mathrm{in}
    \end{bmatrix},
\end{align}
where the blocks of $A_{\rho_\mathrm{core}}$ have the same size as the blocks of $A_\rho$ and the blocks of $A_\Phi$ are all $2m\times 2m$.

Just like in the proof of our pure state's case (i.e. \protect\cref{prop:pure-state-case}), we apply the channel to the core state and set the result equal to the initial state obtaining
\begin{align}\label{eq:mixed_system}
    \begin{cases}
        A_\Phi^\mathrm{out} &= A^{(m)}_\rho\\
        R_c\Gamma_\Phi &= R_\rho\\
        A_c^{(n)} + R_cA_\Phi^{\mathrm{in}}R_c^{\mathrm{T}} &= A_\rho^{(n)}.
    \end{cases}
\end{align}

The first equation gives us $A_\Phi^\mathrm{out}$ which we can interpret as the fact that the channel maps the $m$-mode vacuum to the $m$-mode state resulting from measuring vacuum on the last $n$ modes of the initial state. Using the TP condition of channels (see \protect\cref{prop:TP}), we set
\begin{align}\label{eq:TP-in-pf}
A_\Phi^{\mathrm{in}} = R_\Phi ( A_\Phi^{\mathrm{out}} - X)^{-1} R_\Phi^{\mathrm{T}} + X
\end{align}
the first and last equations can be combined into
\begin{align}\label{eq:bargman-marginal}
A_c^{(n)}+R_cXR_c^{\mathrm{T}} = {A_\rho^{(n)}} - R_\rho(A_\rho^{(m)}-X)^{-1}R_\rho^{\mathrm{T}}.
\end{align}
We note that \eqref{eq:bargman-marginal} has an expected physical meaning: the initial state and the core state must have the same marginals on the last $n$ modes. This is expected since the channel is trace-preserving, and tracing the output of the channel is the same as tracing the first $m$ modes of the core state directly (see \protect\cref{fig:TP}). Note that since $A_\Phi^{\mathrm{out}} = A_\rho$, we have that $\mathbb 1 - XA_\Phi^{\mathrm{out}} > 0$ (see \protect\cref{rem:physicality}), and hence, by setting \eqref{eq:TP-in-pf}, we are guaranteed that $\Phi$ is trace-preserving (note that we eliminate discussions about the `c' parameters as it is straightforward and can be calculated at the end). 

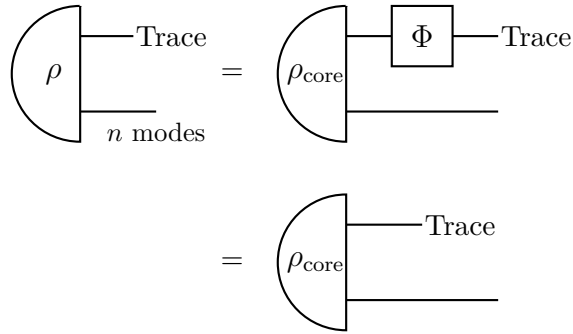
\begin{figure}[htbp]
    \centering
    \begin{tikzpicture}
\draw[thick, black] 
(-1.5, -1.4) arc[start angle=270, end angle=90, radius=0.9] -- (-1.5, 0.4) -- (-1.5, -1.4) -- cycle;
\node at (-1.85, -0.5) {\large $\rho$};

\draw[thick, black] (-1.5, 0) -- (-.8, 0);
\draw[thick, black] (-1.5, -1) -- (-0.5, -1) node[below] {\small $n$ modes};
\node at (.5, -.5) {$=$};

\node at (-.3,0) {Trace};

\draw[thick, black] 
(2, -1.4) arc[start angle=270, end angle=90, radius=0.9] -- (2, 0.4) -- (2, -1.4) -- cycle;
\node at (1.6, -0.5) {\large $\rho_{\text{core}}$};

\draw[thick, black] (2.6, -.4) rectangle (3.4, 0.4);
\node at (3, 0) {\large $\Phi$};

\draw[thick, black] (2, 0) -- (2.6, 0);
\draw[thick, black] (2, -1) -- (4, -1);

\draw[thick, black] (3.4, 0) -- (4, 0);
\node at (4.5,0) {Trace};

\node at (.5, -3) {$=$};

\draw[thick, black] 
(2, -3.9) arc[start angle=270, end angle=90, radius=0.9] -- (2, -2.1) -- (2, -3.9) -- cycle;
\node at (1.6, -3) {\large $\rho_{\text{core}}$};

\draw[thick, black] (2, -2.5) -- (3, -2.5);
\draw[thick, black] (2, -3.5) -- (4, -3.5);
\node at (3.5,-2.5) {Trace};
\end{tikzpicture}
    \caption{The TP condition for the channel $\Phi$ contributes to identifying the core state as one that has the same $n$-mode marginal as the initial state $\rho$.}
    \label{fig:TP}
\end{figure}

Note that we require $\rho_{\mathrm{core}}$ to be physical. This means that it should have finite trace and also be positive. The finite trace condition is implied by assuming $\rho$ has a finite trace and that $\Phi$ is TP. As we already have imposed these conditions, it follows that $\rho_{\mathrm{core}}$ has finite trace. It remains to check whether $\rho_{\mathrm{core}}$ is positive. From hereon, for any $A$ matrix in mode-wise order (which is the convention throughout this proof), we assume without loss of generality the following block form
\begin{align}
A = \begin{bmatrix}
a^\ast & \alpha^\ast\\
\alpha & a
\end{bmatrix}.
\end{align}
This convention applies to submatrices such as $A_\rho^{(m)}$, $A_\rho^{(n)}$, etc. Similarly, for any $R$ matrix, we assume
\begin{align}
R = \begin{bmatrix}
r^\ast & \sigma^\ast\\
\sigma & r
\end{bmatrix}.
\end{align}
We highlight that $r_\rho, \sigma_\rho$ and $r_c, \sigma_c$ are all $n\times m$ matrices. With these conventions in mind, we employ \protect\cref{lem:positivity} to write the positivity condition of $\rho_{\mathrm{core}}$ as
\begin{align}
\begin{bmatrix}
0 & \sigma_c^\dag\\
\sigma_c & \alpha_c^{(n)}
\end{bmatrix}\geq 0.
\end{align}
This is equivalent to
\begin{align}\label{eq:positivity-of-rho-c}
\sigma_c = 0, \quad \alpha_c^{(n)} \geq 0.
\end{align}
Plugging $\sigma_c = 0$ in the second equation of \eqref{eq:mixed_system} we get stringent conditions on $R_\rho$. Let blocks of $\Gamma_\Phi$ be expressed as
\begin{align}
\Gamma_\Phi = \begin{bmatrix}
\gamma_\Phi^\ast & \lambda_\Phi^\ast\\
\lambda_\Phi & \gamma_\Phi
\end{bmatrix},
\end{align}
so we can rewrite $R_c \Gamma_\Phi = R_\rho$ as
\begin{align}
r_c \gamma_\Phi &= r_\rho,\label{eq:overcomplete-1}\\
r_c \lambda_\Phi &= \sigma_\rho, \label{eq:overcomplete-2}
\end{align}
which immediately implies
\begin{align}
\sigma_\rho \sigma_\rho^\dag + r_\rho r_\rho^\dag = r_c(\gamma_\Phi \gamma_\Phi^\dag + \lambda_\Phi \lambda_\Phi^\dag) r_c^\dag.
\end{align}
As $r_c$ is an $n\times m$ matrix, we get
\begin{align}\label{eq:mixed-stellar-condition}
\mathrm{rank}(\sigma_\rho \sigma_\rho^\dag + r_\rho r_\rho^\dag) \leq m.
\end{align}
Therefore, we have ruled out the possibility of a physical stellar decomposition when \eqref{eq:mixed-stellar-condition} is violated (i.e., the `only if' direction). Note that \eqref{eq:mixed-stellar-condition} is an extremely stringent condition if $n > m$ since almost all density matrices in this case satisfy $\mathrm{rank }(\sigma_\rho\sigma_\rho^\dag + r_\rho r_\rho^\dag) = 2m$. However, as the rank is upper bounded by the dimension, we have $\mathrm{rank}(\sigma_\rho\sigma_\rho^\dag + r_\rho r_\rho^\dag) \leq n$, which implies that if $n\leq m$ the condition is already satisfied. In what follows, we prove that if \eqref{eq:mixed-stellar-condition} is satisfied, one can indeed find a set of physical core states and physical channels that can satisfy the decomposition. 

\begin{claim}\label{claim:solution-for-mixed-case}
Assuming \eqref{eq:iff-condition}, the following has a solution for $r_c$:
\begin{align}\label{eq:r_c-solution}
r_c r_c^\dag = \begin{bmatrix}
0 & \mathbb 1
\end{bmatrix} R_\rho (X - A_\rho^{(m)})^{-1} R_\rho^{\mathrm{T}}\begin{bmatrix}
\mathbb 1\\
0
\end{bmatrix} + \sigma_\rho \alpha_\rho^{(m)}{}^{+} \sigma_\rho^\dag.
\end{align}
Such a solution for $r_c$, together with
\begin{align}
\begin{split}
\sigma_c &= 0\\
\alpha_c^{(n)} &= \alpha_\rho^{(n)} - \sigma_\rho \alpha^{(m)}_\rho {}^{+} \sigma_\rho^\dag\\
a_c^{(n)} &= a_\rho^{(n)} + \begin{bmatrix}
0 & \mathbb 1
\end{bmatrix}R_\rho(X - A^{(m)})^{-1} R_\rho^{\mathrm{T}}
\begin{bmatrix}
0\\
\mathbb 1
\end{bmatrix}\\
A_\Phi^{\mathrm{out}} &= A_\rho^{(m)}\\
\Gamma_\Phi &= R_c^+ R_\rho\\
A_\Phi^{\mathrm{in}} &= \Gamma_\Phi (A_\Phi^{\mathrm{out}} - X)^{-1} \Gamma_\Phi^{\mathrm{T}} + X,
\end{split}
\end{align}
where $\cdot^+$ indicates pseudo-inverse, form a physical solution for the stellar decomposition of $\rho$. 
\end{claim}
Let us first show that \eqref{eq:r_c-solution} has a solution. Note that if we measure modes in the set $N$ of $\rho$ all in the vacuum state (to get a state on the set $M$ of modes), we get a density matrix whose A matrix is $A_\rho^{(m)}$. This state must be physical, which implies $\mathbb 1 -  XA_\rho^{(m)} > 0$ by \protect\cref{rem:physicality}. Moreover, note that $XR_c^{\mathrm{T}} X = R_c^\dag$. As a result, we have
\begin{align}
\begin{split}
&\begin{bmatrix}
0 & \mathbb 1
\end{bmatrix} R_\rho (X - A_\rho^{(m)})^{-1} R_\rho^{\mathrm{T}}\begin{bmatrix}
\mathbb 1\\
0
\end{bmatrix}\\
&=
\begin{bmatrix}
0 & \mathbb 1
\end{bmatrix} R_\rho (1 - XA_\rho^{(m)})^{-1} R_\rho^\dag\begin{bmatrix}
0\\
\mathbb 1
\end{bmatrix}\geq 0.
\end{split}
\end{align}
This implies that RHS of \eqref{eq:r_c-solution} is indeed positive semi-definite. Therefore, a solution $r_c$ exists if either $n\geq m$, or if $n<m$ but the RHS has a rank less than $m$ (this is because we are taking $r_c$ to have dimensions of $n\times m$). As discussed above, $n\geq m$ always satisfies \eqref{eq:iff-condition}, and therefore, we will focus on $n<m$ and show the rank condition in that case: notice that \eqref{eq:iff-condition} implies
\begin{align}\label{eq:P_m}
R_\rho = \begin{bmatrix}
\Pi_m^\ast & 0\\
0 & \Pi_m
\end{bmatrix}
R_\rho,
\end{align}
where $\Pi_m\in\mathbb C^{n\times n}$ is some projector of rank $m$. Utilizing this, we get
\begin{align}
\begin{split}
&\begin{bmatrix}
0 & \mathbb 1
\end{bmatrix} R_\rho (X - A_\rho^{(m)})^{-1} R_\rho^{\mathrm{T}}\begin{bmatrix}
\mathbb 1\\
0
\end{bmatrix}
=
\Pi_m\begin{bmatrix}
0 & \mathbb 1
\end{bmatrix}
R_\rho (1 - XA_\rho^{(m)})^{-1} R_\rho^\dag\begin{bmatrix}
0\\
\mathbb 1
\end{bmatrix}
\Pi_m.
\end{split}
\end{align}
Furthermore, note that $\sigma_c = \Pi_m \sigma_c$ is implied by \eqref{eq:P_m}. Putting all together, we have shown the rank condition, and therefore, \eqref{eq:r_c-solution} has a solution.

Now, we need to verify that the proposed solution satisfies positivity of $\rho_{\mathrm{core}}$, as well as complete positivity of $\Phi$. Recall that positivity of $\rho$ is equivalent to (see \protect\cref{lem:positivity})
\begin{align}
\begin{bmatrix}
\alpha_\rho^{(m)} & \sigma_\rho^\dag\\
\sigma_\rho & \alpha^{(n)}_\rho
\end{bmatrix} \geq 0,
\end{align}
which is equivalent to $\alpha^{(n)}_\rho \geq \sigma_\rho \alpha_\rho^{(m)} {}^{+} \sigma_\rho$ and $\alpha_\rho^{(m)}\geq 0$ (see, \protect\cref{lem:block-pos}). Therefore, with our choice in the claim we get that $\alpha_c^{(n)}\geq0$. As we discussed before in \eqref{eq:positivity-of-rho-c} this implies that $\rho_c\geq0$. 

Note that $R_cR_c^+ R_\rho = R_\rho$ which also implies that with our choice of $\Gamma_\Phi$ in the claim, we get $R_c\Gamma_\Phi = R_\rho$. Furthermore, the equation for $a^{(n)}_c$ imposes equality for diagonal blocks of \eqref{eq:bargman-marginal}, while the equations for $r_c$ and $\alpha_c^{(n)}$ makes sure that the off-diagonals block are equal. Hence, we have satisfied all equations in \eqref{eq:mixed_system} and all that is left to check, is complete positivity of $\Phi$. To this end, note that with our choice of parameters, the A matrix of $\Phi$ can be decomposed as
\begin{align}\label{eq:decomposed-a-phi}
A_\Phi = 
\begin{bmatrix}
\mathbb 1 & 0\\
0 & R_c^+
\end{bmatrix}
\begin{bmatrix}
A_\rho^{(m)} & R_\rho^{\mathrm{T}}\\
R_\rho & A_\rho^{(n)} - A_c^{(n)}
\end{bmatrix}
\begin{bmatrix}
\mathbb 1 & 0\\
0 & R_c^+ {}^{\mathrm{T}}
\end{bmatrix},
\end{align}
where, to obtain the bottom-right block, we have again used \eqref{eq:bargman-marginal}\footnote{Note that here we have assumed $\mathrm{rank}(R_c) = m$, so that $\mathbb 1_m = R_c^+ R_c$. If this was not the case, we would end up adding a positive semi-definite (indeed a projector) term to the bottom-right block of \eqref{eq:phi-physicality-break-down} which keeps the positivity argument intact.}.
Note that \eqref{eq:decomposed-a-phi} shows how one can relate $A_\Phi$ to $A_\rho$. From this, we show how complete positivity of $\Phi$ can be related to positivity of $\rho$. From \protect\cref{lem:cp-condition-in-appendix}, the matrix we need to check for complete positivity of $\Phi$ is
\begin{align}\label{eq:phi-physicality-break-down}
\begin{bmatrix}
\mathbb 1 & 0\\
0 & r_c^+
\end{bmatrix}
\begin{bmatrix}
\alpha_\rho^{(m)} & \sigma_\rho^\dag\\
\sigma_\rho & \alpha^{(n)}_\rho - \alpha^{(n)}_c
\end{bmatrix}
\begin{bmatrix}
\mathbb 1 & 0\\
0 & r_c^+ {}^\dag
\end{bmatrix} \overset{?}{\geq} 0,
\end{align}
which is indeed satisfied by our choice of $\alpha^{(n)}_c$ (see \protect\cref{lem:block-pos}).

\subsection{Proof of \texorpdfstring{\protect\cref{prop:m>=n}}{}}\label{app:pf-of-m>=n}

To show this, first let us point out that it is sufficient to consider $m=n$ as if $m>n$ one can always add $m-n$ many vacuum states to the second partition of the state and perform the pure state stellar decomposition associated to the evenly partitioned state. It is clear that the extra modes remain in the vacuum state after the decomposition and therefore can be removed in the end. Hence, in what follows, we restrict our attention to the $m=n$ case.

Note that a pure Gaussian core state in mode-wise order:
\begin{align}
    A_{|\psi_\mathrm{core}\rangle\langle\psi_\mathrm{core}|} &= \begin{bmatrix}
        0&R_c^{\mathrm{T}}\\
        R_c&A_c^{(n)}
    \end{bmatrix},\\
    b_{|\psi_\mathrm{core}\rangle\langle\psi_\mathrm{core}|} &= \begin{bmatrix}
        0\\b_c^{(n)}
    \end{bmatrix},
\end{align}
has $A_c^{(n)}=\begin{bmatrix}
    a^* & 0\\0&a
\end{bmatrix}$, $R_c=\begin{bmatrix}
    r^* & 0\\0&r
\end{bmatrix}$ and $b_c^{(n)}=\begin{bmatrix}
    \beta^*\\\beta
\end{bmatrix}$ for some $n\times n$ block $a$, $n\times m$ block $r$ and $n$-vector $\beta$.
Then, the condition in Eq.~\eqref{eq:bargman-marginal} becomes
\begin{align}
    \begin{bmatrix}
        a^* & r^*r^{\mathrm{T}}\\
        rr^\dagger& a
    \end{bmatrix}={A_\rho^{(n)}} - R_\rho({A_\rho^{(m)}}-X)^{-1}R_\rho^{\mathrm{T}},
\end{align}
which can be solved independently for $a$ and $r$ (since $r r^\dag$ can be full-rank in this case), and finally one can use the second equation in \eqref{eq:mixed_system} to find $\Gamma_\Phi = R_c^{+}R_\rho$. The physicality conditions are automatically satisfied as a pure state trivially implies a positive semi-definite density matrix, and the trace condition had been considered in trace-preserving condition of $\Phi$.

\subsection{Proof of \texorpdfstring{\protect\cref{prop:formal-stellar-decomp}}{}}\label{app:pf-of-formal-decomposition}

Given a Gaussian object with Abc parametrization $A_G, b_G, c_G$, we can partition the $A$ matrix and $b$ vector into blocks of appropriate size (which in general may not represent physical objects)
\begin{align}
    A_G &= \begin{bmatrix}
        A_m &  R^{\mathrm{T}}\\
        R & A_n
    \end{bmatrix}, \quad
    b_G = \begin{bmatrix}
        b_m\\b_n
    \end{bmatrix},
\end{align}
and then write the decomposition as
\begin{align}
\begin{split}
    A_{S_{\mathrm{c}}} &= \begin{bmatrix}
        0 &  R^{\mathrm{T}}\\
        R & A_n
    \end{bmatrix}, \quad
    b_{S_{\mathrm{c}}} = \begin{bmatrix}
    0\\b_n
    \end{bmatrix},\\
    A_T &= \begin{bmatrix}
        A_m & \mathbb{1} \\
        \mathbb{1} & 0
    \end{bmatrix}, \quad
    b_T = \begin{bmatrix}
    b_m\\0
    \end{bmatrix}.
\end{split}
\end{align}
Note that the operator $S$ has the property that its Bargmann function has a Taylor series that terminates whenever the order on the first $m$ variables exceeds the order on the last $n$ variables, which can be interpreted as its Fock amplitudes terminating whenever the photon number on the first $m$ Hilbert spaces exceeds the photon number on the last $n$ Hilbert spaces. The $T$ operator has divergent Fock coefficients (because the $A_T$ matrix has eigenvalues outside of the unit disc), and should be fused with other components before making sense of its Fock expansion.

As a note, we highlight that the contraction of $S_c$ and $T$ is well-defined as the contraction condition in \protect\cref{prop:contraction-condition} (in this case $\norm{0+0} = 0 <2$) is satisfied.

\subsection{Proof of \texorpdfstring{\protect\cref{prop:sdp}}{}}\label{app:pf-of-sdp}

Recall that a Gaussian state is fully characterized by its first two moments \cite{serafini2017quantum}, which we denote by $(\Sigma, \mu)$. An $N$-mode Gaussian state is physical if and only if
\begin{align}
\frac2\hbar\Sigma \geq i\Omega_N,
\end{align}
where $\Omega_N = \bigoplus_{i=1}^N \Omega_1$ with
\begin{align}
\Omega_1 = \begin{bmatrix}
0 & 1\\
-1 & 0
\end{bmatrix}.
\end{align}
Note that we are employing the so-called XPXP ordering (which is introduced and compared to the other XXPP ordering in \cite{chadwick2022classical}). Moreover, one can characterize a Gaussian channel $\Phi$ by two matrices and a vector $(X,Y,d)$, where $X,Y \in \mathbb R^{2N\times 2N}$ and $d\in\mathbb R^{2N}$. Recall that the action of an $N$-mode Gaussian channel on a Gaussian state parametrized by $(\Sigma,\mu)$ is given by
\begin{align}\Phi: (\Sigma,\mu) \mapsto (X\Sigma X^{\mathrm{T}} + Y, X\mu + d).
\end{align}
Also, a Gaussian channel is CPTP if and only if $\frac{2}{\hbar}Y + iX \Omega_N X^{\mathrm{T}} \geq i\Omega_N$. We refer the reader to \cite{serafini2017quantum} for an in-depth introduction to the phase-space formulation of Gaussian objects.

We say that a one-mode channel $\Phi$ can be \textit{factored out} from an $(1+n)$-mode state $\rho$ if $\rho = (\Phi \otimes \mathcal I)(\sigma)$ for some density matrix $\sigma$. Let $X,Y\in\mathbb R^{2\times 2}$ describe $\Phi$, and $(\Sigma, \mu)$ and $(\Sigma', \mu')$ represent $\rho$ and $\sigma$, respectively. The equation $\rho = (\Phi\otimes\mathcal I) (\sigma)$ translates into
\begin{align}
\Sigma = \begin{bmatrix}
X & 0\\
0 & \mathbb 1
\end{bmatrix}
\Sigma' 
\begin{bmatrix}
X^{\mathrm{T}} & 0\\
0 & \mathbb 1
\end{bmatrix}
+ \begin{bmatrix}
Y & 0\\
0 & 0
\end{bmatrix}.
\end{align}
Using the physicality condition of $\sigma$, we have $\frac2\hbar\Sigma'\geq i\Omega_N$, which gives
\begin{align}
\Sigma \geq i\frac\hbar2  \begin{bmatrix}
X & 0\\
0 & \mathbb 1
\end{bmatrix}
\Omega_N
\begin{bmatrix}
X^{\mathrm{T}} & 0\\
0 & \mathbb 1
\end{bmatrix}
+
\begin{bmatrix}
Y & 0\\
0 & 0
\end{bmatrix}.
\end{align}
Asking for $\Phi$ to be CPTP, we need to add the condition $i X \Omega_1 X^{\mathrm{T}} + \frac{2}{\hbar}Y \geq i \Omega_1$. Putting all together, we get the following optimization problem
\begin{align}\label{eq:optimization}
\begin{split}
&\text{maximize} \quad \vec v^{\mathrm{T}} Y \vec v\\
&\text{s.t.} \quad i X \Omega_1 X^{\mathrm{T}} + \frac{2}{\hbar}Y \geq i \Omega_1,\\
& \quad \Sigma \geq i\frac\hbar2  \begin{bmatrix}
X & 0\\
0 & \mathbb 1
\end{bmatrix}
\Omega_N
\begin{bmatrix}
X^{\mathrm{T}} & 0\\
0 & \mathbb 1
\end{bmatrix}
+
\begin{bmatrix}
Y & 0\\
0 & 0
\end{bmatrix},\\
&\quad Y\in\mathrm{Sym}(\mathbb R^{2\times 2}), X\in\mathbb R^{2\times 2}.
\end{split}
\end{align}
Note that the second constraint can be re-written as 
\begin{align}
\Sigma \geq i\frac\hbar2  \begin{bmatrix}
X\Omega_1 X^{\mathrm{T}} & 0\\
0 & \Omega_{n}
\end{bmatrix}
+
\begin{bmatrix}
Y & 0\\
0 & 0
\end{bmatrix}.
\end{align}
Note that $X\Omega_1 X^{\mathrm{T}}$ is the only quadratic constraint we have. Moreover, for any matrix $X\in\mathbb R^{2}$ we have $X\Omega_1 X^{\mathrm{T}} = \det(X) \Omega$. Therefore, one can let $r = \det(X)\in\mathbb R$ be a new free optimization variable. Moreover, note that $Y + i\frac\hbar{2}\Omega r$ for symmetric matrices $Y$ and real-valued $r$ sweeps over all $2\times2$ Hermitian matrices. Letting $Z = Y + i\frac\hbar{2}\Omega r$ in our optimization problem \eqref{eq:optimization} we can rewrite it as the following semi-definite program (SDP)
\begin{align}
\begin{split}
&\text{maximize} \quad \vec v^{\mathrm{T}} Z \vec v\\
&\text{s.t.} \quad Z \geq i\frac{\hbar}{2} \Omega_1,\\
& \quad \Sigma - 
\begin{bmatrix}
0 & 0\\
0 & -i\frac{\hbar}{2}\Omega_n
\end{bmatrix}
\geq  \begin{bmatrix}
Z & 0\\
0 & 0
\end{bmatrix},\\
&\quad Z\in\mathrm{Herm}(\mathbb C^{2\times 2}).
\end{split}
\end{align}

\subsection{Dual formulation of \texorpdfstring{\protect\cref{prop:sdp}}{}}\label{app:dual-sdp}

The dual formulation is a useful tool, as the equality of the primal and dual values proves that the value is reachable. Recall that our primal problem was phrased as
\begin{center}
\underline{Primal Problem}
\end{center}
\vspace{-.3 cm}
\begin{align}\label{eq:primal}
\begin{split}
&\text{maximize} \quad \tr(M Z)\\
&\text{s.t.} \quad Z \geq i\frac{\hbar}{2} \Omega_1,\\
& \quad \Sigma - 
\begin{bmatrix}
0 & 0\\
0 & -i\frac{\hbar}{2}\Omega_n
\end{bmatrix}
\geq  \begin{bmatrix}
Z & 0\\
0 & 0
\end{bmatrix},\\
&\quad Z\in\mathrm{Herm}(\mathbb C^{2\times 2}),
\end{split}
\end{align}
where we have put a generic matrix $M$ to capture different figures of merit one might be interested in. We list a few of these choices in \protect\cref{tab:fom-M}.

\begin{table}[t]
\centering
\renewcommand{\arraystretch}{1.5}
\begin{tabular}{l l}
\toprule
\textbf{Figure of Merit} & \textbf{Choice of $M$} \\
\midrule
Effective squeezing for lattice spacing $\vec v$ & $\frac{1}{\pi\hbar^2}\vec v \vec v^{\mathrm{T}}$ \\
Symmetric effective squeezing & $\frac{1}{\hbar}\mathbb{1}_2$ \\
Symmetric effective squeezing for arbitrary lattice & $\frac{1}{\hbar}S^{\mathrm{T}} S$ \\
\bottomrule
\end{tabular}
\caption{The choice of matrix $M$ in the primal optimization \eqref{eq:primal}. The first row corresponds to effective squeezing as discussed in the main text. The second row corresponds to the symmetric effective squeezing $\sigma_{\mathrm{sym}}$ (see Eq.~\eqref{eq:symmeric-eff-sqz}). The last row corresponds to the symmetric effective squeezing with respect to a lattice that has undergone a symplectic transform $S$. This is relevant when one is interested in the preparation of an arbitrary GKP state (as opposed to the sensor state).}
\label{tab:fom-M}
\end{table}

For this generic formulation, we obtain the dual problem below.
\begin{center}
\underline{Dual Problem}
\end{center}
\vspace{-.3 cm}
\begin{align}
\begin{split}
\mathrm{minimize} \quad &\mathrm{tr}\left( \left(\Sigma_\rho + i\frac{\hbar}{2}\Omega_{m+1}\right) Z \right)\\
\mathrm{s.t.} \quad & \begin{bmatrix}
\mathbb 1_{2} & 0 
\end{bmatrix}
Z
\begin{bmatrix}
\mathbb 1_2\\
0
\end{bmatrix}\geq M,\\
&Z\geq 0,\\
&Z\in\mathbb C^{2(m+1)\times 2(m+1)}.
\end{split}
\end{align}
Here, $\mathbb 1_2$ refers to the $2\times2$ identity matrix.

\subsection{Proof of \texorpdfstring{\protect\cref{prop:staircase-values}}{}}\label{app:pf-of-staircase}

We prove the following lemma which is a stronger statement and can be applied to non-Gaussian unitaries and states as well.
\begin{lemma}\label{lem:local-unitary-gives-same-channel}
Let $\rho_1$ and $\rho_2$ be bipartite quantum states over $PQ$. If $\rho_2 = (\mathbb 1 \otimes U) \rho_1 (\mathbb 1\otimes U^\dagger)$ for some unitary $U$ acting on $Q$, then a channel $\Phi$ acting on $P$ can be factored out of $\rho_1$ if and only if it can be factored out of $\rho_2$. 
\end{lemma}
\begin{proof}
Assume we can factor out $\Phi$ from $\rho_1$. Then, $\rho_2 = (\Phi \otimes \mathcal I)(\sigma)$ for some state $\sigma$. Therefore
\begin{align}
\begin{split}
\rho_2 &= (\mathbb 1 \otimes U) \rho_1 (\mathbb 1 \otimes U^\dagger)\\
&= (\mathbb 1 \otimes U)\left[  (\Phi\otimes \mathcal I)(\sigma) \right] (\mathbb 1\otimes U^\dagger)\\
&= (\Phi\otimes \mathcal I) \left[ (\mathbb 1 \otimes U) \sigma (\mathbb 1\otimes U^\dagger) \right],
\end{split}
\end{align}
which implies that $\Phi$ can also be factored out from the $P$ component of $\rho_2$. Using a similar logic (as $\rho_1 = (\mathbb 1\otimes U^\dagger)\rho_2(\mathbb 1 \otimes U)$) we get that $\Phi$ can be factored out of $\rho_1$ if it can be factored out of $\rho_2$.
\end{proof}
We now get back to the staircase situation and show that a staircase of width $M$ can be transformed into a staircase of width 2 (along with some product states) via unitaries on the last $M-1$ modes. We refer to \protect\cref{fig:staircase-reduction} where we show how the reduction works by reducing a four-mode staircase to a three-mode one. Below we present this proof in words.

Our proof is based on induction. Consider a staircase of width $M$. We apply the inverse of the last beam-splitter commuting the losses through the added beam-splitter. Since the beam-splitters cancel each other, we have removed the connection of the last mode to the first $M-1$ modes. Therefore, we reduced the staircase of width $M$ to one of width $M-1$. By doing so, we can go all the way to a two-level staircase.

We have therefore shown that the feasible set of our SDP for all staircases is the same.

\begin{figure}
\resizebox{\linewidth}{!}{%
\begin{quantikz}
\lstick{$\ket{\xi_1}$} & \arrow[d] & \gate{\mathcal{L}_\eta} &\\
\lstick{$\ket{\xi_2}$}  & & \arrow[d] & \gate{\mathcal{L}_\eta} &\\
\lstick{$\ket{\xi_{3}}$} & &  &\arrow[d] & \gate{\mathcal{L}_\eta} &  &\\
\lstick{$\ket{\xi_4}$} & & &  & \gate{\mathcal{L}_\eta}& \arrow[u, color=blue] &
\end{quantikz}
=
\begin{quantikz}
\lstick{$\ket{\xi_1}$} & \arrow[d] & \gate{\mathcal{L}_\eta} &\\
\lstick{$\ket{\xi_2}$}  & & \arrow[d] & \gate{\mathcal{L}_\eta} &\\
\lstick{$\ket{\xi_{3}}$} & &  &\arrow[d] & &\gate{\mathcal{L}_\eta}  &\\
\lstick{$\ket{\xi_4}$} & & &  &\arrow[u, color=blue]&\gate{\mathcal{L}_\eta}&  &
\end{quantikz}
=
\begin{quantikz}
\lstick{$\ket{\xi_1}$} & \arrow[d] & \gate{\mathcal{L}_\eta} &\\
\lstick{$\ket{\xi_2}$}  & & \arrow[d] & \gate{\mathcal{L}_\eta} &\\
\lstick{$\ket{\xi_{3}}$} &  & &\gate{\mathcal{L}_\eta}  &\\
\lstick{$\ket{\xi_4}$} &\gate{\mathcal{L}_\eta}& 
\end{quantikz}
}
\caption{Illustration for reduction of a four-mode staircase to a three-mode staircase in the proof of \protect\cref{prop:staircase-values}. The blue beam-splitter added represents the unitary $U$ in \protect\cref{lem:local-unitary-gives-same-channel} that reduces the four-mode staircase to a three-mode staircase. The first equality simply commutes the losses through the beam-splitter (this can be done so long as the two losses are the same). Lastly, the beam-splitters between the bottom modes cancel each other, and we end up with a three-mode staircase. Using the same logic, one can inductively reduce a staircase of width $M\geq 2$ to that of width $2$.}
\label{fig:staircase-reduction}
\end{figure}
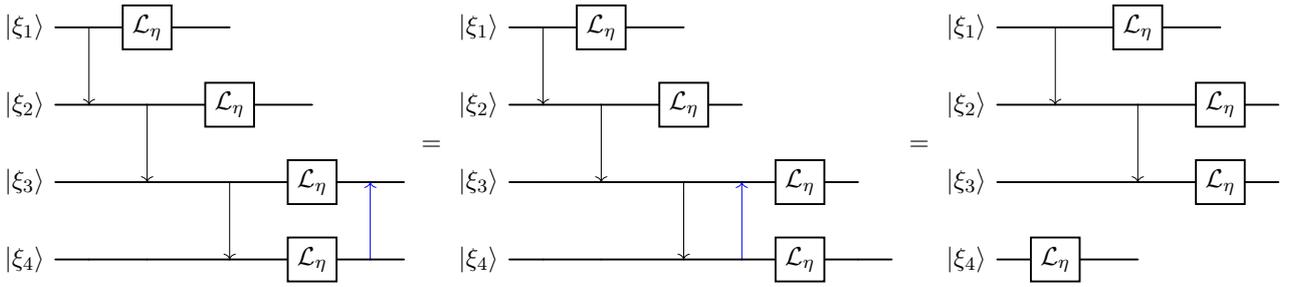

\subsection{Extension of \texorpdfstring{\protect\cref{prop:sdp}}{} to multimode GKP states}\label{app:sdp-extension}

Here we show that one can write an SDP to characterize the set of multimode Gaussian channels that can be factored out of a set of given modes from a given mixed Gaussian state. Using this technique, we can bound the stabilizer expectation values of any candidate multimode GKP state. We summarize this result in the following proposition. Note that we are applying the XPXP convention in this section (similar to \protect\cref{sec:GKP-bound}).

\begin{proposition}
Let $\rho$ be a state over $(m+n)$ modes. For any Gaussian density matrix $\rho$ with covariance matrix $\Sigma$, the following SDP
\begin{align}
\begin{split}
&\mathrm{maximize} \quad \frac{1}{\pi \hbar^2}\vec v^{\mathrm{T}} Z \vec v\\
&\mathrm{s.t.} \quad Z \geq i\frac{\hbar}{2} \Omega_m,\\
& \quad \Sigma - 
\begin{bmatrix}
0 & 0\\
0 & -i\frac{\hbar}{2}\Omega_n
\end{bmatrix}
\geq  \begin{bmatrix}
Z & 0\\
0 & 0
\end{bmatrix},\\
&\quad Z\in\mathrm{Herm}(\mathbb C^{2m\times 2m}),
\end{split}
\end{align}
computes an upper bound on the largest achievable value of 
\begin{align}
|\tr(\tilde \rho\, D_{\vec v})|
\end{align}
where $\tilde\rho$ is any $m$-mode state that one might achieve by performing any (possibly non-Gaussian) post-selection on the last $n$ modes of $\rho$ that may succeed with any non-zero probability.
\end{proposition}

\begin{proof}
To prove this, we recall that if $\rho = (\Phi \otimes \mathcal I)[\sigma]$ for some state $\sigma$ and a Gaussian channel $\Phi$ with $(X,Y)$ parametrization acting on the first $m$-modes, then
\begin{align}
|\tr(\tilde\rho D_{\vec v})| \leq \exp(-\frac{1}{2\hbar^2} \vec v^{\mathrm{T}} Y \vec v),
\end{align}
where $\vec v = (x_1,p_1,\cdots, x_m, p_m)\in\mathbb R^{2m}$, and $D_{\vec v} = \bigotimes_{j=1}^m \exp(i (p_j Q_j - x_j P_j))$ is the multi-mode displacement by $\vec v$ in the phase-space. Therefore, we maximize $\vec v^{\mathrm{T}} Y \vec v$ where $Y$ corresponds to the $Y$ matrix of a Gaussian channel that can be factored out from the state. 
Following the steps similar to \protect\cref{app:pf-of-sdp}, we get the following optimization problem
\begin{align}\label{eq:opt-many-mode}
\begin{split}
&\text{maximize} \quad \vec v^{\mathrm{T}} Y \vec v\\
&\text{s.t.} \quad i X \Omega_m X^{\mathrm{T}} + \frac{2}{\hbar}Y \geq i \Omega_m,\\
& \quad \Sigma \geq i\frac\hbar2  \begin{bmatrix}
X & 0\\
0 & \mathbb 1
\end{bmatrix}
\Omega_{m+n}
\begin{bmatrix}
X^{\mathrm{T}} & 0\\
0 & \mathbb 1
\end{bmatrix}
+
\begin{bmatrix}
Y & 0\\
0 & 0
\end{bmatrix},\\
&\quad Y\in\mathrm{Sym}(\mathbb R^{2m\times 2m}), X\in\mathbb R^{2m\times 2m}.
\end{split}
\end{align}
Note that $X\Omega_m X^{\mathrm{T}}$ is an anti-symmetric matrix. Indeed, we can use the following lemma, to replace $X\Omega_m X^{\mathrm{T}}$ with any generic anti-symmetric matrix, say $A$.

\begin{lemma}[\cite{greub2012linear, haber_pfaffian}]
For any $2m\times 2m$ real anti-symmetric matrix, say $A$, there exists a real matrix $X\in\mathbb R^{2m \times 2m}$ such that
\begin{align}
A = X \Omega_m X^{\mathrm{T}}.
\end{align}
\end{lemma}

As a result, we rewrite our optimization problem \eqref{eq:opt-many-mode} as
\begin{align}
\begin{split}
&\text{maximize} \quad \vec v^{\mathrm{T}} Y \vec v\\
&\text{s.t.} \quad iA+ \frac{2}{\hbar}Y \geq i \Omega_m,\\
& \quad \Sigma \geq 
\begin{bmatrix}
Y+i\frac\hbar2A & 0\\
0 & i\frac{\hbar}{2}\Omega_{n}
\end{bmatrix}
\\
&\quad Y\in\mathrm{Sym}(\mathbb R^{m\times m}), A\in\mathrm{Asym}(\mathbb R^{2m\times 2m}),
\end{split}
\end{align}
where $\mathrm{Asym}(\mathbb R^{2m\times 2m})$ represents the set of anti-symmetric real matrices of size $2m\times 2m$. Letting $Z:=Y + i\frac{\hbar}{2}A$, and noting conditions on $Y$ and $A$ imply that $Z$ is Hermitian, we obtain the formulation in the statement of the proposition.
\end{proof}

\section{Bargmann parametrization of common Gaussian elements}\label{sec:triples}
In this section we summarize the Abc parametrization of common quantum optical states and transformations.
When used in conjunction with the inner product formulas, introduced in \protect\cref{app:inner-product-in-bargmann}, they can describe a wide variety of circuits and Gaussian operations. At any point one can use the Abc triple of a Gaussian object to parametrize the recurrence relation in Eq.~\eqref{eq:recrel} and compute the Fock basis representation.

\subsection{States}
In \protect\cref{tab:states}, we have listed the Bargmann parameterization of a few common states in the literature of quantum optics. We adopt the notation
\begin{align}
X = \begin{bmatrix}
0 & 1\\
1 & 0
\end{bmatrix},
\quad 0_{2} = 
\begin{bmatrix}
0\\
0
\end{bmatrix}.
\end{align}

Moreover, the thermal state is defined as
\begin{align}
    \rho_{\mathrm{th}}(\bar{n}) = (1-e^{-\beta}) \cdot e^{-\beta \hat{N}},
\end{align}
with
\begin{align}
\bar{n} = \frac{e^{-\beta}}{1-e^{-\beta}},
\end{align}
being the average particle number of the thermal state.

\begin{table}[h!]
\centering
\small
\renewcommand{\arraystretch}{1.4}
\begin{tabular}{l l c c l}
\toprule
\textbf{Name} & \textbf{Description} & \textbf{$A$} & \textbf{$b$} & \textbf{$c$} \\
\midrule
{\small Bargmann Eigenstate} & $e^{\frac{1}{2}|\alpha|^2} \ket{\alpha}$ 
&  $0$ 
& $\alpha$ 
& $1$ \\

{\small Coherent State} & $\ket{\alpha}$ 
& $0$ 
& $\alpha$ 
& $e^{-\frac{1}{2}|\alpha|^2}$ \\

{\small Squeezed Vacuum State} & $S(r, \phi)\ket{0}$ 
& $-e^{i\phi} \tanh r$ 
& $0$ 
& {\small $\sqrt{\mathrm{sech}(r)}$} \\

{\small Displaced Squeezed State} & $D(\alpha)S(r, \phi)\ket{0}$ 
& $ -e^{i\phi} \tanh r $ 
& $\alpha + \alpha^* e^{i\phi} \tanh r$ 
& 
$\frac{e^{-\frac{1}{2}|\alpha|^2 - \frac{1}{2} {\alpha^*}^2 e^{i\phi} \tanh r }}{\sqrt{\cosh r}} $%
\\
{\small Two-mode Squeezed Vacuum State} & $S_2(r,\phi)\ket{00}$ & $e^{i\phi}\tanh{r} X$ & $0_{2}$ & {\small$ \mathrm{sech}(r)$}\\
{\small Quadrature Eigenstate} & $\ket{x}_\phi$ & $-e^{2i\phi}$ & $\sqrt{\frac{2}{\hbar}} xe^{i\phi}$ & $\frac{e^{-x^2/(2\hbar)}}{(\pi\hbar)^{1/4}}$ \\
{\small Thermal State} & $\rho_{\mathrm{th}}(\bar{n})$ & $\frac{\bar{n}}{\bar{n}+1}X$ & $0_{2}$ & $(\bar{n}+1)^{-1}$\\
\bottomrule
\end{tabular}
\caption{Quadratic-exponential form parameters $(A, b, c)$ for various quantum states in the Bargmann representation.}
\label{tab:states}
\end{table}

\subsection{Transformations}
In this section, we list the Bargmann parametrization of some unitaries. All the triples in this section are in output-input order (which for unitaries coincides with the type-wise order). \protect\cref{tab:unitaries} shows the parametrization of some of the single-mode unitary gates. Note that we have employed the convention
\begin{align}
S(\xi) = \exp\left(\frac12(\xi^\ast a^2 - \xi a^\dagger {}^2)\right),
\end{align}
for the squeezing gate, where $\xi = r e^{i\theta}$ represents the complex squeezing parameter.

\begin{table}[h!]
\centering
\small
\renewcommand{\arraystretch}{1.5}
\begin{tabular}{l l c c c}
\toprule
\textbf{Name} & \textbf{Description} & \textbf{$A$} & \textbf{$b$} & \textbf{$c$} \\
\midrule
Identity 
& {\large$\mathbb{1}$} 
& $X$ 
& $0_{2}$ 
& $1$ \\

Rotation 
& {\large$e^{i\theta \hat{N}}$} 
& $e^{i\theta} X$ 
& $0_{2}$ 
& $1$ \\

Displacement 
& {\large$e^{\alpha a^\dagger - \alpha^* a}$ }
& $X$ 
& $\begin{bmatrix} \alpha \\ -\alpha^* \end{bmatrix}$ 
& $e^{-\frac{1}{2}|\alpha|^2}$ \\

Squeezing 
& {\large$e^{\frac{1}{2}(\xi^* a^2 - \xi a^{\dagger 2})}$} 
& $\begin{bmatrix} -e^{i\phi} \tanh r & \text{sech}\,r \\ \text{sech}\,r & e^{-i\phi} \tanh r \end{bmatrix}$ 
& $0_{2}$ 
& $\frac{1}{\sqrt{\cosh r}}$ \\
\bottomrule
\end{tabular}
\caption{Quadratic-exponential form parameters $(A, b, c)$ for various single-mode Gaussian unitaries.}
\label{tab:unitaries}
\end{table}

Below, we list the parameters of some multi-mode Gaussian unitaries and some single-mode channels:
\begin{itemize}
\item For a beamsplitter gate defined as
\begin{align}
B(\theta,\phi) = e^{\theta(e^{i\phi} a_1a_2^\dagger - e^{-i\phi} a_1^\dagger a_2)},
\end{align}
we have the following Abc parametrization
\begin{align}
A &= \begin{bmatrix}
0 & 0 & \cos\theta & -e^{-i\phi}\sin\theta \\
0 & 0 & e^{i\phi}\sin\theta & \cos\theta \\
\cos\theta & e^{i\phi}\sin\theta & 0 & 0 \\
-e^{-i\phi}\sin\theta & \cos\theta & 0 & 0
\end{bmatrix},
\quad
b = 0_{4},
\quad
c = 1.
\end{align}
\item For an $N$-mode interferometer defined by $U\in\mathrm{SU}(N)$ determining its action on the single-photon subspace, we have the following Abc parametrization
\begin{align}
A &= \begin{bmatrix} 0_{N\times N} & U \\ U^{\mathrm{T}} & 0_{N\times N} \end{bmatrix},
\quad
b = 0_{2N},
\quad
c = 1.
\end{align}
\item For an $N$-mode real interferometer (which does not mix position and momentum coordinates in phase space), defined by $V\in\mathrm{SO}(N)$ determining its action on the single-photon subspace, we have the following Abc parametrization
\begin{align}
A &= \begin{bmatrix} 0_{N\times N} & V \\ V^{\mathrm{T}} & 0_{N\times N} \end{bmatrix},
\quad
b = 0_{2N},
\quad
c = 1.
\end{align}
\item For a two-mode squeezing gate, defined as
\begin{align}
S_2(\xi) = \exp\left(\frac12(\xi^\ast a_1a_2-\xi a_1^\dagger a_2^\dagger) \right),
\end{align}
with $\xi = re^{i\phi}$, we have that
\begin{align}
A &= \begin{bmatrix}
0 & e^{i\phi}\tanh r & \text{sech } r & 0 \\
e^{i\phi}\tanh r & 0 & 0 & \text{sech } r \\
\text{sech } r & 0 & 0 & e^{-i\phi}\tanh r \\
0 & \text{sech } r & e^{-i\phi}\tanh r & 0
\end{bmatrix}, \quad
b = 0_{4}, \quad
c = \frac{1}{\cosh r}.
\end{align}
\item For a photon loss channel with transmissivity $\eta$, we have
\begin{align}
A &= \begin{bmatrix}
0 & 0 & \sqrt{\eta} & 0 \\
0 & 0 & 0 & \sqrt{\eta} \\
\sqrt{\eta} & 0 & 0 & 1-\eta \\
0 & \sqrt{\eta} & 1-\eta & 0
\end{bmatrix}, \quad
b = 0_{4}, \quad
c = 1.
\end{align}
We note that one can define ``Gaussian core channels'' as follows: a Gaussian channel $\Phi$ is a \textit{Gaussian core channel}, if it maps any input state with a finite Fock cutoff to an output state with a finite Fock cutoff. Concretely, if $\Phi[\ket{n}\bra{m}]$ has a finite Fock cutoff for any $n,m\in\mathbb N$. The loss channel is then a natural example of a core Gaussian channel.

\item For an amplification channel with gain $g$, we have
\begin{align}
A = \begin{bmatrix}
0 & 1-1/g & 1/\sqrt{g} & 0 \\
1-1/g & 0 & 0 & 1/\sqrt{g}\\
1/\sqrt{g} & 0 & 0 & 0 \\
0 & 1/\sqrt{g} & 0 & 0
\end{bmatrix},\quad
b = 0_{4}, \quad
c = 1/g.
\end{align}
\item For the Fock damping operator $e^{-\beta \hat N}$, we have
\begin{align}\label{eq:fd}
A = \begin{bmatrix} 0 & e^{-\beta} \\ e^{-\beta} & 0 \end{bmatrix}, \quad
b = 0_{2}, \quad
c = 1.
\end{align}
\item We highlight that one can always write a set of Gaussian Kraus operators for any Gaussian channel. For instance, we have that the action of the loss channel with transmissivity $\eta$, denoted by $\mathcal L_\eta$, can be written as
\begin{align}
\mathcal L_\eta[\bullet] = \int_{z\in\mathbb C} K(z) \, \bullet \, K(z)^\dagger\, \mathrm d\mu(z),
\end{align}
where $K(z)$ is an operator. Note that $K$ can be parametrized in the Bargmann space as 
\begin{align}
F_K(z,w,v) = \bra{w^\ast} K(z)\ket{v}\, e^{\frac12(|w|^2+|z|^2+|v|^2)}. 
\end{align}
Since we know that a loss channel is equivalent to inputting an ancilla vacuum to a beamsplitter and tracing out the ancilla, we can readily obtain the Abc parametrization of $F_K$ by removing the third row and column (i.e.~input on first mode) of the beamsplitter with $\phi=0$ and $\theta = \arcsin(\sqrt{1-\eta})$ (see \protect\cref{app:inner-product-in-bargmann}), which yields
\begin{align}\label{eq:loss-kraus}
A = \begin{bmatrix}
0 & 0 & -\sqrt{1-\eta} \\
0 & 0 & \sqrt{\eta} \\
-\sqrt{1-\eta} & \sqrt{\eta} & 0
\end{bmatrix}, \quad
b = 0_{3}, \quad
c = 1.
\end{align}
This parametrizes a continuous Kraus operator with variable order $(z, w, v)$. Lastly, note how it is straightforward to remove the first row and column (i.e.~output on first mode) and show that the remaining $2\times 2$ bottom-right block is the same as \eqref{eq:fd}, i.e.
\begin{align}
K(0) =  \sqrt{\eta}^{\hat{N}} = e^{-\beta \hat{N}},
\end{align}
as the Fock damping operator can be written as a beamsplitter with a vacuum input and a vacuum-postselected output.
\end{itemize}


\subsection{Maps between representations}

For convenience, we summarize the maps between common continuous-variable representations and the Bargmann representation. These maps are explained in detail in \protect\cref{app:representations}. The following representations are written in the output-input ordering.

\begin{itemize}
\item The mapping kernel between Bargmann and quadrature representations of a pure state is parametrized by the following triple:
\begin{align}
A_\phi=\begin{bmatrix}
-\frac{\mathbb{1}}{\hbar} & e^{-i\phi}\sqrt{\frac{2}{\hbar}}\mathbb{1}\\
e^{-i\phi}\sqrt{\frac{2}{\hbar}}\mathbb{1} & -e^{-2i\phi}\mathbb{1}
\end{bmatrix},\quad
b_\phi = 0_{2},\quad
c_\phi = \frac{1}{(\pi\hbar)^{n/4}}.
\end{align}
\item The mapping kernel between Bargmann and $s$-parametrized \emph{phase space} functions (Stratonovich-Weyl kernel) is parametrized by the following triple:
\begin{align}
A_{\Delta_{s}} = \frac{2}{s-1}\begin{bmatrix}
    X & -\mathbb 1\\
    -\mathbb 1 & \frac{s+1}{2}X
\end{bmatrix},\quad
b_{\Delta_{s}}=0_{4},\quad
c_{\Delta_{s}}=\frac{2}{\pi^n\abs{s-1}^n}.
\end{align}
\item The mapping kernel between Bargmann and $s$-parametrized \emph{characteristic} functions (Fourier transform of Stratonovich-Weyl kernel) is parametrized by the following triple:
\begin{align}
A_{T_{s}} = \begin{bmatrix}
\frac{s-1}{2}X&\Omega^T\\
\Omega&X
\end{bmatrix},\quad
b_{T_{s}} = 0_{4},\quad
c_{T_{s}} = 1.
\end{align}
\end{itemize}

\end{document}